%% file: main.tex
\documentclass[onecolumn,11pt]{article}
\pdfoutput=1
\usepackage[top=1in, bottom=1in, left=1in, right=1in]{geometry}
\input{Header} 
\usepackage{mathtools}
\usepackage{setspace}
\usepackage{palatino} 
\usepackage[numbers,sort&compress]{natbib}
\usepackage{bm}
\usepackage[utf8]{inputenc} 
\usepackage[T1]{fontenc}    
\usepackage{hyperref}
\usepackage{url}            
\usepackage{booktabs}       
\usepackage{amsfonts}       
\usepackage{nicefrac}       
\usepackage[nopatch=footnote]{microtype}      
\usepackage{lipsum}
\usepackage{fancyhdr}       
\usepackage{graphicx}       
\usepackage{xcolor}
\usepackage{dsfont}
\usepackage{bm}
\hypersetup{
  colorlinks = true, 
  urlcolor   = blue, 
  linkcolor  = blue, 
  citecolor  = red 
}

\title{{\fontsize{17}{17}\selectfont \textbf{Separation of periodic orbits\\ in the delay embedded space of chaotic attractors}}\vspace{-.15in}}

\author{
  \normalsize{Prerna M Patil$^{1*}$, Eurika Kaiser$^{1}$, J. Nathan Kutz$^{2}$, Steven L. Brunton$^{1*}$} \\
  \footnotesize{$^1$ Department of Mechanical Engineering, University of Washington, Seattle, WA 98195, United States}\\
  \footnotesize{$^2$ Department of Applied Mathematics, University of Washington, Seattle, WA 98195, United States\vspace{-.2in}}
}
  
\date{}
\begin{document}
\maketitle

\blfootnote{$^*$ Corresponding authors (prernap@uw.edu,sbrunton@uw.edu).}

\begin{abstract}
This work explores the intersection of time-delay embeddings, periodic orbit theory, and symbolic dynamics. 
Time-delay embeddings have been effectively applied to chaotic time series data, offering a principled method to reconstruct relevant information of the full attractor from partial time series observations. In this study, we investigate the structure of the unstable periodic orbits of an attractor using time-delay embeddings. 
First, we embed time-series data from a periodic orbit into a higher-dimensional space through the construction of a Hankel matrix, formed by arranging time-shifted copies of the data.  We then examine the influence of the width and height of the Hankel matrix on the geometry of unstable periodic orbits in the delay-embedded space. The right singular vectors of the Hankel matrix provide a basis for embedding the periodic orbits.
We observe that increasing the length of the delay (e.g., the height of the Hankel matrix) leads to a clear separation of the periodic orbits into distinct clusters within the embedded space. Our analysis characterizes these separated clusters and provides a mathematical framework to determine the relative position of individual unstable periodic orbits in the embedded space. Additionally, we present a modified formula to derive the symbolic representation of distinct periodic orbits for a specified sequence length, extending the Polyá-Redfield enumeration theorem.
\end{abstract}

\input{Introduction}
\input{Background}
\input{MethodOfDelays}
\section{Long time delay embeddings of the periodic orbits}\label{sec:Theorems}
\input{LongTimeDelayEmbeddings}
\section{Results}\label{sec:results}
In this section, we show the results for the unstable periodic orbits in the delay-embedded space. The separation of the UPOs is observed as the height of the Hankel matrix increases. We demonstrate the separation of the UPOs for the three datasets of the Lorenz attractor and two datasets of the Rössler attractor.
\input{LorenzAttractorv2} 
\input{RosslerAttractor} 

\section{Conclusions}
We analyzed the periodic orbits of the Lorenz and Rössler systems using long time-delay embeddings. For the Lorenz system, we observed that long time-delay embeddings of the Hankel matrix cause the UPOs to separate in the embedded space. The UPOs also cluster based on the ratio $\frac{\rho-1}{\rho+1}$ where $\rho$ is the ratio of the $A$ symbols to the $B$ symbols in the symbolic designation of the UPO (obtained from Theorem(\ref{th:SeparationOfUPOs})). We examine the clustering of symmetric UPOs around the center of the unfolded chaotic trajectory whereas the $A$-heavy and $B$-heavy orbits cluster on diametrically opposite ends of the central cluster. UPOs with the same ratio cluster together irrespective of the order of the symbols in the symbolic notation.  

In case of the Rössler attractor, the separation of the UPOs follows the same principal as the Lorenz attractor. For the Rössler attractor this analysis is not dependent on the symbolic designation of the UPOs. The plane of separation for the UPOs for the symbolic designation is not the $y=0$ unlike the Lorenz case. For the Rössler the plane of separation is obtained from the first return map of the Poincaré section at $x=0$. Also the time spent by the UPO in each of the symbolic lobes ($A$ and $B$) are not co-related to the number of symbols. The analysis has been presented and the UPOs have been rearranged and the ratio $\rho$ has been recalculated depending on the amount of time spent by the UPO in $y>0$ and $y<0$. It is observed that for the separation of UPOs in the embedded space, this recalculated ratio plays a significant role than the plane of separation of the symbolic dynamics. 

We further present the modified Redfield-Polyá enumeration theorem which provides the number of UPOs for a given sequence length ($n$) and number of symbols ($k$). This theorem is obtained by imposing additional constraints on the Redfield-Polyá enumeration theorem based on the properties of the attractor for which the periodic orbits are computed.

We anticipate the use of this analysis in Hankel based methods like DMDc, HAVOK and Koopman operators. This analysis can be applied for the study of control in orbits of chaotic attractor or in turbulent flows. This will also be relevant in the study of stabilization of  periodic orbits. 

This study opens up numerous avenues for future research. Extending this work to complex PDEs like the Navier-Stokes equations or the Kuramoto-Shivashinsky equation would be an important extension of this research. The separation of the UPOs in the embedded space for long time delay embeddings and the clustering of the orbits can be used to study the shadowing behavior of the chaotic trajectory and obtain metrics to determine which UPOs are visited more often~\cite{yalniz2021coarse,page2024recurrent,maiocchi2022decomposing}. Understanding the transition behavior between UPOs can be used to construct reduced order models of the chaotic dynamical system or turbulent flows. This study can also be leveraged for control of unstable periodic orbits. It will be interesting to combine the control of UPOs using the time delay embeddings. 
\section*{Acknowledgments}
The authors acknowledge support from the National Science Foundation AI Institute in Dynamic Systems (Grant number 2112085) and The Boeing Company. This work has benefited from fruitful discussions with Samuel Otto, Zachary Nicolaou, Doris Voina and Hangjun Cho. 

\section*{Data accessibility}
The data and codes used for this analysis can be found at \cite{myCode}.

 \begin{spacing}{.9}
 \setlength{\bibsep}{6.25pt}
 \bibliographystyle{unsrt}
 \bibliography{References}
 \end{spacing}
\end{document}

%% file: Header.tex
\usepackage{amsmath,amsfonts,amscd,amssymb,amsthm}
\usepackage{graphicx}
\usepackage{epstopdf}
\usepackage{overpic}
\usepackage{cancel}
\usepackage{rotating}
\usepackage{caption,subcaption,mwe,float} 
\usepackage{color}
\usepackage{multirow}
\usepackage{wrapfig}
\usepackage{multicol}
\usepackage{multirow}
\usepackage{adjustbox}
\usepackage{setspace}
\usepackage{comment}
\usepackage{algorithm,algcompatible}
\usepackage{algpseudocode}
\usepackage{tikz}
\usetikzlibrary[arrows.meta]
\usepackage{booktabs}

\usepackage{mathtools}
\usepackage{psfrag}
\usepackage{tikz}
\usepackage[T1]{fontenc}
\usepackage{array}
\usepackage{makecell}
\newcolumntype{x}[1]{>{\centering\arraybackslash}p{#1}}

\usepackage[bottom,flushmargin,hang,multiple]{footmisc}
\usepackage{lipsum}
\newcommand\blfootnote[1]{%
  \begingroup
  \renewcommand\thefootnote{}\footnote{#1}%
  \addtocounter{footnote}{-1}%
  \endgroup
}

\algnewcommand\algorithmicinput{\textbf{Input:}}
\algnewcommand\INPUT{\item[\algorithmicinput]}
\algnewcommand\algorithmicoutput{\textbf{Output:}}
\algnewcommand\OUTPUT{\item[\algorithmicoutput]}
\algnewcommand\algorithmicoptional{\textbf{Optional:}}
\algnewcommand\OPTIONAL{\item[\algorithmicoptional]}

\newcommand{\bp}{\mathbf{p}}

\newcommand{\bv}{\mathbf{v}}

\newcommand{\bx}{\mathbf{x}}

\newcommand\norm[1]{\left\lVert#1\right\rVert}

\definecolor{blue}{rgb}{0,0,1}
\definecolor{darkgreen}{rgb}{0,0.5,0}
\definecolor{red}{rgb}{1,0,0}
\definecolor{teal}{rgb}{0,0.5,0.7}

\newtheorem{theorem}{Theorem}
\newtheorem{lemma}{Lemma}

\graphicspath{{Figures/}}
\newcommand{\boundellipse}[3]
{(#1) ellipse (#2 and #3)
}

%% file: Introduction.tex
\section{Introduction}
Unstable periodic orbits (UPOs) have been essential in our understanding of chaotic dynamics~\cite{cvitanovic1991periodic,cvitanovic2005chaos,gutzwiller2013chaos} of real world systems, such as fluids~\cite{kawahara2001periodic,cvitanovic2010geometry,viswanath2007recurrent,chandler2013invariant} and planetary dynamics \cite{broucke1968periodic,koon2000dynamical, arnas2021approximate,koon2000heteroclinic}. Similarly, time-delay embeddings provide a rigorous mathematical approach to understanding the full-state dynamics of systems that are only partially observed~\cite{takens1981detecting, sauer1991embedology}. Recently, delay embeddings have been increasingly used to build data-driven models from partial measurements, even for nonlinear systems~\cite{brunton2017chaos}. 
In both fields, the past decade has seen several breakthroughs based on advances in machine learning~\cite{lusch2018deep,champion2019data,otto2019linearly,gilpin2020deep,linot2020deep} and Koopman operator theory~\cite{brunton2017chaos,das2019delay,kamb2020time,arbabi2017ergodic}.
This work begins to examine UPOs through the lens of time-delay embeddings. In particular, we observe a separation of UPOs in an appropriate delay embedding and provide theoretical results to understand this phenomenon. 

Time-delay embeddings enable the reconstruction of a high-dimensional state-space from partial measurements of a system with latent variables~\cite{sauer1991embedology,gidea2007phase}.  The Whitney embedding theorem and the Takens embedding theorem~\cite{whitney1936differentiable,takens1981detecting,broomhead1986extracting,casdagli1991state,liebert1989proper,gibson1992analytic} provide conditions under which this reconstruction is possible. 
Time-delay embedding has found applications in linear system identification using the eigensystem realization algorithm (ERA)~\cite{juang1985eigensystem} and singular spectrum analysis (SSA)~\cite{broomhead1986extracting,broomhead1989time}. These methods identify the dominant modal content of data in delay coordinates~\cite{kamb2020time} from the singular value decomposition (SVD) of the Hankel matrix.   
This theory has recently been extended to nonlinear systems based on Koopman operator theory~\cite{giannakis2012nonlinear,brunton2017chaos}. 
Machine learning has further enabled the discovery and representation of the diffeomorphism alluded to by Takens~\cite{bakarji2023discovering}. 

In a parallel line of investigation, classical periodic orbit theory has been used to calculate the fractal dimension, the topological entropy, and the Lyapunov exponents of strange attractors~\cite{auerbach1987exploring,viswanath2001lindstedt,pawelzik1991unstable,so1997extracting,badii1994progress,lan2004variational,choe1999computing,broucke1968periodic}. Symbolic dynamics~\cite{williams2004introduction,hao1991symbolic,galias2009symbolic,fang1996symbolic,hao1998symbolic,viswanath2003symbolic,davidchack2000estimating} has been closely related to unstable periodic orbits as a tool for qualitative analysis of a dynamical system. It provides a concise way to depict these periodic orbits, as the numerical details are disregarded, but a high-level view of the dynamics, such as symmetry, periodicity is retained. Fig.~(\ref{fig:SymbDomainsLorenz}) illustrates a few UPOs for the chaotic Lorenz system along with the symbolic partitions of the attractor. The study of periodic orbits has been relevant in the context of celestial mechanics to solve the restricted three-body problem (RTBP)~\cite{broucke1968periodic,koon2000dynamical, arnas2021approximate,koon2000heteroclinic} and in the field of fluid dynamics, these periodic orbits are referred to as the \lq exact coherent structures\rq ~(ECS)~\cite{graham2021exact} or \lq exact recurrent flows\rq~\cite{chandler2013invariant,cvitanovic2013recurrent}. ECS are considered to be saddle points in state space around which chaotic dynamics are organized. Due to the high dimensionality of fluid flows, the computations of these structures are nontrivial. The development of numerical solvers~\cite{viswanath2007recurrent,viswanath2001lindstedt,guckenheimer2000computing,choe1999computing,doedel1998auto97} has enhanced our ability to find these UPOs in high-dimensional fluid flow systems. Unstable periodic solutions of the incompressible Navier-Stokes equations have been found for wall-bounded turbulent flows~\cite{kawahara2001periodic,chandler2013invariant,waleffe2001exact,waleffe1998three,gibson2008visualizing,graham2021exact} and have been instrumental in understanding the turbulence phenomenon. 
It has also been shown that ECS solutions can be represented by a few resolvent modes~\cite{sharma2016correspondence, rosenberg2019computing} i.e., these modes form an efficient basis for the representation of ECS. The connections between resolvent analysis, ECS and Koopman operator have been shown in~\cite{sharma2016low}. The developments in data-driven analysis and Koopman operator theory have led to the identification of UPOs for turbulent data~\cite{page2020searching} as well as in other dynamical systems~\cite{bramburger2021data}. Deep learning techniques have helped in using the identified UPOs to construct low-dimensional representation of turbulent flow and predicting the statistics of the turbulent attractor~\cite{page2024recurrent}.

\begin{figure}
    \centering
    \subfloat[Lorenz attractor]{\includegraphics[width=0.495\textwidth]{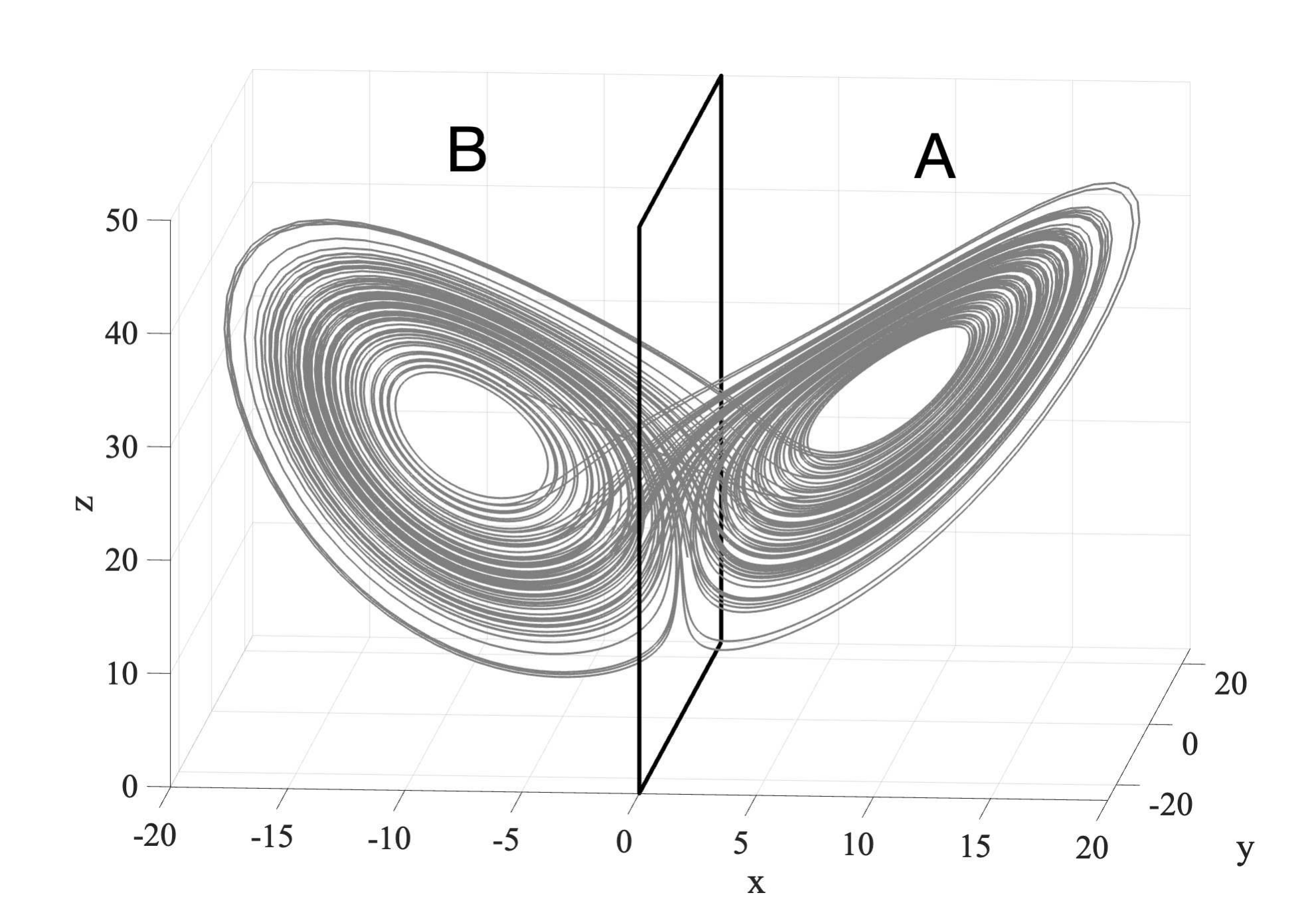}}
    \subfloat[Representations of UPOs]{\includegraphics[trim=12cm 0 12cm 0,clip,width=0.495\textwidth]{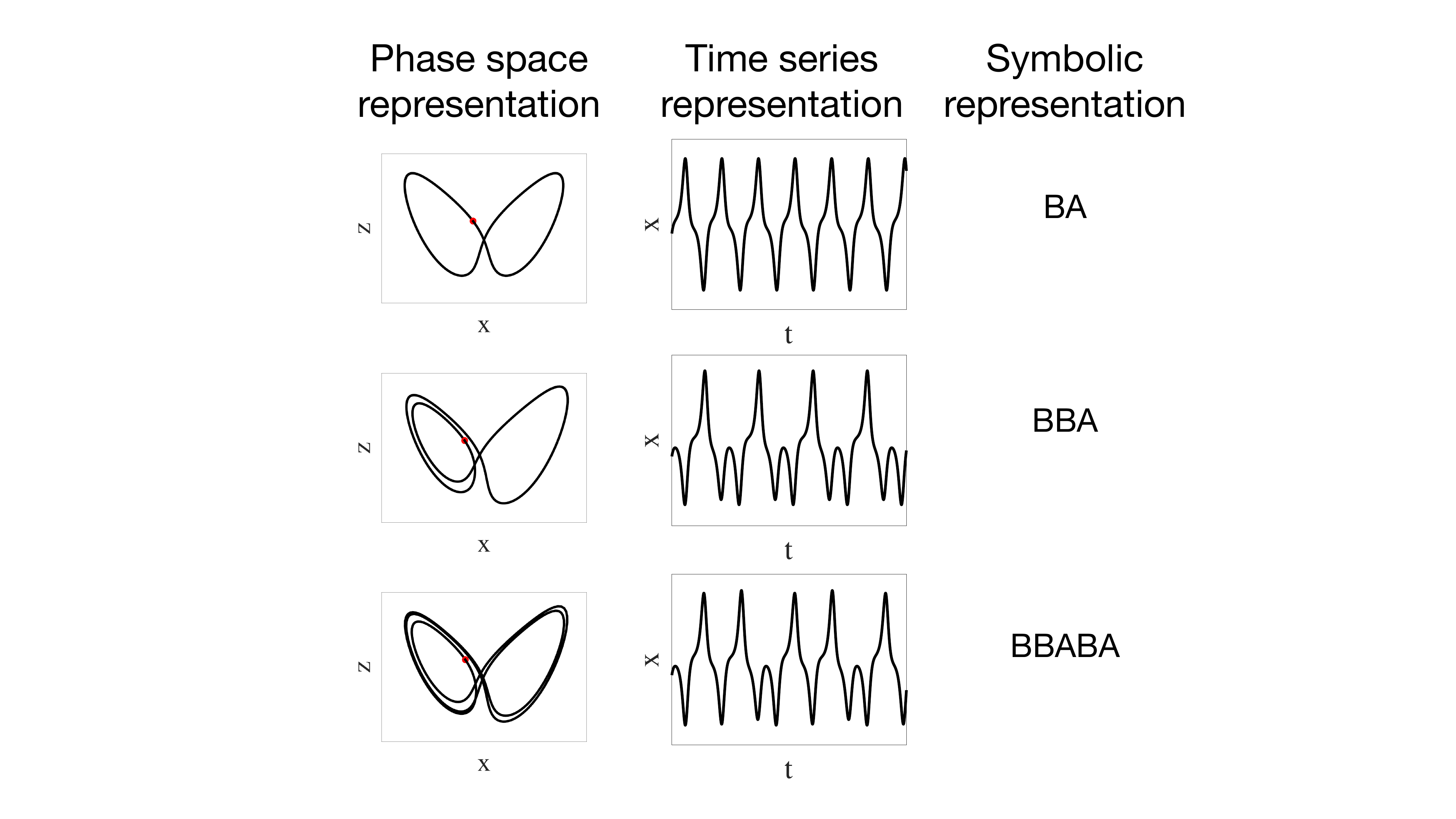}}
    \caption{(a) The two lobes of the Lorenz attractor are denoted by the symbolic dynamics notation. When the trajectory passes through the left lobe i.e., $x<0$ it is denoted as $B$ and when the trajectory passes through the right lobe i.e., $x\geq 0$ it is denoted as $A$. (b) Three representations of the unstable periodic orbits: phase space representation, time series representation and symbolic representation. We will investigate these orbits using time delays. }
    \label{fig:SymbDomainsLorenz}
\end{figure}
This work begins to explore the structure of unstable periodic orbits in delay-embedded spaces with long time delays.  We investigate this structure in the Lorenz and Rössler systems.  
The analysis presented in this paper is exploratory in nature and will improve the understanding of chaotic attractors by analyzing the building blocks of the attractor i.e., the periodic orbits. 
This paper contributes to the analysis of dynamical systems using periodic orbits in the following ways: 
\begin{itemize}
    \item First, it conducts an analysis of the singular vectors of the Hankel matrix derived from the long time-delay embeddings of unstable periodic orbits.
    \item Second, it offers a mathematical examination of disentangled unstable periodic orbits and their proximity to the plane of separation in the embedded space.
    \item Third, it presents a combinatorial formulation for calculating the number of unstable periodic orbits, given the sequence length and number of symbols, through a modification of the Redfield-Polyà enumeration theorem.
\end{itemize}  

The remainder of the paper is structured as follows: section \ref{sec:Background} details the Lorenz and Rössler attractors studied, their symbolic dynamics, and the datasets used for the subsequent analysis of their periodic orbits. In section \ref{sec:Theorems}, we introduce a theorem that explains the mechanism behind the separation of unstable periodic orbits in the embedded space. Finally, section \ref{sec:results} demonstrates the results of long time-delay embeddings for these unstable periodic orbits.

%% file: Background.tex
\section{Background}\label{sec:Background}
In this section, details about symbolic dynamics for unstable periodic orbits for the Lorenz and Rössler attractor are provided along with the dataset used for the subsequent construction of the Hankel matrices. In this paper, scalar variables are denoted by non-bold letters ($x,y,z$), vectors are denoted by bold lowercase letters ($\bx$) and matrices are denoted by non-bold uppercase letters ($U,V,\Sigma$) except for letters $A,B$ which are used to denote symbolic dynamics.
\subsection{Lorenz attractor}
The evolution equations of the state variables of the Lorenz system~\cite{lorenz1963deterministic} are given by the equations:
\begin{align}
    \dot{x}(t) &= \sigma(y-x),\nonumber\\
    \dot{y}(t) &= -x z+\rho x-y,\\
    \dot{z}(t) &= x y-\beta z.\nonumber
\end{align}
Here, the parameter values are taken to be $\beta=8/3, \sigma=10$ and $\rho=28$ which correspond to geometry of the system, Prandtl number and Rayleigh number respectively. The Lorenz system at the given values has three equilibrium values- one is the origin $\bp_0=[0,0,0]$, the other two are symmetric: $\bp_+$ and $\bp_-$ with coordinates $[\pm \sqrt{\beta(\rho-1)},\pm\sqrt{\beta(\rho-1)},\rho-1]$ for the parameter values considered here the equilibrium points are $[\pm 8.485, \pm 8.485, 27]$. 
\subsubsection{Symbolic dynamics of the Lorenz attractor}
In case of the Lorenz attractor, the phase space of the attractor is divided into two non-overlapping regions and the symbols $A$ and $B$ are assigned to each region. The two domains for the Lorenz attractor are depicted in Fig.~(\ref{fig:SymbDomainsLorenz}a). If the trajectory passes through the plane dividing the phase space and completes one loop in the region $x<0$, that part of the trajectory is labeled as $B$. Similarly, if the trajectory passes through the divider plane and completes one loop in the region $x >0$, that part of the trajectory is labeled as $A$~\cite{viswanath2003symbolic, barrio2015database}. 

A periodic orbit can now be represented as a finite sequence ($AB, ABBA, AABB,$ etc) that can be repeated indefinitely. In this work, the use of symbolic dynamics has been to mainly obtain a nomenclature for the periodic orbits.  

\subsubsection{Dataset of unstable periodic orbits of the Lorenz attractor}
Numerous algorithms have been proposed for the computation of the unstable periodic orbits of the Lorenz attractor \cite{boghosian2011unstable, viswanath2003symbolic, viswanath2001lindstedt,lan2004variational}. The unstable periodic orbits used in this analysis for the Lorenz attractor are obtained from the dataset of 111011 periodic orbits given in \cite{viswanath2003symbolic}. The computation of the periodic orbits is performed using an extension of the Lindsted-Poincaré algorithm \cite{viswanath2001lindstedt}. The computations exploit the symbolic dynamics of the attractor and notions from hyperbolicity theory. The dataset consists of all periodic orbits of sequence length of 20 or less. UPOs constructed in this dataset have values close to machine precision (14 accurate digits). This dataset is used for the subsequent analysis. The dataset can be found at \cite{LorenzDataset} and it contains the time series of each unstable periodic orbit for $x,y,z$ variables. The time periods of the orbits vary from $t_{min} =1.55865$ to $t_{max} = 9.15097$. This data will be referred to as the \textit{Divakar dataset} for the rest of this paper.

Due to the symmetry in the topology of the Lorenz attractor, if there exists a symbolic sequence $s_1 s_2 s_3...s_n$ where $s_i$ are either $A$ or $B$, then a corresponding mirror image of the sequence can be found by flipping the signs of $x$ or $y$ variables in the trajectory. For the purpose of analysis, UPOs are divided into three subsets: (a) unstable periodic orbits of type $A^n B$ and $B^nA$; (b) symmetric unstable periodic orbits of type $A^n B^n$; and (c) all symbolic sequences of length less than 8. Dividing the UPOs into these subsets will enhance our understanding of how periodic orbits separate in embedded space and the criteria for their clustering.

\subsection{Rössler attractor}\label{sec:RosslerPOTSD}
The evolution equations of the state variables of the Rössler system\cite{rossler1976equation,rossler1979equation} are given by the equations:
\begin{align}
    \dot{x}(t) &= -(y+z),\nonumber\\
    \dot{y}(t) &= x+ay,\\
    \dot{z}(t) &= b+z(x-c).\nonumber
\end{align}
Here, the chaotic Rössler attractor is studied for the parameter values $a=0.43295,b=2$ and $c=4$. The Rössler system has two equilibrium points given by $\bp_{+}$ and $\bp_{-}$ with the coordinates given by $[ap^{\pm},-p^{\pm},p^{\pm}]$
 where $p^{\pm} = \frac{c\pm \sqrt{c^2-4ab}}{2a}$. For values considered here, the equilibrium points are $\bp_{+}=[3.7703, -8.7085, 8.7085]$ and $\bp_{-}=[0.2297,  -0.5305, 0.5305]$. 
 The equilibrium points govern the behavior of the attractor. The $\bp_{-}$ point is at the center of the spiral divergence of the attractor whereas the $\bp_{+}$ point determine the folding of the attractor on itself. 
\subsubsection{Symbolic dynamics of the Rössler attractor}
To determine the symbolic dynamics of the Rössler attractor a Poincaré map is defined in Fig.~(\ref{fig:RosslerFRM}a). The Poincaré section is defined to be $x =ap^{-}$ which is a plane passing through the $\bp_{-}$ equilibrium point and $y<-p^{-}$. The first return map is computed for this Poincaré section and is shown in Fig.~(\ref{fig:RosslerFRM}b). 
In case of the Rössler attractor the phase space can be divided into two regions and the symbols $A$ and $B$ are assigned to each section. The two domains of the Rössler attractor are depicted in Fig.~(\ref{fig:RosslerFRM}c).
\begin{figure}
    \centering
    \subfloat[Poincaré section]{\includegraphics[width=0.33\textwidth]{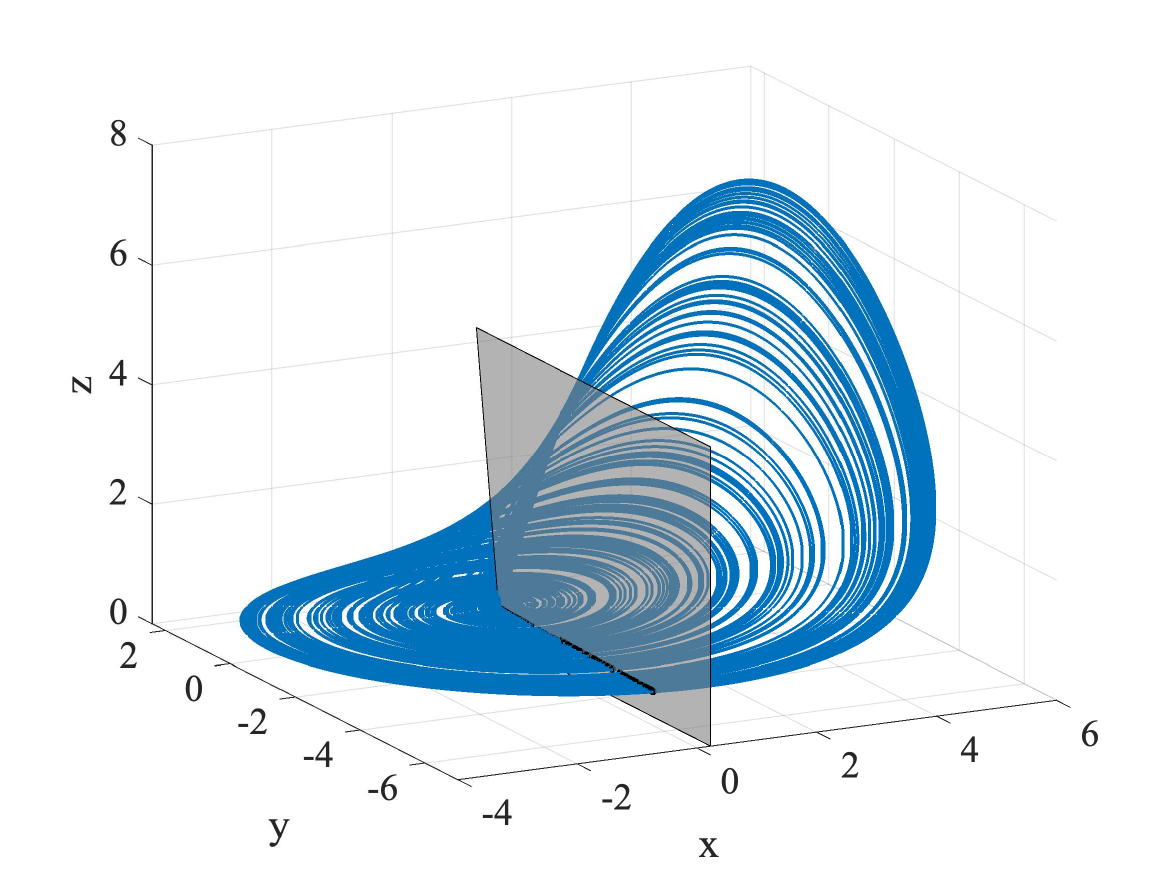}}
    \subfloat[First return map]{\includegraphics[width=0.33\textwidth]{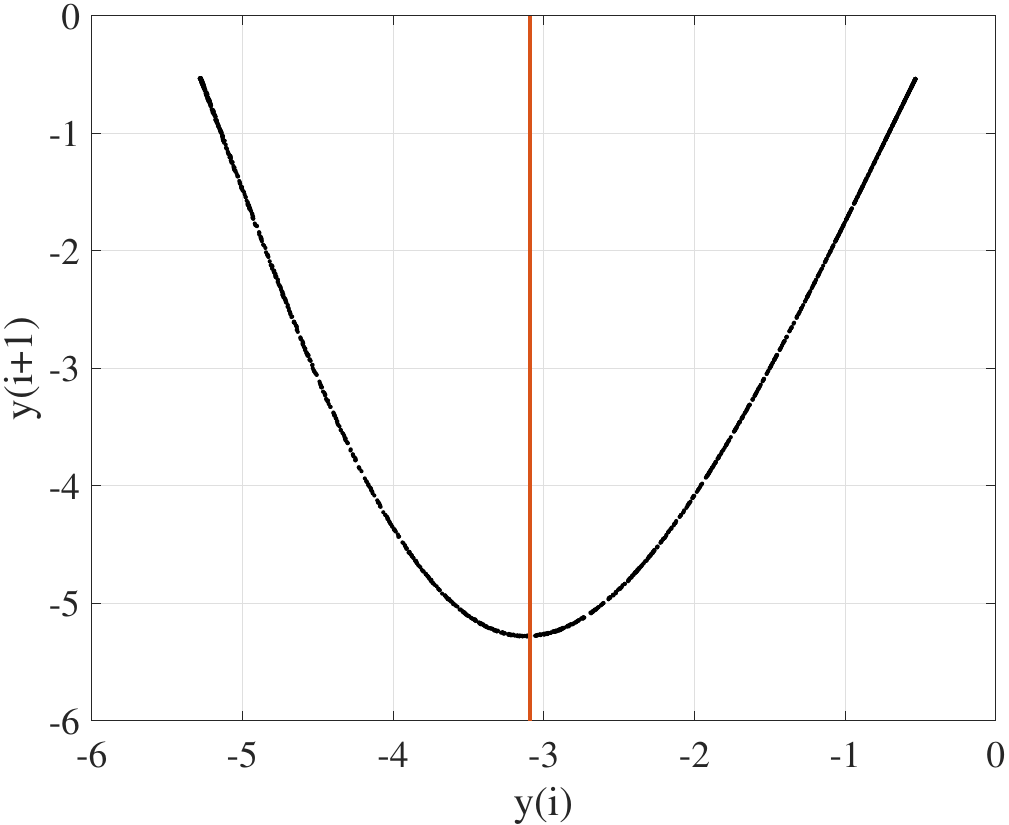}}
    \subfloat[Symbolic dynamics]{\includegraphics[width=0.33\textwidth]{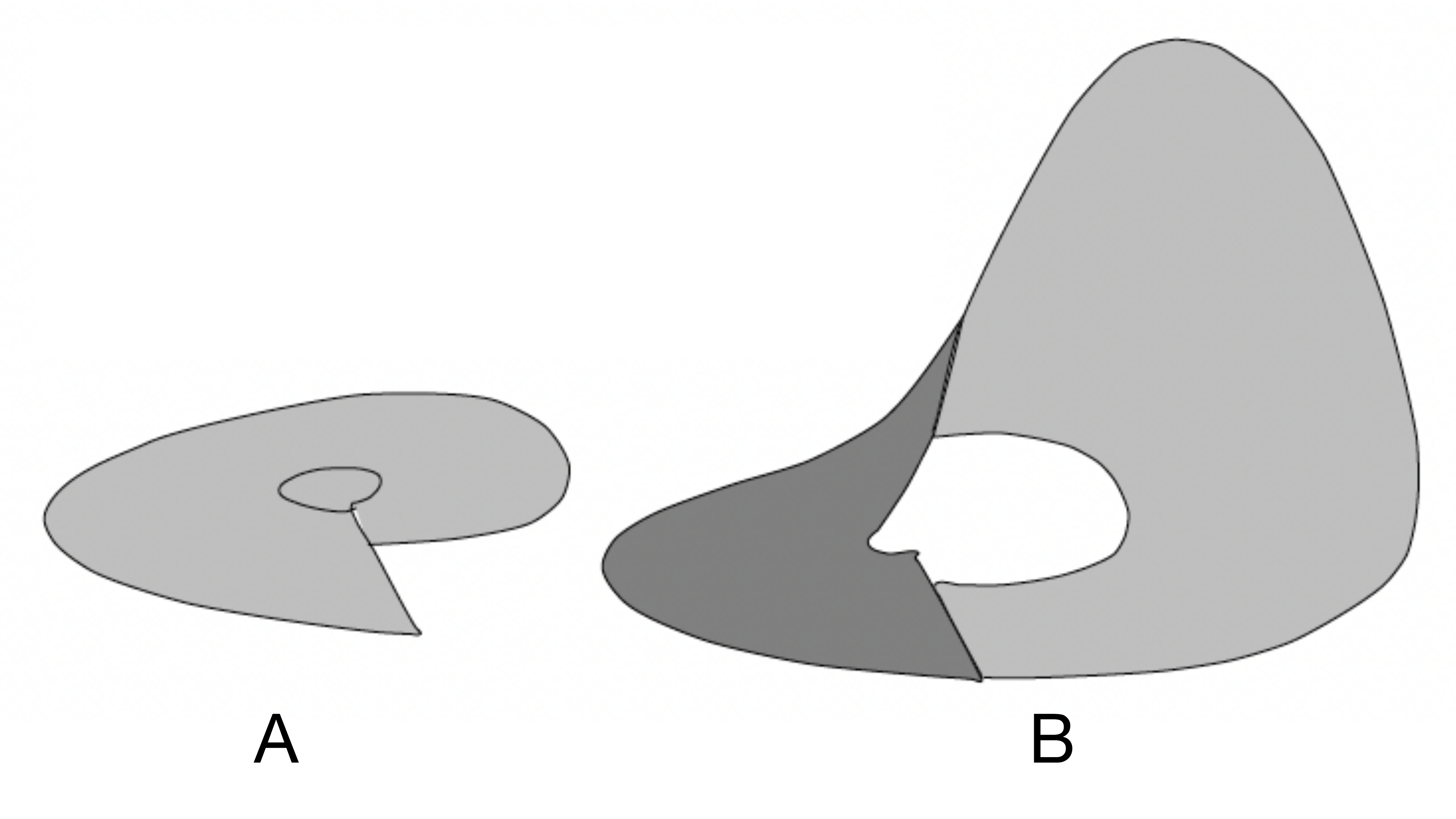}}
    \caption{Division of the phase space into two regions and symbolic notation  assigned to each region in the Rössler attractor.}
    \label{fig:RosslerFRM}
\end{figure}
The symbolic dynamics is given by
\begin{equation*}
    \begin{cases}
        A & \text{for }  y < y_c, \\
        B & \text{for }  y > y_c,
    \end{cases}
\end{equation*} 
where $y_c \approx -3.04$.
A periodic orbit can now be represented as a finite sequence ($AB,ABBA,AAAB$ etc) that can be repeated indefinitely. 
\subsubsection{Dataset of unstable periodic orbits of the Rössler attractor}
The unstable periodic orbits of the Rössler attractor are obtained from the dataset of periodic orbits given in \cite{letellier1994caracterisation}. The computation of the periodic orbits is performed using the method proposed by \cite{dutertre1995caracterisation}. In this case, the analysis is performed using 41 orbits computed for the value of parameters: $a = 0.43295, b = 2, c = 4.$. The UPOs are divided into two subsets: (a) unstable periodic orbits of type $A^nB$ and $AB^n$ and (b)  symbolic sequences of length less than 8. 

For the values of the parameters considered here, the population of periodic orbits is \textit{complete } i.e., all possible symbolic sequences are realized as distinct trajectories till sequence length 7. Also, the periodic orbits of sequence length 1 i.e., $A$ and $B$ are present as physical trajectories. 
The dataset of UPOs contains the values of the initial conditions for each of the unstable periodic orbits. The time series can be obtained by solving the system of equations for the Rössler attractor. The time period of the orbits varies from $t_{min}=6.225$ to $t_{max}=44.695$. This data will henceforth be referred to as the \textit{Christophe dataset}.

%% file: MethodOfDelays.tex
\subsection{Hankel matrix and time delay embedding}
Consider a dynamical system,
\begin{equation}
\dot{\bx}(t) = {\bf f}(\bx),
\end{equation}
where $\bx(t)\in\mathbb{R}^d$ represents a vector of state variables of the dynamical system at time $t$ and $\mathbf{f}:\mathbb{R}^d \rightarrow \mathbb{R}^d$ is a nonlinear function that represents the evolution of the state variables.
We take measurements of the state vector and define a scalar observable function $g \in\mathbb{R}^1$.
In the simplest case, this observable measures a single component of the state, i.e. $y(t) := g(\bx(t)) = x(t)$.
The ``delay matrix", also known as the ``Hankel matrix", is then defined as

\begin{equation*}
{H}_{p,q} =\tikz[remember picture, baseline=(mat.center)]{\node[inner sep=0](mat){$\begin{bmatrix} 
    y(t_0) & y(t_0 + \tau) & y(t_0 + 2\tau) & \ldots & y(t_0 + (q-1)\tau) \\
    y(t_0 + \tau) & y(t_0 + 2\tau) & y(t_0 + 3\tau) & \ldots & y(t_0 + q\tau)\\
    \vdots & \vdots & \vdots & \ddots & \vdots \\
    y(t_0 + (p-1)\tau) & y(t_0 + p\tau) & y(t_0 + (p+1)\tau) & \ldots & y(t_0 + (q+p-2)\tau)
    \end{bmatrix}$};}
\begin{tikzpicture}[overlay,remember picture,
>=Triangle]
\draw[black,thick,->] node[anchor=south west] (nn1) at (mat.north west)
{$t_{width}=q\tau$} (nn1.east) -- (nn1-|mat.north east) 
node[midway,above,black]{};
\draw[black,thick,->] node[anchor=north west,align=center, inner xsep=0pt] (nn2) at 
(mat.north east)
{$t_{height}=p\tau$} (nn2.south) -- (nn2.south|-mat.south) 
node[midway,above,black,rotate=-90]{};
\end{tikzpicture}.
\end{equation*}
Here, $q$ is called the embedding dimension which denotes the number of columns in the delay vectors and $p$ denotes the number of rows of the Hankel matrix. $\tau$ denotes the time interval between consecutive columns of the delay vector. If $\tau = \Delta t$, then the samples are taken from consecutive time measurements. The total time of the trajectory is denoted by $t_{width} = q\tau$. Total time of the trajectory denoted by the height of the Hankel matrix is $t_{height}=p\tau$. The skew diagonal terms of the Hankel matrix are constant. 

For multi-component observations, we can define a more general form of Hankel matrix~\cite{dylewsky2022principal} 
\begin{equation}
    {H}_{pd,q} = \begin{bmatrix} 
    \bx(t_0) & \bx(t_0 + \tau) & \bx(t_0 + 2\tau) & \ldots & \bx(t_0 + (q-1)\tau) \\
    \bx(t_0 + \tau) & \bx(t_0 + 2\tau) & \bx(t_0 + 3\tau) & \ldots & \bx(t_0 + q\tau)\\
    \vdots & \vdots & \vdots & \ddots & \vdots \\
    \bx(t_0 + (p-1)\tau) & \bx(t_0 + p\tau) & \bx(t_0 + (p+1)\tau) & \ldots & \bx(t_0 + (q+p-2)\tau)
    \end{bmatrix}.
\end{equation}
\begin{figure}[t]
    \centering
    \includegraphics[width=0.99\textwidth]{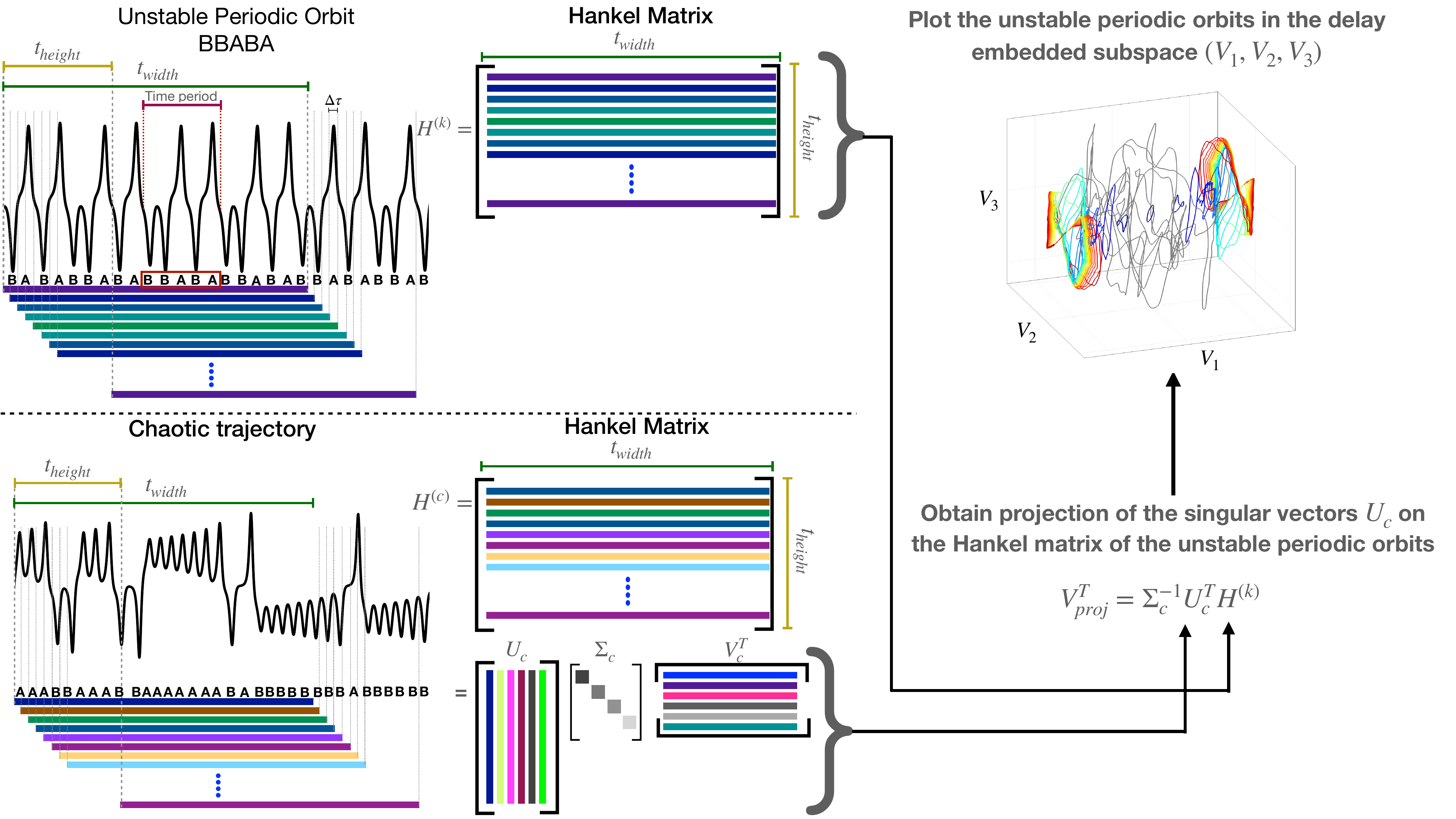}
    \includegraphics[width=0.99\linewidth]{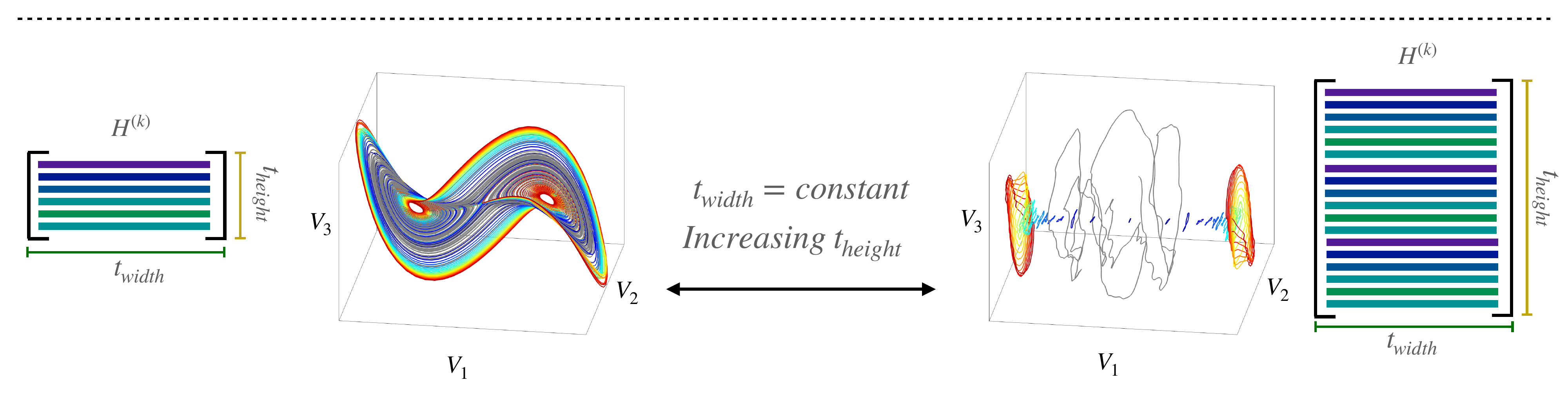}
    \caption{Construction of Hankel matrices for a single unstable periodic orbit and chaotic trajectory. The parameters for the construction of Hankel matrices are $t_{width}, t_{height}, \Delta \tau$ (denoted in the figure above). After performing the singular value decomposition of both the Hankel matrices the projection of the singular vectors is obtained $V_{proj}^T$. The projected vectors are then plotted in the embedded space for all the periodic orbits. The periodic orbits in the embedded space for increasing height of the Hankel matrix is shown in the panel at the bottom. For smaller $t_{height}$ the diffeomorphic form of the Lorenz attractor is obtained. However, as the height is subsequently increased the periodic orbits separate and form clusters.}
    \label{fig:HankelMatrix}
\end{figure}

 In the context of chaotic attractors, Whitney's embedding theorem~\cite{whitney1936differentiable} states that a smooth $m$-dimensional compact manifold $M$ can be diffeomorphically embedded into $\mathbb{R}^{2m+1}$, meaning that $M$ and its image are diffeomorphic. Takens extended Whitney's theorem~\cite{takens1981detecting} to delay-coordinate maps, showing that under generic conditions, a delay-coordinate map from a $m$-dimensional manifold to $\mathbb{R}^{2m+1}$ is also a diffeomorphism, making it accessible for reconstructing system dynamics from experimental time series data. 
 The quality of the reconstruction can be controlled using the time intervals between observations $\tau$ and the embedding dimensions ($p$ and $q$). 

The singular value decomposition (SVD) of the Hankel matrix is given by,
\begin{equation}
    {H}_{p,q} = {U \Sigma V^T}.
\end{equation}
The left singular vectors of the Hankel matrix have been shown to be an approximation of the space-time proper orthogonal decomposition (POD) modes~\cite{frame2023space}. The modes can also be interpreted as principal component trajectories (PCT), which can be used to reconstruct complex trajectories~\cite{dylewsky2022principal}. From time-delay scalar measurements, the principal components obtained from SVD form a temporal basis, with functions similar to a Fourier or wavelet basis. It has been shown that for long time-delay embeddings, the delay coordinate SVD converges to Fourier modes~\cite{vautard1989singular}. 

While the aforementioned methods, construct the Hankel matrix from the time series of the partial measurements of the dynamical system; in this work, we construct the Hankel matrix from the measurements of a single unstable periodic orbit of the attractor. The next section describes more details about the construction of the delay-embedded space using projections.

%% file: LongTimeDelayEmbeddings.tex
We construct the Hankel matrix for both the chaotic trajectories and the associated periodic orbits. 
The \texttt{ode45} function from MATLAB is used to advance the Lorenz and Rössler ordinary differential equations for both chaotic trajectory and trajectories of the unstable periodic orbits.
The Hankel matrix associated with the chaotic trajectory can be written as,
\begin{equation}\label{eq:ChoaticHankel}
    {H}^c_{p,q} = \begin{bmatrix} 
    x^c(t_0) & x^c(t_0 + \tau) & \ldots & x^c(t_0 + (q-1)\tau) \\
    x^c(t_0 + \tau) & x^c(t_0 + 2\tau)  & \ldots & x^c(t_0 + q\tau)\\
    \vdots & \vdots & \ddots & \vdots \\
    x^c(t_0 + (p-1)\tau) & x^c(t_0 + p\tau) & \ldots & x^c(t_0 + (q+p-2)\tau)
    \end{bmatrix}_{p \times q}.
\end{equation}
The superscript $x^c$ denotes the time series data obtained from the chaotic trajectory. Similarly, we also compute the Hankel matrices for the individual UPOs. Due to the chaotic nature of the dynamical systems and numerical errors in the computation of solution due to time-stepping, the UPOs deviate from the periodic trajectory. The data is artificially augmented by finding the time period of the UPO and replicating the copies of the time series of a single period until $t_{end}=1000$ which represents the temporal length of the augmented time series. Consider a UPO numbered $k$, the Hankel matrix can be denoted by ${H}^{(k)}_{p,q}$ and can be represented by, 
\begin{equation}\label{eq:UPOHankel}
     {H}^{(k)}_{p,q} = \begin{bmatrix} 
    x^{(k)}(t_0) & x^{(k)}(t_0 + \tau) & \ldots & x^{(k)}(t_0 + (q-1)\tau) \\
    x^{(k)}(t_0 + \tau) & x^{(k)}(t_0 + 2\tau)  & \ldots & x^{(k)}(t_0 + q\tau)\\
    \vdots & \vdots & \ddots & \vdots \\
    x^{(k)}(t_0 + (p-1)\tau) & x^{(k)}(t_0 + p\tau) & \ldots & x^{(k)}(t_0 + (q+p-2)\tau)
    \end{bmatrix}_{p\times q}.
\end{equation}
The subscript denoting the size of the Hankel matrix has been dropped in the subsequent text to maintain conciseness in the notation. We perform SVD of the Hankel matrix obtained from the chaotic trajectory. 
\begin{equation}
    {H}^c = {U}^c {\Sigma}^c {V^c}^T,
\end{equation}
where $U^c, V^c$ represent the left singular singular vectors and right singular singular vectors respectively and $\Sigma^c$ represents a diagonal matrix with the singular values as the diagonal terms arranged in descending value.  It is then possible to use this basis to embed a periodic orbit as
\begin{equation}
    {V}_{proj}^T = {\Sigma^c}^{-1} {U^c}^T {H}^{(k)},
\end{equation}
where ${V}_{proj}^T$ represents the projection of the time series onto the left singular vectors of the chaotic trajectory. Projecting ${\Sigma^c}^{-1} {U^c}^T$ onto the Hankel matrices $H^{(k)}$ separates the UPOs in the embedded space obtained by the basis vectors of $V_{proj}$ given by $\bv_{proj}^1, \bv_{proj}^2, \bv_{proj}^3$. The next section presents a mathematical analysis on the separation of UPOs, the characteristics of the resulting clusters, and the relative positions of individual UPOs in the embedded space.

\subsection{Separation of UPOs into clusters}
We now demonstrate the properties of the Hankel matrices of both the chaotic trajectory and unstable periodic orbits. First, we present a lemma which shows that the range of the Hankel matrix obtained from the UPO, $\mathcal{R}({H}^{(k)}_{p,q})$, is contained within the range of the Hankel of the chaotic trajectory, $\mathcal{R}(H^c_{p,q})$,  as the width of the Hankel matrices is increased for a constant height of the Hankel matrices. Here $\mathcal{R}$ denotes the range of the matrices. 

We then proceed to a theorem which shows the mechanism of the separation of the UPOs based on the projection of the vector $\bm{\rho}$ which is normal to the plane of separation. This shows that the separation of UPOs depends on the time spent by the UPOs in the different symbolic regions.

We finally show the existence of a vector similar to $\bm{\rho}$ in the delay-embedded space for which separation of the UPOs is possible.  
\begin{lemma}
    Let ${H}^c_{p,q}$ be the Hankel matrix obtained from the chaotic trajectory as shown in Eq.~(\ref{eq:ChoaticHankel}) and let ${H}^{(k)}_{p,q}$ be the Hankel matrix obtained from a periodic orbit as shown in Eq.~(\ref{eq:UPOHankel}). Then, there exists a natural number $n$ such that $\mathcal{R}({H}^{(k)}_{p,q}) \subset \mathcal{R}(H^c_{p,q})$ for all $q \geq n$.
\end{lemma}
\begin{proof} Consider a column of $H^{(k)}_{p,q}$ denoted by $\begin{bmatrix}
    x^{(k)}(t_0)\\x^{(k)}(t_1)\\ \vdots \\ x^{(k)}(t_{p-1})
\end{bmatrix}$. \\
By continuity of the flow, $\forall \epsilon > 0$ there exists $\delta >0$ such that $\lVert x(t_0)-\tilde{x}(t_0) \rVert < \delta$ then,
\begin{equation*}
    \norm{\begin{bmatrix}
        x(t_0)\\x(t_1)\\ \vdots \\ x(t_{p-1})
    \end{bmatrix} - \begin{bmatrix}
        \tilde{x}(t_0)\\ \tilde{x}(t_1)\\ \vdots \\ \tilde{x}(t_{p-1})
    \end{bmatrix}} < \epsilon
\end{equation*}
By ergodicity, $\exists$ $t_{\epsilon}$ such that $\lVert x^{(k)}(t_0)-{x}^c(t_\epsilon) \rVert < \delta$ 
\begin{equation*}
    \norm{\begin{bmatrix}
        x^{(k)}(t_0)\\x^{(k)}(t_1)\\ \vdots \\ x^{(k)}(t_{p-1})
    \end{bmatrix} - \begin{bmatrix}
        {x}^c(t_\epsilon)\\ {x}^c(t_\epsilon+ \tau)\\ \vdots \\ {x}^c(t_{\epsilon}+(p-1)\tau)
    \end{bmatrix}} < \epsilon
\end{equation*}
$\exists \tilde{q}$ such that $\overline{x}^c = \begin{bmatrix}
    {x}^c(t)\\ {x}^c(t+\tau)\\ \vdots \\ {x}^c(t+(p-1)\tau)
\end{bmatrix}$ is a column of the $H^c_{p,\tilde{q}}$. 
Hence, $\overline{x}^c \in \bigcup\limits_{\tilde{q}}\mathcal{R}(H^c_{p,\tilde{q}})$.\\
\begin{equation}\label{eq:lemma1.1}
    \therefore \mathcal{R}(H^{(k)}_{p,\tilde{q}}) \subset \overline{\bigcup\limits_{\tilde{q}} \mathcal{R}(H^c_{p,\tilde{q}})}.
\end{equation}
Here, the overline denotes the closure of union of range of the Hankel matrix. Now, $\exists n$ such that $\overline{\bigcup\limits_{\tilde{q}} \mathcal{R}(H^c_{p,\tilde{q}})} = \mathcal{R}(H^c_{p,n})$ i.e.,
\begin{equation}\label{eq:lemma1.2}
    \mathcal{R}(H^c_{p,1}) \subset \mathcal{R}(H^c_{p,2}) \subset \cdots \subset \mathcal{R}(H^c_{p,n-1}) \subset \mathcal{R}(H^c_{p,n}) = \mathcal{R}(H^c_{p,n+1}) =\mathcal{R}(H^c_{p,n+2})
\end{equation}
Let $d_k$ be the dimension of $\mathcal{R}(H^c_{p,k}) $ and $d_k$ be bounded above.  
\begin{equation*}
    d_1 \leq d_2 \leq d_3 \leq \cdots \leq d_k.
\end{equation*}
By the monotone convergence theorem, $d_k$ converges to $d$.\\
$\forall \epsilon > 0$ there exists a $n_q$ such that $\left| d_k - d \right| < \epsilon$ for all $k\geq n_q$.

Using Eq.~(\ref{eq:lemma1.1}) and Eq.~(\ref{eq:lemma1.2}), for all $q$,
\begin{equation*}
    \mathcal{R}(H^{(k)}_{p,q}) \subset \overline{\bigcup\limits_{\tilde{q}} \mathcal{R}(H^c_{p,\tilde{q}})} =\mathcal{R}(H^c_{p,n}) = \mathcal{R}(H^c_{n+1}).
\end{equation*}
Take $q\geq n$, 
\begin{align*}
    \mathcal{R}(H^{(k)}_{p,q}) \subset \mathcal{R}(H^c_{p,n}) = \mathcal{R}(H^c_{p,\tilde{n}})\\
    \therefore \mathcal{R}(H^{(k)}_{p,\tilde{n}}) \subset \mathcal{R}(H^c_{p,\tilde{n}}) \qquad \forall \tilde{n} \geq n.
\end{align*}

\end{proof}

\begin{theorem}\label{th:SeparationOfUPOs}
Let $\bx:\mathbb{R} \to \mathbb{R}^d$ be a periodic orbit with period $\tilde{t}$, that is, $\bx(t) = \bx(t+\tilde{t})$ for every $t \in \mathbb{R}$ and let $H_{pd,q}$ be the $pd \times q$ Hankel matrix formed by sampling this trajectory at uniform intervals $\Delta t$ starting from $t_0$ with $\Delta \theta := \Delta t / \tilde{t}$ irrational.
Given $\bm{\sigma} \in \mathbb{R}^d$ defining regions $A = \{ \bx \in \mathbb{R}^d : \bm{\sigma}^T \bx \geq 0 \}$ and $B = \{ \bx \in \mathbb{R}^d : \bm{\sigma}^T \bx < 0 \}$, let $\mathcal{I}_A = \{t \in [0,\tilde{t}) : \bx(t) \in A \}$, $\mathcal{I}_B = \{t \in [0,\tilde{t}) : \bx(t) \in B \}$, and
$$
\alpha = \frac{1}{|\mathcal{I}_A|} \int_{\mathcal{I}_A} \bm{\sigma}^T \bx(t) \ \text{d}t, \qquad
\beta = \frac{1}{|\mathcal{I}_B|} \int_{\mathcal{I}_B} \bm{\sigma}^T \bx(t) \ \text{d}t.
$$
Then we have
$$
\lim_{p \to \infty} \frac{1}{p} \left(\mathds{1}_p \otimes \bm{\sigma} \right)^T H_{pd,q}
= \left( \alpha \frac{|\mathcal{I}_A|}{\tilde{t}} + \beta \frac{|\mathcal{I}_B|}{\tilde{t}} \right) \mathds{1}_q^T.
$$
\end{theorem}
\begin{proof}
Consider a Hankel matrix of $q$ columns ($width=q$) and $p$ rows ($height=p$),
\begin{equation*}
    {H}^{(k)}_{pd,q} = \begin{bmatrix}
        \mathbf{x}^{(k)}(t_0) & \mathbf{x}^{(k)}(t_0+\tau) & \mathbf{x}^{(k)}(t_0 + 2\tau) & \ldots & \mathbf{x}^{(k)}(t_0 + (q-1)\tau) \\ 
        \mathbf{x}^{(k)}(t_0+\tau) & \mathbf{x}^{(k)}(t_0+2\tau) & \mathbf{x}^{(k)}(t_0 + 3\tau) & \ldots & \mathbf{x}^{(k)}(t_0 + q\tau) \\
        \vdots & \vdots &  \vdots &\ddots & \vdots \\
        \mathbf{x}^{(k)}(t_0+(p-1)\tau) & \mathbf{x}^{(k)}(t_0+p\tau) & \mathbf{x}^{(k)}(t_0 + (p+1)\tau) & \ldots & \mathbf{x}^{(k)}(t_0 + (p+q-2)\tau)
    \end{bmatrix}.
\end{equation*}
Here, ${H}^{(k)}_{pd,q}$ represents the Hankel matrix of the $k$-th UPO and size of the Hankel matrix is $\text{rows} \times \text{columns} = pd\times q$ denoted by the subscripts. Each element of the Hankel matrix is the state vector of the UPO. In case of the Lorenz attractor the state vector will be $\mathbf{x} = [x \quad y \quad z]^T$. 

Consider a vector $\bm{\sigma}$ which is the normal to the plane of separation of symbolic dynamics. The plane of separation divides the phase space into two regions $A = \{ \bx \in \mathbb{R}^d : \bm{\sigma}^T \bx \geq 0 \}$ and $B = \{ \bx \in \mathbb{R}^d : \bm{\sigma}^T \bx < 0 \}$. 
Let, 
\begin{equation*}
    \vec{\bm{\sigma}}_p = \frac{1}{p} \mathds{1} \otimes \bm{\sigma} = \frac{1}{p}\begin{bmatrix}
        \bm{\sigma}\\ \bm{\sigma} \\ \bm{\sigma} \\ \vdots \\ \bm{\sigma}
    \end{bmatrix}_{pd \times 1}.
\end{equation*}
The vector $\vec{\bm{\sigma}}_p$ represents the repetition of the $\bm{\sigma}$ vector $p$ times. Now consider, 
\begin{align*}
    \vec{\bm{\sigma}}_p ^T  {H}^{(k)}_{pd,q} &= \frac{1}{p} \begin{bmatrix}
        \bm{\sigma}^T\quad \bm{\sigma}^T \quad  \ldots \quad \bm{\sigma}^T
    \end{bmatrix}_{1 \times pd} \begin{bmatrix}
        \mathbf{x}^{(k)}(t_0) & \mathbf{x}^{(k)}(t_0+\tau)  & \ldots & \mathbf{x}^{(k)}(t_0 + (q-1)\tau) \\ 
        \mathbf{x}^{(k)}(t_0+\tau) & \mathbf{x}^{(k)}(t_0+2\tau)  & \ldots & \mathbf{x}^{(k)}(t_0 + q\tau) \\
        \vdots & \vdots & \ddots & \vdots \\
        \mathbf{x}^{(k)}(t_0+(p-1)\tau) & \mathbf{x}^{(k)}(t_0+p\tau) & \ldots & \mathbf{x}^{(k)}(t_0 + (p+q-2)\tau)
    \end{bmatrix}_{pd \times q}\\
    &= \frac{1}{p} \begin{bmatrix}
        \sum_{m=0}^{p-1} \bm{\sigma}^T \mathbf{x}^{(k)}(t_m) &  \sum_{m=1}^{p} \bm{\sigma}^T\mathbf{x}^{(k)}(t_m) & \ldots & \sum_{m=q-1}^{p+q-2} \bm{\sigma}^T \mathbf{x}^{(k)}(t_m)
    \end{bmatrix}_{1\times q}\\
\end{align*}
\begin{figure}[t!]
    \centering
    \includegraphics[width=0.70\textwidth]{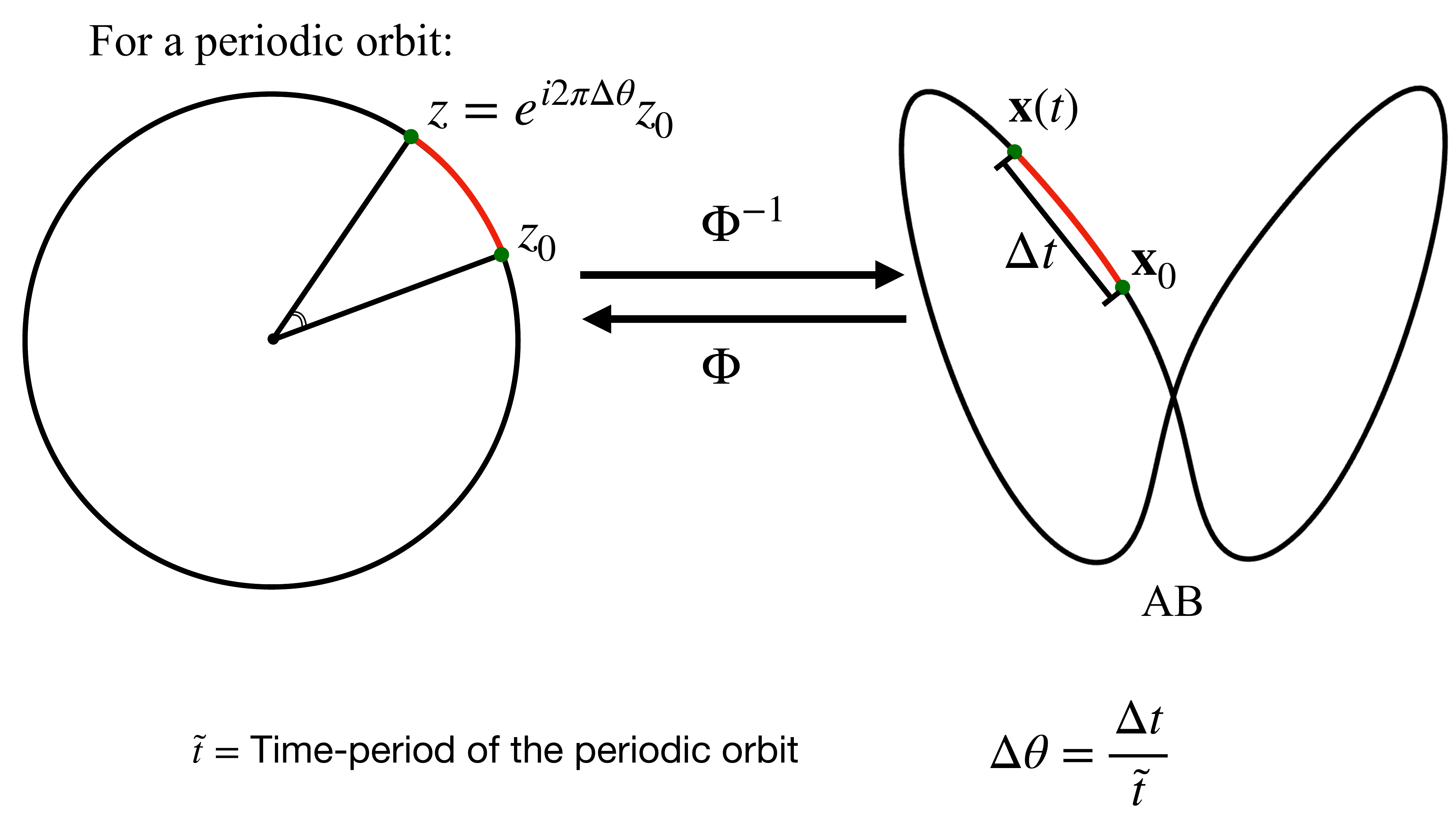}
    \caption{Schematic of the transformation of a UPO onto a unit circle and every time-step of the UPO represents the rotation along the circle.}
    \label{fig:CircleUPO}
\end{figure}
As shown in Fig.~(\ref{fig:CircleUPO}), consider the diffeomorphism of the periodic orbit to a unit circle. The function $\Phi^{(k)}(\bm{x}^{(k)}) = z$ is used to obtain the transformation from the periodic orbit to the point along a unit circle. This function $\Phi^{(k)}$ is invertible i.e., $\left(\Phi^{(k)}\right)^{-1}(z) = \bm{x}^{(k)} $. We consider a transformation denoted as $T_{\Delta\theta}$ such that $T_{\Delta\theta}(z) = e^{i 2\pi\Delta\theta}z$. This transformation corresponds to a rotation on a circle, where the angle evolves by $2\pi\Delta \theta$ radians for each step. If $\Delta\theta$ is rational, then, for any $z_0$, the sequence of iterates will visit only finite points, and the process is not ergodic. On the other hand, if $\Delta \theta$ is irrational then $T_{\Delta\theta}^n$ never repeats and we can show that the process is ergodic~\cite{cornfeld2012ergodic}.

Thus, using the transformation $\left(\Phi^{(k)}\right)(z) = \bm{x}^{(k)}$, we can say that for the periodic orbit to be ergodic, we require $\frac{\Delta t}{\tilde{t}^{(k)}}$ should be irrational, where $\tilde{t}^{(k)}$ represents the time period of the $k$-th periodic orbit. Given this condition is satisfied, the periodic orbit is ergodic and therefore, we can use the Birkoff's ergodic theorem~\cite{cornfeld2012ergodic} which states that, given an ergodic transformation on a space, the discrete time average of a function under this transformation converges to the space average of that function under the measure:
\begin{equation}\label{eq:SigmaTH}
    \lim_{p \to \infty}\frac{1}{p}\begin{bmatrix}
        \sum_{m=0}^{p-1} \bm{\sigma}^T \mathbf{x}^{(k)}(t_m) &  \sum_{m=1}^{p} \bm{\sigma}^T \mathbf{x}^{(k)}(t_m) & \ldots & \sum_{m=q-1}^{p+q-2} \bm{\sigma}^T \mathbf{x}^{(k)}(t_m)
    \end{bmatrix}
\end{equation}
Now, considering the definition of symbolic dynamics of the given attractor, where the trajectory is divided into regions defined by symbols, 
\begin{align*}
    A = \{ \bx \in \mathbb{R}^d : \bm{\sigma}^T \bx \geq 0 \}, \\
    B = \{ \bx \in \mathbb{R}^d : \bm{\sigma}^T \bx < 0 \}.
\end{align*}
Let the time spent by the unstable periodic orbit in each of the regions be,
\begin{align*}
    \mathcal{I}_A = \{t \in [0,\tilde{t}) : \bx(t) \in A \},\\
    \mathcal{I}_B = \{t \in [0,\tilde{t}) : \bx(t) \in B \}.
\end{align*}
For a given symbol, the average distance of the trajectory from the plane of separation is given by,  
\begin{align}\label{eq:IaIb}
    \alpha = \frac{1}{|\mathcal{I}_A|} \int_{\mathcal{I}_A} \bm{\sigma}^T \bx(t) \ \text{d}t, \\
\beta = \frac{1}{|\mathcal{I}_B|} \int_{\mathcal{I}_B} \bm{\sigma}^T \bx(t) \ \text{d}t.
\end{align}
For a given UPO, the value of $\bm{\sigma}^T \bx$ can be simplified to,
\begin{equation}
   \frac{1}{\tilde{t}}\int_0^{\tilde{t}} \bm{\sigma}^T \bx(t) dt = \left( \alpha \frac{|\mathcal{I}_A|}{\tilde{t}} + \beta \frac{|\mathcal{I}_B|}{\tilde{t}} \right).
\end{equation}
Thus, Eq.~(\ref{eq:SigmaTH}) can be rewritten as,
\begin{equation*}
    \lim_{p \to \infty}\frac{1}{p}\begin{bmatrix}
        \sum_{m=0}^{p-1} \bm{\sigma}^T \mathbf{x}^{(k)}(t_m) &  \sum_{m=1}^{p} \bm{\sigma}^T \mathbf{x}^{(k)}(t_m) & \ldots & \sum_{m=q-1}^{p+q-2} \bm{\sigma}^T \mathbf{x}^{(k)}(t_m)
    \end{bmatrix} = \left( \alpha \frac{|\mathcal{I}_A|}{\tilde{t}} + \beta \frac{|\mathcal{I}_B|}{\tilde{t}} \right) \mathds{1}_q^T.
\end{equation*}
\end{proof}
In case of the Lorenz attractor, the plane of separation is given by, $x=0$ and $\bm{\sigma} = [1 \quad 0 \quad 0]$. With this choice, the values of $\alpha$ and $\beta$ are nearly independent of the choice of UPO and satisfy $\beta \approx -\alpha$ by symmetry. This has been computed numerically and compared for a few UPOs and shown in Table~(\ref{tab:AlphavsBeta}).
\begin{table}
    \centering
    \begin{tabular}{|c|c|c|c|}
         \hline
         UPO     & $\alpha$ & $\beta$ \\
         \hline
         AB      & 6.0399   & -5.9294 \\
         AAB     & 5.9420   & -6.1885 \\
         ABBA    & 6.1807   & -6.2399 \\ 
         ABABA   & 5.9550   & -6.1489 \\ 
         AABABB  & 6.0217   & -6.1235 \\ 
         AABBBAA & 6.1883   & -6.6153 \\
         \hline
    \end{tabular}
    \caption{Numerical analyses of the values of $\alpha$ and $\beta$ computed from data for different UPOs. Column 1 denotes the symbolic name of a UPO. Columns 2 and 3 denote the values of $\alpha =\frac{1}{|\mathcal{I}_A|} \int_{\mathcal{I}_A} \bm{\sigma}^T \bx(t) \ \text{d}t $ and $\beta =\frac{1}{|\mathcal{I}_B|} \int_{\mathcal{I}_B} \bm{\sigma}^T \bx(t) \ \text{d}t$ respectively. We see that the the value of $\beta$ is negative and approximately equal to the value of $\alpha$: $\beta \approx \alpha$.}
    \label{tab:AlphavsBeta}
\end{table}

\begin{table}
    \centering
    \begin{tabular}{|c|c|c|c|}
         \hline
         UPO & $\frac{B}{A+B}$ & $\frac{\text{Time spent in B lobe}}{\text{Total time period}}$  & \% Error\\
         \hline
         AB      & 0.5      & 0.5    &  0\\
         ABB     & 0.6667   & 0.6688 & 0.3247\\
         ABBB    & 0.75     & 0.7438 & 0.8264\\ 
         ABBBB   & 0.80     & 0.7869 & 1.6421\\ 
         ABBBBB  & 0.833333 & 0.8167 & 1.9909\\
         ABBBBBB & 0.8571   & 0.8374 & 2.3017\\
         AABABA  & 0.3363   & 0.3333 & 0.89 \\ 
         AABABB  & 0.5049   & 0.5000 & 0.9918 \\
         BBABAB  & 0.6637   & 0.6667	& 0.4462\\
         \hline
    \end{tabular}
    \caption{Analysis of the relation between the symbolic dynamics of the Lorenz attractor and the time spent by the UPOs in the $B-$lobes. The first column denotes the symbolic name of the UPO. The second column shows the ratio of number of $B$ symbols in the UPO with the total sequence length of the UPO. In the third column, we compute the number of time steps spent by the UPO in the $B-$lobe vs the total time steps in 1 time period of the unstable periodic orbit. The error between columns 2 and 3 is calculated in column 4.}
    \label{tab:UPORatioAtoBData}
\end{table}
For a UPO $\bx^{(k)}(t)$ taking $M$ trips around lobe $A$ and $N$ trips around lobe $B$, the fractions of time spent on each lobe are approximately $|\mathcal{I}_A|/\tilde{t} \approx M/(M+N)$, and $|\mathcal{I}_B|/\tilde{t} \approx N/(M+N)$. This is verified numerically and the results are presented in Table~(\ref{tab:UPORatioAtoBData}). The table shows data for a few unstable periodic orbits which relates the symbolic dynamics of the Lorenz attractor to the time spent by the UPO in each lobe. The second column in the table denotes the ratio of the $B$-symbols in a UPO to the total number of symbols in the UPO. The third column denotes the ratio of the time spent by the UPO in $B$-lobe to the total time period of the UPO. We observe that the two values match each other and the difference in the values for higher symbolic length is less than 5\%. From this data, we observe that the time spent by the periodic orbit for one loop in either $A$ or $B$-lobe is approximately equal. The speed of the evolution of the loops depends on the distance from the equilibrium point. Hence for smaller loops i.e., closer to the equilibrium point, the periodic orbit loop evolves slowly whereas for larger loops the evolution is faster as the distance is larger. Thus, for a UPO whose symbolic name is $ABBAA$ with $\tilde{t}$ as the time-period, the periodic orbit spends 3/5-ths of $\tilde{t}$ in the $A$-lobe and 2/5-ths of $\tilde{t}$ in the $B$-lobe. 

Defining the ratio $\rho_k = M/N$ and applying the theorem above with these approximations gives
$$
\lim_{p\to\infty} \frac{1}{p}\left(\mathds{1}_p^T \otimes \begin{bmatrix}
    1 & 0 & 0
\end{bmatrix} \right) H_{p,q}^{(k)}
\approx
\alpha \left( \frac{\rho_k - 1}{\rho_k + 1} \right) \mathds{1}_q^T.
$$
This means that for sufficiently large $p$, the projection of the Hankel matrix along a certain direction takes a constant value measuring the ratio of $A$-loops to $B$-loops around the UPO.

The Lorenz attractor is a special case, where the number of the symbols in the symbolic name of the UPO corresponds to its distance from the plane of separation which is $x=0$. However, in case of the Rössler attractor the ratio $\rho^{(k)}$ will denote the ratio of time spent by the UPO in $y<0$ to the amount of time spent by the UPO in $y>0$.

\begin{theorem}
    Let $H = U \Sigma V^T$ and $\tilde{H} = \tilde{U}\tilde{\Sigma}\tilde{V}^T$ be the economy SVDs of two real matrices with $m$ rows and let $\bm{\rho} \in \mathbb{R}^m$.
    If $\mathcal{R}(\tilde{H}) \subset \mathcal{R}(H)$, then there exists a vector $\hat{\bm{\rho}} \in \mathbb{R}^{\text{rank}(H)}$ satisfying
    $
    \bm{\rho}^T \tilde{H} 
    = \hat{\bm{\rho}}^T \Sigma^{-1} U^T \tilde{H}.
    $
\end{theorem}
\begin{proof}
    We first show that $\tilde{U}^T U$ is surjective (onto). Let there exist $\bm{q}$ such that $\bm{q} \perp \mathcal{R}(\tilde{U}^T U)$. 
    \begin{equation*}
        U^T \tilde{U} \bm{q} = 0  \implies \tilde{U}\bm{q} \in \mathcal{N}(U^T) =\mathcal{R}(U)^{\perp},
    \end{equation*}
    where $\mathcal{N}(U^T)$ denotes the null space of $U^T$. \\
    Hence, 
    \begin{equation*}
        \tilde{U}\bm{q} \in \mathcal{R}(U) \cap \mathcal{R}(U)^{\perp} = \{0\}\\
        \implies \tilde{U} \bm{q} = 0 
    \end{equation*}
    We, thus show that $\bm{q}=\{ 0 \}$ for arbitrary $\bm{q}$ in $\mathcal{R}(\tilde{U}^T U)^{\perp}$ then, $\mathcal{R}(\tilde{U}^T U)$ becomes the whole space $W$ as $W = \mathcal{R}(\tilde{U}^T U) \oplus \mathcal{R}(\tilde{U}^T U)^{\perp}$.  
    Hence, $\tilde{U}^T U$ is surjective.
    
    Now, $\exists$ $\bm{w}$ such that $\tilde{U}^T \bm{\rho} = \tilde{U}^T U \bm{w}$. Let $\hat{\bm{\rho}} = \Sigma \bm{w}$, then,
    \begin{align*}
        \tilde{U}^T U \Sigma^{-1} \hat{\bm{\rho}} &= \tilde{U}^T U \Sigma^{-1} \Sigma \bm{w} \\
                                             &= \tilde{U}^T U \bm{w} \\
                                             &= \tilde{U} \bm{\rho} 
    \end{align*}
    \begin{equation}\label{eq:lemma1}
        \therefore \bm{\rho}^T \tilde{U} = \hat{\bm{\rho}}^T \Sigma^{-1} U^T \tilde{U}.
    \end{equation}
    Now, using the property of left singular vectors of SVD $\tilde{U}\tilde{U}^T= I$ we can obtain,
    \begin{align*}
        \bm{\rho}^T \tilde{H} &= \bm{\rho}^T \tilde{U}\tilde{U}^T \tilde{H} \\
                         &= \hat{\bm{\rho}}^T\Sigma^{-1} U^T \tilde{U} \tilde{U}^T \tilde{H} \qquad \text{Using Eq.~(\ref{eq:lemma1})}\\
                         &= \hat{\bm{\rho}}^T \Sigma^{-1} U^T \tilde{H}.
    \end{align*}
 
\end{proof}

\subsection{Computing the number of unstable periodic orbits for a given sequence length}
Finding the total number of unique periodic orbits for a given sequence length and number of symbols is an interesting combinatorics problem. In this work, we utilize the Redfield-Polyá enumeration theorem and impose a few more constraints to find the total number of periodic orbits for a given sequence length $n$ and $k-$symbols. In 1937, George Polyá published a paper~\cite{polya1937kombinatorische} in combinatorial analysis describing the permutations of configurations with inherent symmetry~\cite{de1971survey, riordan1957combinatorial}. Although this theorem was popularized by Polyá, it was first published by JH Redfield~\cite{redfield1927theory} in 1927 and hence is sometimes referred to as the Redfield-Polyá (R-P) enumeration theorem. The theorem is a generalization of the Burnside's lemma~\cite{burnside1911theory} which takes in account symmetry for computing the permutations. The theorem has found applications in theoretical physics, number theory and graph theory. It has been particularly used in chemistry to find the enumeration of compounds.

The Redfield-Polyá (R-P) enumeration theorem can be used to answer questions in combinatorics such as: how many non-equivalent combinations can be formed for a $n$-bead necklace using beads of $k$-colors?
The concept of finding the total number of combinations of periodic orbits of a given $n-$sequence length with $k$-symbols ($A,B,C,...$) is exactly the same as the necklace enumeration problem. The total number of periodic orbits for a given sequence length can be computed using the R-P enumeration theorem by enforcing constraints on the number of combinations. Using the R-P enumeration theorem all the cyclic combinations of the UPOs can be obtained however, we notice that sequences $AB$ and $ABAB$ are considered to be the same and need to be eliminated from the unique UPOs computation. The sequence $ABAB$ can be considered as a periodic orbits of sequence length 2. The larger sequence $n$ is non-unique in cases where smaller repeating sequences are formed within the larger sequence.  Besides this, mono-symbolic sequences of type $A^n$ and $B^n$ are considered invalid in case of Lorenz i.e., we impose a constraint at all symbols should be used for a symbolic sequence to be valid. Hence, after imposing the above mentioned constraints on the R-P theorem, the number of unique periodic orbits of $n$ sequence length for $k$ symbols is given by a recursive function $P(n)$,
\begin{align}
    P(n)= \begin{cases}
        \frac{1}{n} \sum_{i=1}^n k^{gcd(i,n)} - \sum_{i\vert n} P(i) - k,  &n > k, \\ 
          (n-1)! & n=k, \\ 
          0 & n<k.
         \end{cases}
\end{align}
where, $i \vert n$ is shorthand for \lq$i$ divides $n$\rq and $gcd(i,n)$ is the greatest common divisor between two numbers $i$ and $n$.

The problem of $n$-ary necklace differs from finding the unique symbolic sequences for UPOs due to additional constraints imposed on the periodic orbits. In symbolic dynamics, repeating sequences are considered equivalent, since the entire sequence is supposed to repeat infinitely. To eliminate this repetition, it is crucial to find the divisors of $n$ and eliminate those sequences to maintain the uniqueness of the UPOs. The term $P(i)$ takes care of this redundancy in the equation for $P(n)$.  
Table~(\ref{tab:CountUPOs}) refers to the number of UPOs obtained for the Lorenz attractor for 2 symbols ($A$ and $B$) up to the sequence length 14 ($k=2$ and $n=14$). Column 1 denotes the sequence length. The second column denotes the naïve combinations that can be computed in cases where the cyclic nature and symmetry are not taken into account. The third and fourth columns denote the number of combinations obtained from the R-P theorem and constrained R-P theorem. Fig.~(\ref{fig:NumUPOs}b) denotes naïve total possible combinations of UPOs, number of possible combinations of UPOs obtained from Redfield-Pólya enumeration theorem and the unique UPOs obtained by a modification to the Redfield-Pólya enumeration theorem. 

The symbolic dynamics assigns a distinct label to each distinct trajectory; however, there might be symbolic sequences which are not realized as trajectories. If all possible symbolic sequences are realized as physical trajectories then the symbolic dynamics is called \textit{complete}, if some sequences are not realized as physical trajectories then the symbolic dynamics is \textit{pruned}~\cite{cvitanovic1993symmetry}. The parameter values taken for both the Lorenz and the Rössler case provide a \textit{complete} set of UPOs.

\begin{table}
    \centering
    \begin{tabular}{|c|c|c|c|}
    \hline
         Sequence  & Total non-cyclic  & Combinations from  & Modified Redfield-Pólya   \\
         length & possible combinations & Redfield-Pólya enumeration & enumeration theorem\\
          & &  theorem & \\
          \hline
          2 &   $4$ &   3    &   1    \\
          3 &   $8$ &  4     &    2   \\
          4 &   $16$ &   6    &    3   \\
          5 &   $32$ &   8    &   6    \\
          6 &   $64$ &   14    &  9     \\
          7 &   $128$ &    20   &  18     \\
          8 &   $256$ &   36    &   30   \\
          9 &   $512$ &   60    &    56   \\
          10 &  $1024$ &  108     &   99    \\
          11 &  $2048$ &   188    &    186   \\
          12 &   $4096$ &  352     &    335   \\
          13 &   $8192$ &  632     &    630   \\
          14 &   $16384$ & 1182 &  1161 \\
          \hline
    \end{tabular}
    \caption{The possible combinations of unique UPOs obtained for a given sequence length. The second column denotes the possible combinations of the sequences without taking into account the cyclicity of the combinations. The Redfield-Pólya enumeration theorem takes into account the cyclicity and the symmetries in the possible combinations and provides the total combinations. In the context of UPOs for the Lorenz attractor, symbolic sequences of just one letter (example, $A^n$ and $B^n$) are not feasible. Also the sequences of type $AB$ and $ABAB$ are considered to be equal. After eliminating these recurring UPOs the modified Redfield-Pólya enumeration theorem provides the number of unique UPOs for the Lorenz attractor.}
    \label{tab:CountUPOs}
\end{table}
\begin{figure}[t]
    \centering
    \subfloat[Tree Hierarchy]{\includegraphics[width=0.46\textwidth]{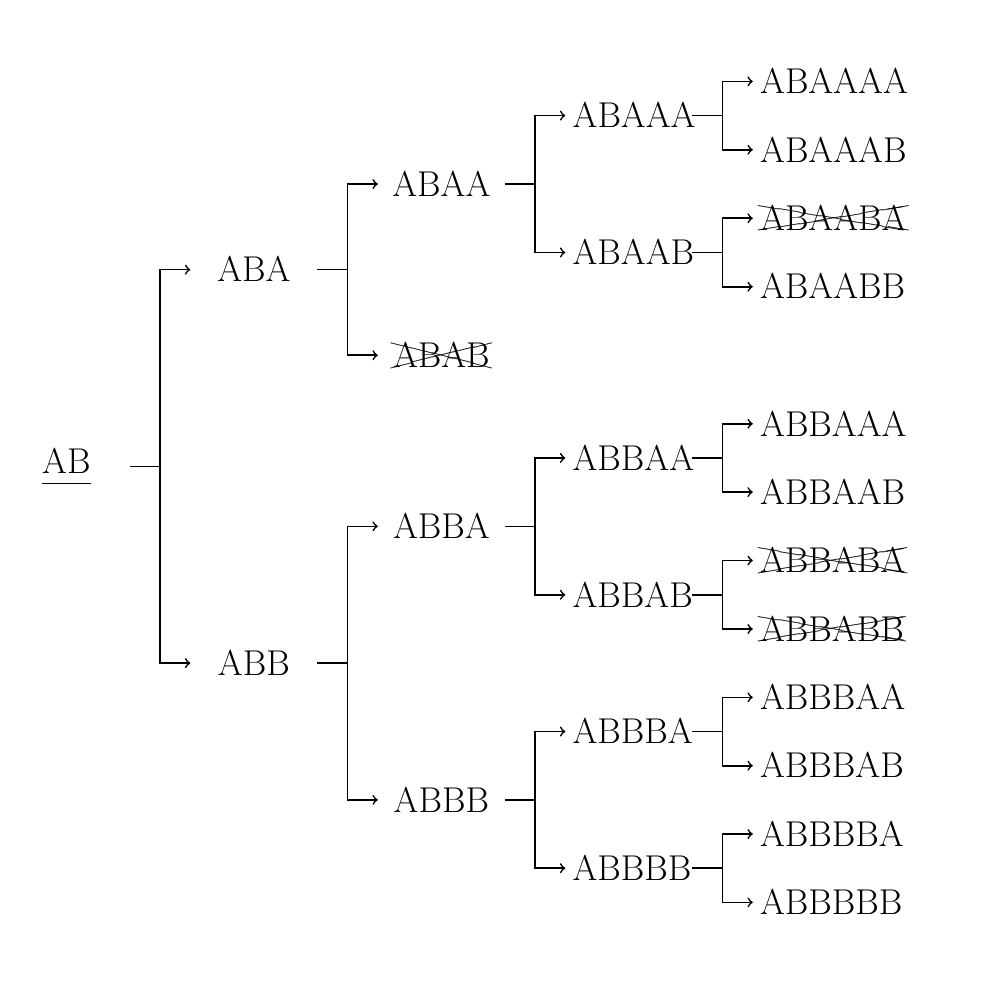}}
    \subfloat[Number of unique UPOs for a given sequence length]{\includegraphics[width=0.53\textwidth]{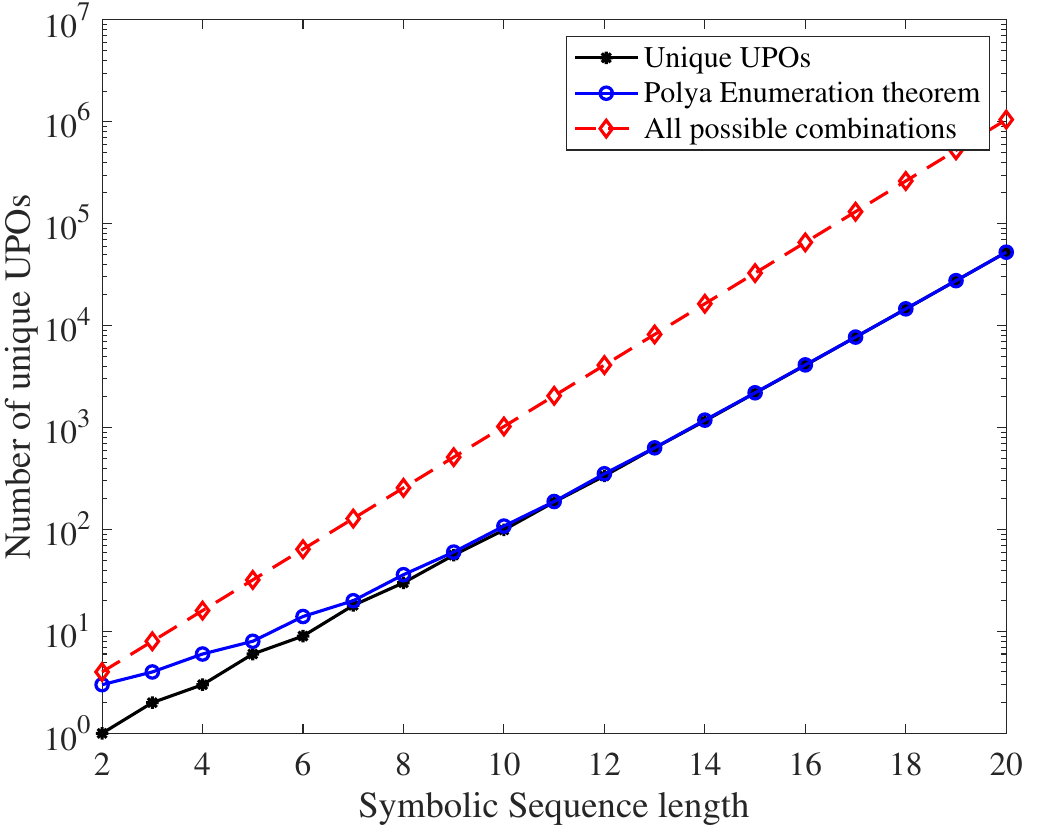}}
    \caption{(a) Hierarchical clustering of the UPOs where repeated UPOs are eliminated from the tree. For example, the UPOs $AB$ and $ABAB$ are the same and so $ABAB$ is eliminated from the tree. (b) This figure shows the comparison between the number of unique UPOs obtained for a given symbolic sequence length. All possible values are computed using the linear permutations which can be given as $2^{N}$, where $N$ is the sequence length. The Redfield-Pólya enumeration theorem considers the circular permutations and eliminates the values which are repeated due to the rotation symmetry. Finally the number of unique UPOs are computed by eliminating the repeated UPOs such as $AB$ and $ABAB$.}
    \label{fig:NumUPOs}
\end{figure}

%% file: LorenzAttractorv2.tex
\subsection{Lorenz attractor}
The initial conditions for the chaotic trajectory are taken to be $[x(t_0),y(t_0),z(t_0)]=[11,14,27]$. The $\Delta t$ for the evolution of the ODE is $0.01$, the value of $\tau$, which is value of time interval between the consecutive columns of Hankel matrix is also taken to be $0.01$, and $\Delta t=\tau$.
\subsubsection{Case I: Unstable periodic orbits of type \texorpdfstring{$\mathbf{A^n B}$ and $\mathbf{B^nA}$}{An B and BnA}}
The data for this case is taken from the \textit{Divakar dataset} where the UPOs of the type $A^nB$ and their symmetric counterparts $AB^n$ are considered where $n=1,2,3,4,\cdots,23$. This dataset includes $45$ unstable periodic orbits: 23 UPOs each from $A^nB$ and $AB^n$. The redundant UPO, $AB$ present in both sets was removed, resulting in 45 unique UPOs. The results of the long time delay embeddings are obtained by varying the $t_{height}$ of the Hankel matrices. The Hankel matrices are computed for $t_{height} = [0.5,1,3,5,10,15]$ and a constant value for the number of columns of the Hankel matrix i.e., $t_{width}=100$ is considered. The results are compiled in the Fig.~(\ref{fig:Case1LorenzUPOUnfolding}a-f). For $t_{height} =[0.5,1,3,5]$ we observe that the UPOs overlap. As longer embeddings are taken into account the chaotic attractor starts to untangle and the corresponding  UPOs start to separate. For longer and longer embeddings the UPOs become circles in the embedded space where they converge to the Fourier modes. As observed in Fig.~(\ref{fig:Case1LorenzUPOUnfolding}d-f), $A^nB$ and $AB^n$ sets of orbits cluster around the diametric opposite ends of the unfolded chaotic trajectory. 

As the UPOs separate in the embedded space, spatial proximity of two UPOs in the $(x,y,z)$-phase space corresponds to spatial proximity in the embedded space. Fig.~(\ref{fig:UPOProximity}a) demonstrates the proximity of the UPOs in the phase space. It is observed that the trajectories of UPOs $AB^{23}$ and $AB^{17}$ stay close to each other. This trend is replicated with the trajectories in the embedded space. We observe in Fig.~(\ref{fig:UPOProximity}b) that after the attractor unfolds the two UPOs $AB^{23}$ and $AB^{17}$ still maintain proximity in the embedded space. Whereas, the UPO $AB$ shares no common trajectory with the other two UPOs and hence we see the separation between $AB$ and the other two UPOs in the embedded space in fact the separation is amplified in the embedded space. The distance between the UPOs or the distance of the UPOs from the plane of separation $v_1=0$ is given by Theorem(\ref{th:SeparationOfUPOs}). In case of UPO $AB^{23}$ and $AB^{17}$ the value of the function $f(\rho)=\frac{\rho^{(k)}-1}{\rho^{(k)}+1}$ is $0.91667$ and $0.8889$ where, $\rho$ refers to the ratio of the A-symbols vs B-symbols for a given UPO of the Lorenz attractor. Whereas, the value of $f(\rho)$ for $AB$ is $f(1)=0$. Thus, we can observe that the UPOs $AB^{23}$ and $AB^{17}$ in the embedded space are closer to each other than the UPO $AB$. 
We observe that $\lim_{\rho \to \infty} f(\rho) = 1$ and $\lim_{\rho \to \infty} f(1/\rho) = -1$. We expect the $A$-heavy and $B$-heavy UPO to lie on the opposite ends of the attractor in the delay-embedded space.  

\begin{figure}
\setkeys{Gin}{width=\linewidth}
    \begin{minipage}{0.93\textwidth}
\hspace*{-0.02\linewidth}\begin{subfigure}[b]{0.36\linewidth}
  \includegraphics{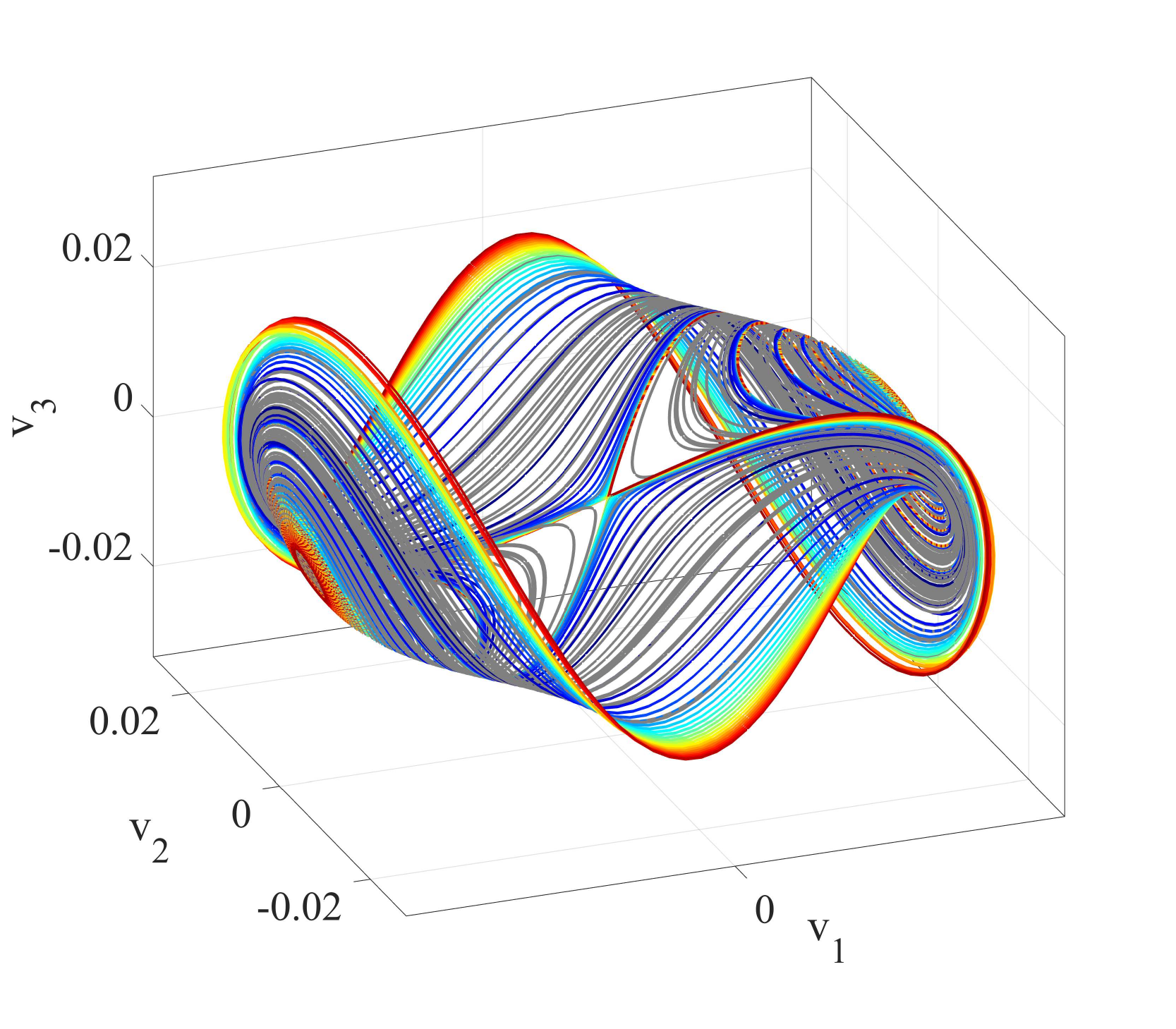}
  \caption{$t_{height}=0.5$}
  \label{fig:Case1Fig1}
\end{subfigure}
\hfill
\hspace*{-0.04\linewidth}\begin{subfigure}[b]{0.36\linewidth}
  \includegraphics{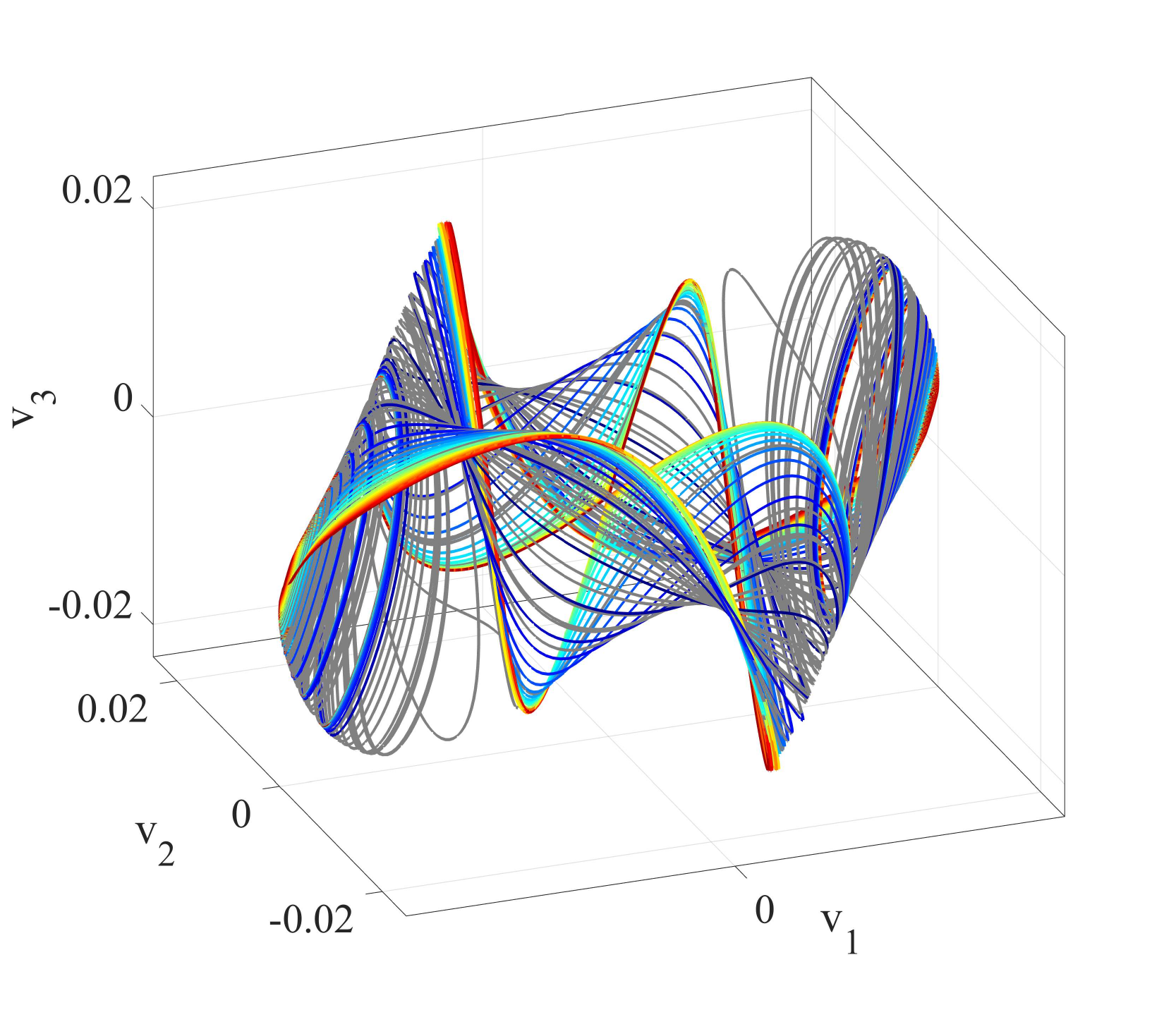}
  \caption{$t_{height}=1$}
  \label{fig:Case1Fig2}
\end{subfigure}
\hfill
\hspace*{-0.04\linewidth}\begin{subfigure}[b]{0.36\linewidth}
  \includegraphics{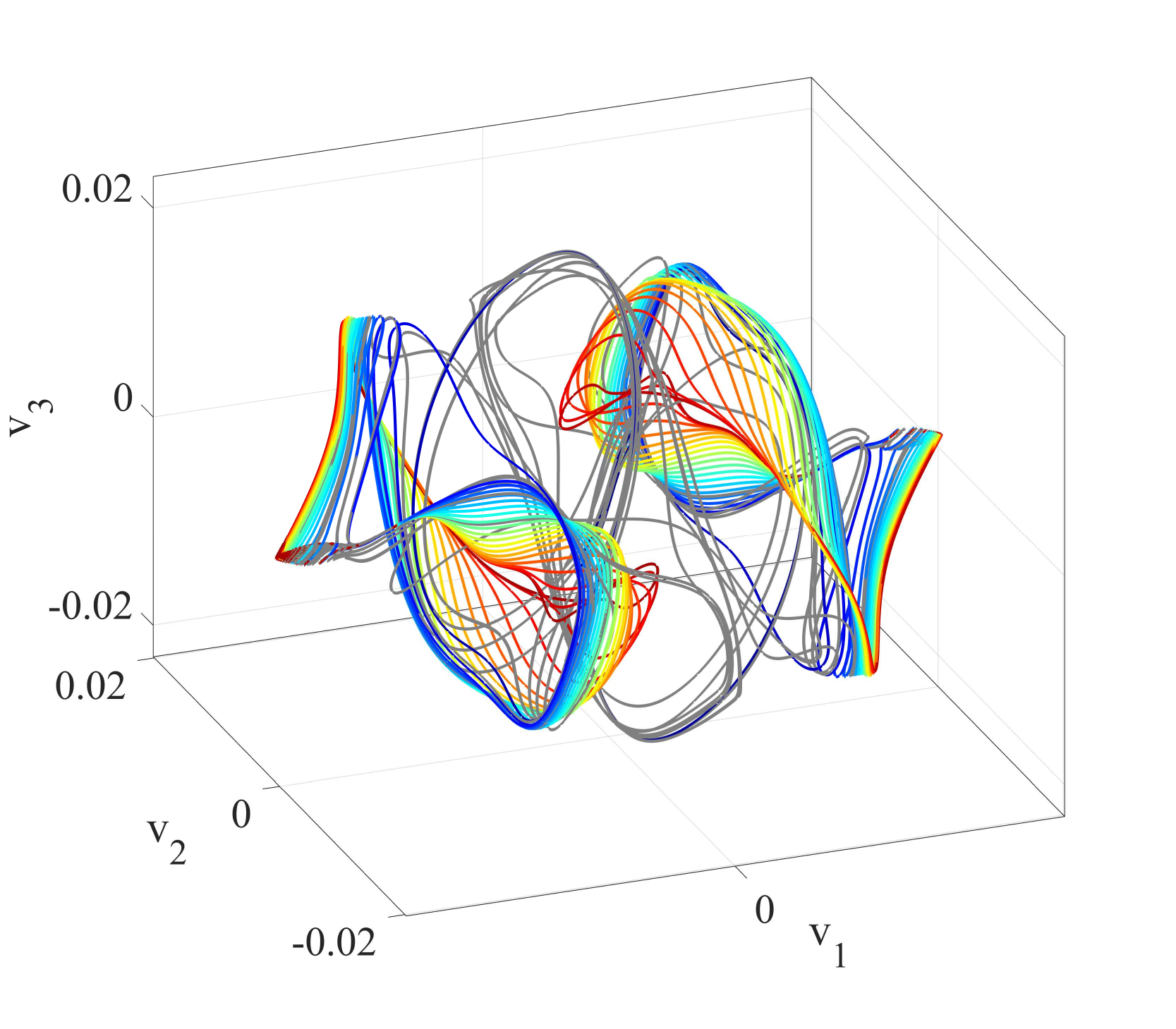}
  \caption{$t_{height}=3$}
  \label{fig:Case1Fig3}
\end{subfigure}
\hfill
\hspace*{-0.02\linewidth}\begin{subfigure}[b]{0.36\linewidth}
  \includegraphics{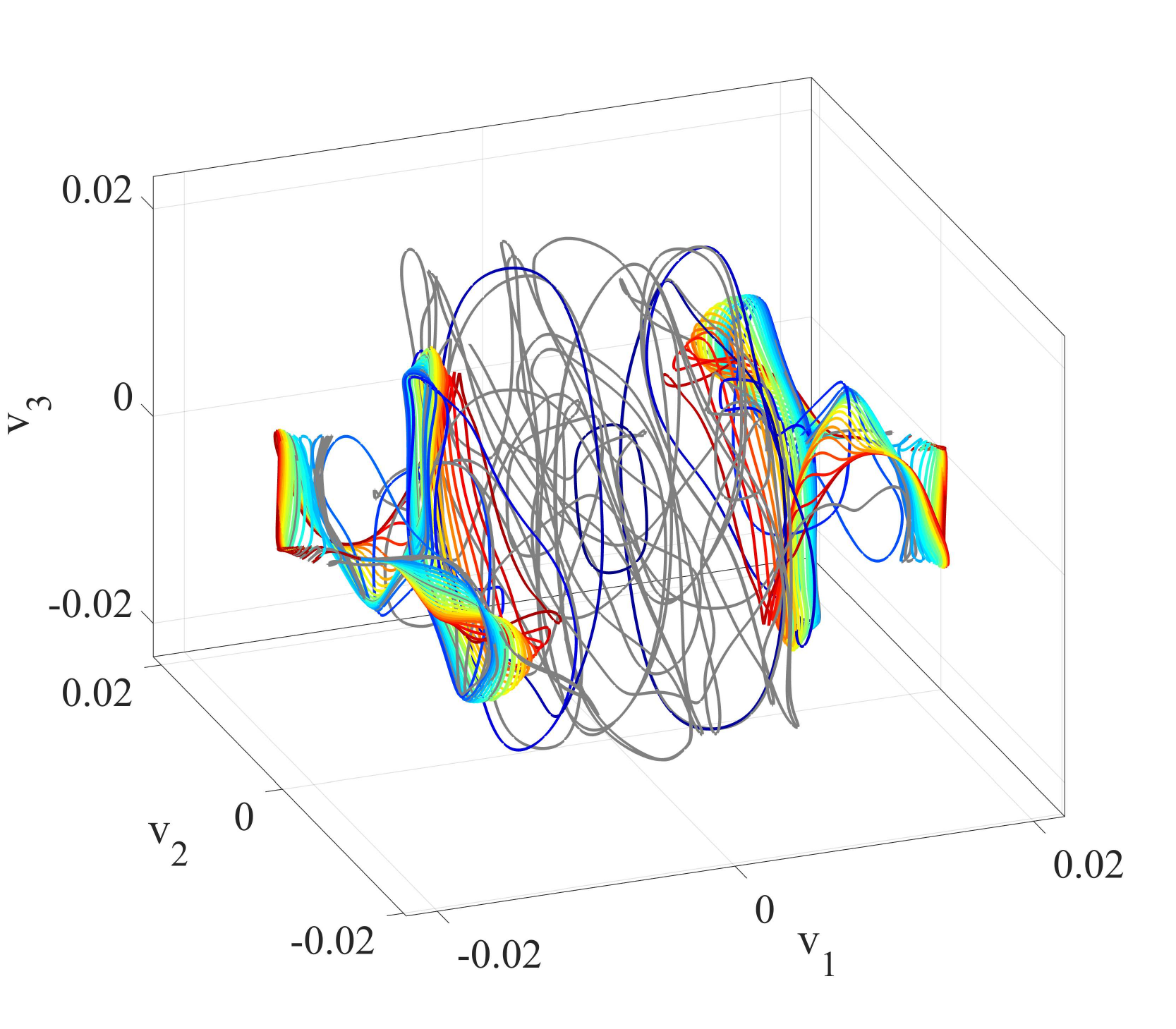}
  \caption{$t_{height}=5$}
  \label{fig:Case1Fig4}
\end{subfigure}
\hfill
\hspace*{-0.04\linewidth}\begin{subfigure}[b]{0.36\linewidth}
  \includegraphics{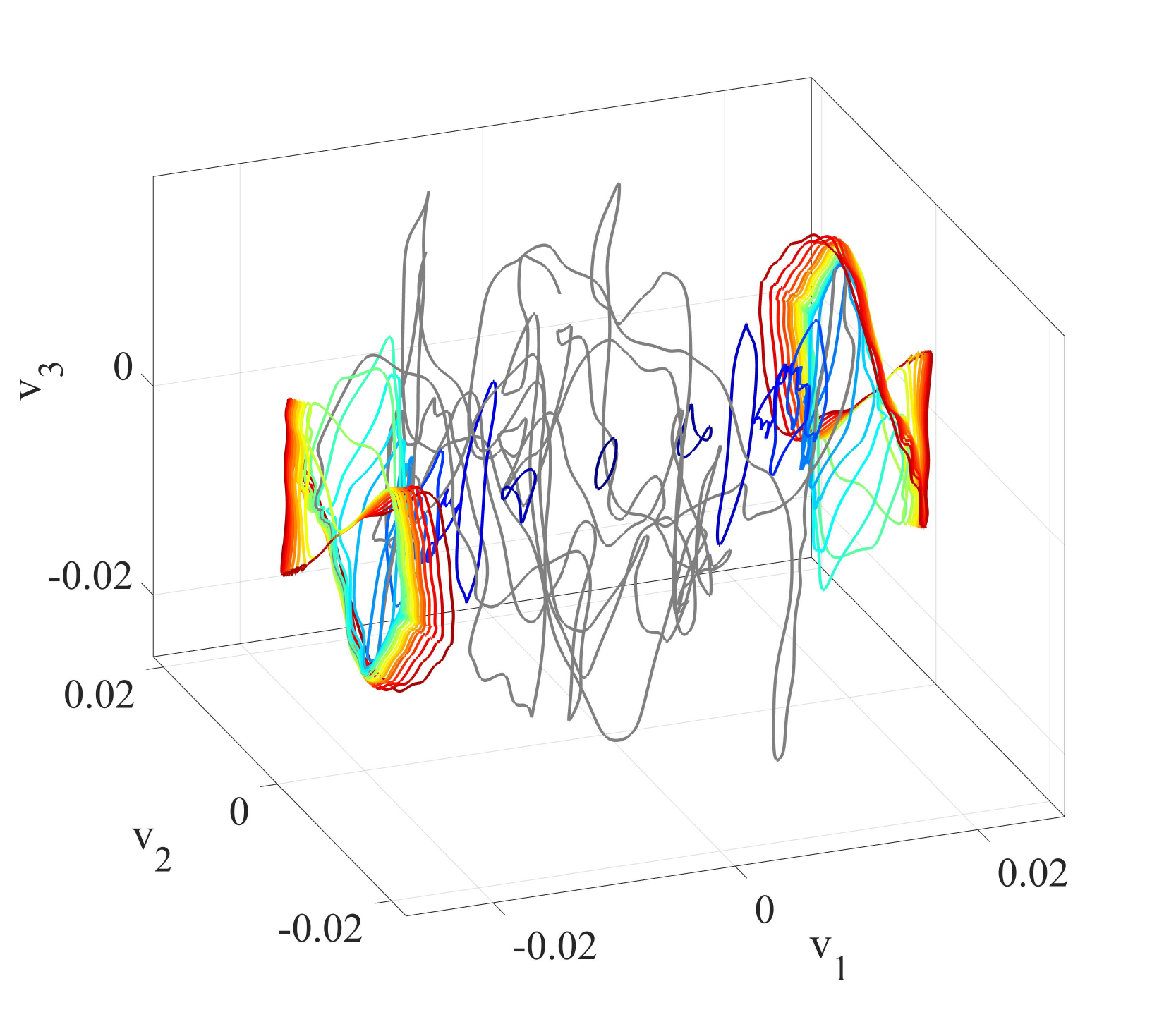}
  \caption{$t_{height}=10$}
  \label{fig:legendcase1}
\end{subfigure}
\hfill
\hspace*{-0.04\linewidth}\begin{subfigure}[b]{0.36\linewidth}
  \includegraphics{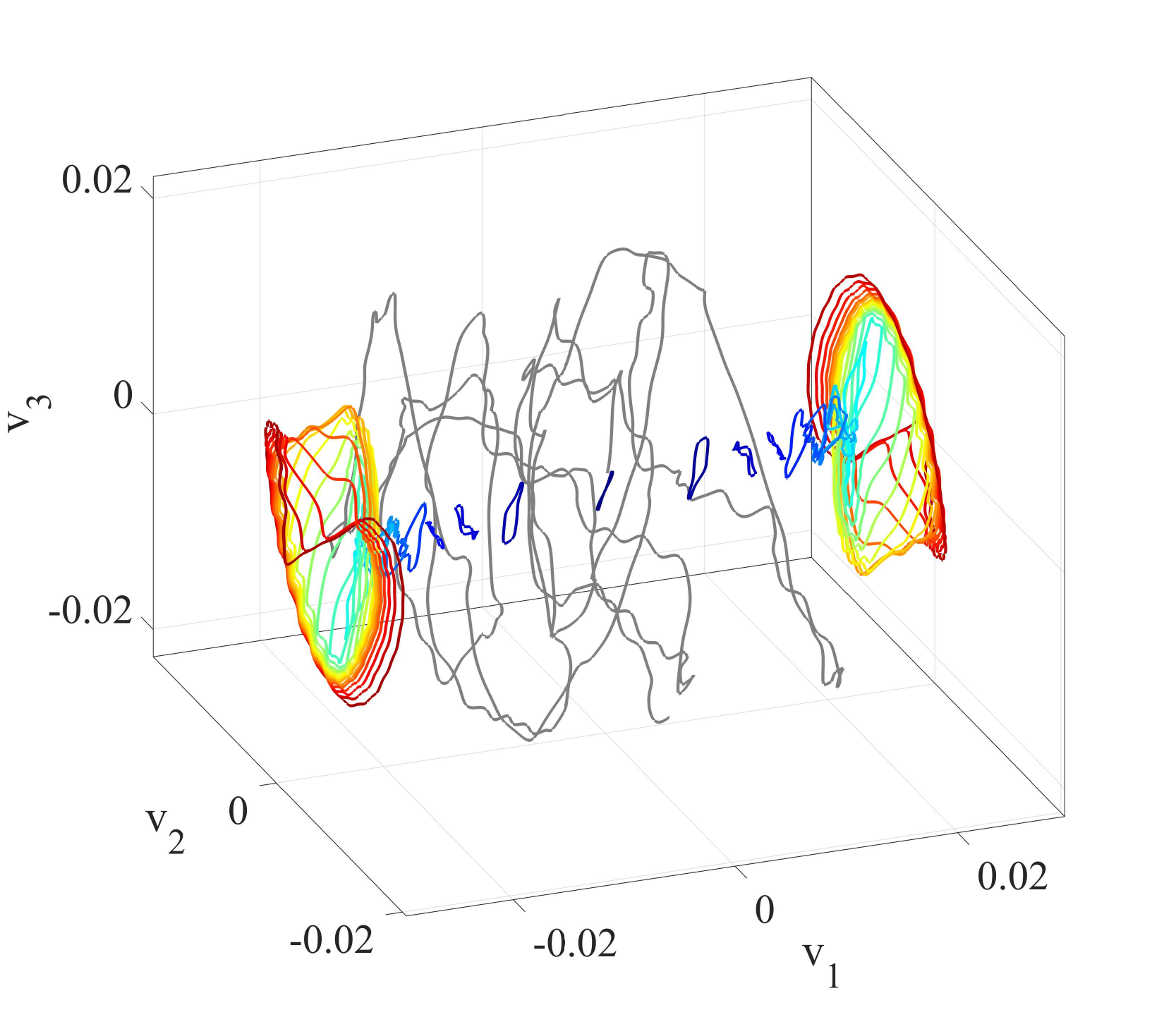}
  \caption{$t_{height}=15$}
  \label{fig:AnBBnALTDEUPOs}
\end{subfigure}
\hfill
    \end{minipage}
\hfill
\hspace{-0.01\linewidth}\begin{minipage}{0.07\textwidth}
\includegraphics{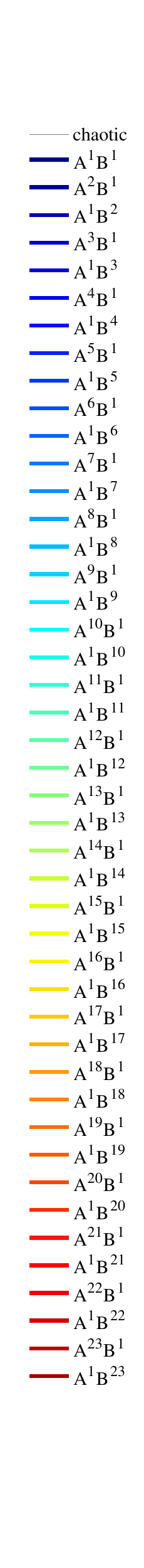}
\end{minipage}
\caption{Lorenz attractor: Unfolding of the attractor for the unstable periodic orbits of the type $A^nB$ and $B^nA$. We observe that the orbits cluster around the diametric opposite ends of the unfolded chaotic trajectory. The grey represents the chaotic trajectory.}
\label{fig:Case1LorenzUPOUnfolding}
\end{figure}
\begin{figure}[H]
    \centering
    \subfloat[Three UPOs in the phase space]{\includegraphics[width=0.35\textwidth]{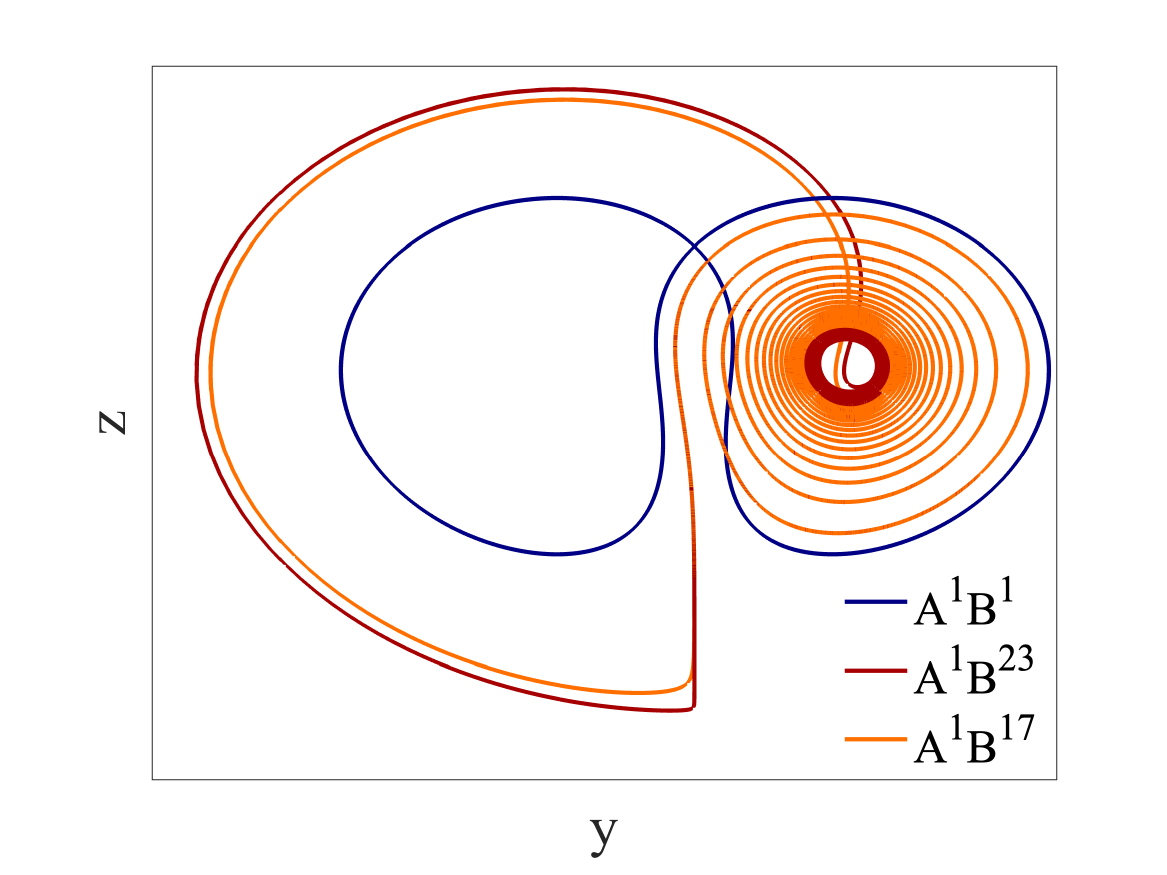}}
    \subfloat[Three UPOs in the embedded space]{\includegraphics[width=0.35\textwidth]{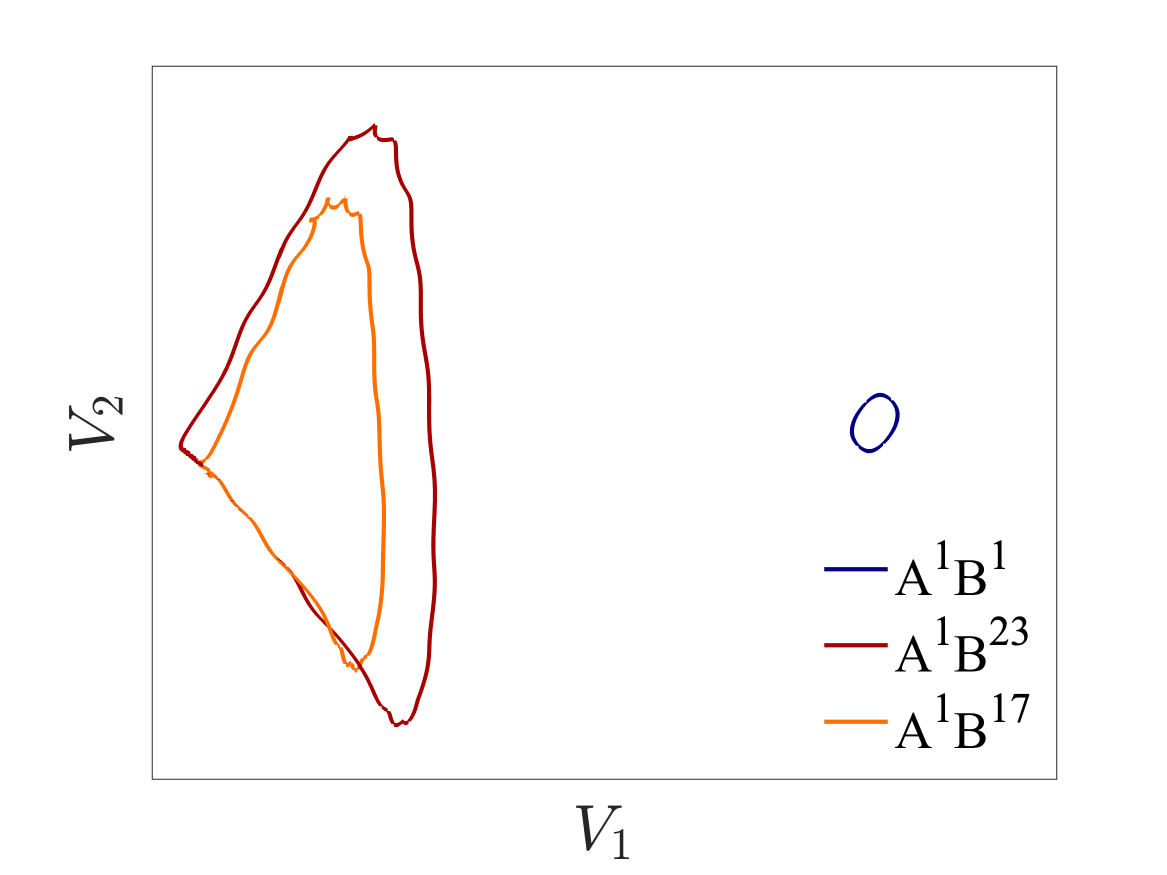}}
    \caption{Comparison between the proximity of the UPOs in the phase space and the embedded space for the Lorenz attractor. }
    \label{fig:UPOProximity}
\end{figure}
\subsubsection{Case II: Unstable periodic orbits of type \texorpdfstring{$\mathbf{A^n B^n}$}{An Bn}}

The data for this case is taken from the \textit{Divakar dataset} where the symmetric UPOs of the type $A^nB^n$ are considered where $n=1,2,\cdots,19$. This data-set contains 19 unstable periodic orbits. 
The results of the long time delay embeddings are obtained by varying the $t_{height}$ of the Hankel matrices. The Hankel matrices are computed for $t_{height}=[0.5,1,3,5,10,15]$ and $t_{width}=100$. The results for this case are compiled in Fig.~(\ref{fig:AnBnLTDEUPOs}a-f). 
For $t_{height}=[0.5,1,3,5]$ we observe that the UPOs overlap. As longer embeddings are taken into account the chaotic attractor starts to untangle the UPOs separating them in the process. The $A^nB^n$ orbits cluster at the center of the unfolded chaotic trajectory. Similar to the previous case, proximity in the phase space translated to proximity in the embedded space of the attractor. 
The relative positions of the two sets of UPOs (Case I and Case II) is depicted in Fig.~(\ref{fig:AHeavyBHeavyUPOs}a). In this figure, we observe that the UPOs with higher number of $A$ symbols gravitate towards one side of the unfolded attractor, whereas the UPOs higher number of $B$ symbols gravitate towards diametrically opposite side of the unfolded attractor. In comparison, the symmetric UPOs with equal number of $A$ and $B$ components occupy the center of the unfolded chaotic attractor. 
\begin{figure}
\setkeys{Gin}{width=\linewidth}
    \begin{minipage}{0.925\textwidth}
\hspace*{-0.02\linewidth}\begin{subfigure}[b]{0.36\linewidth}
  \includegraphics{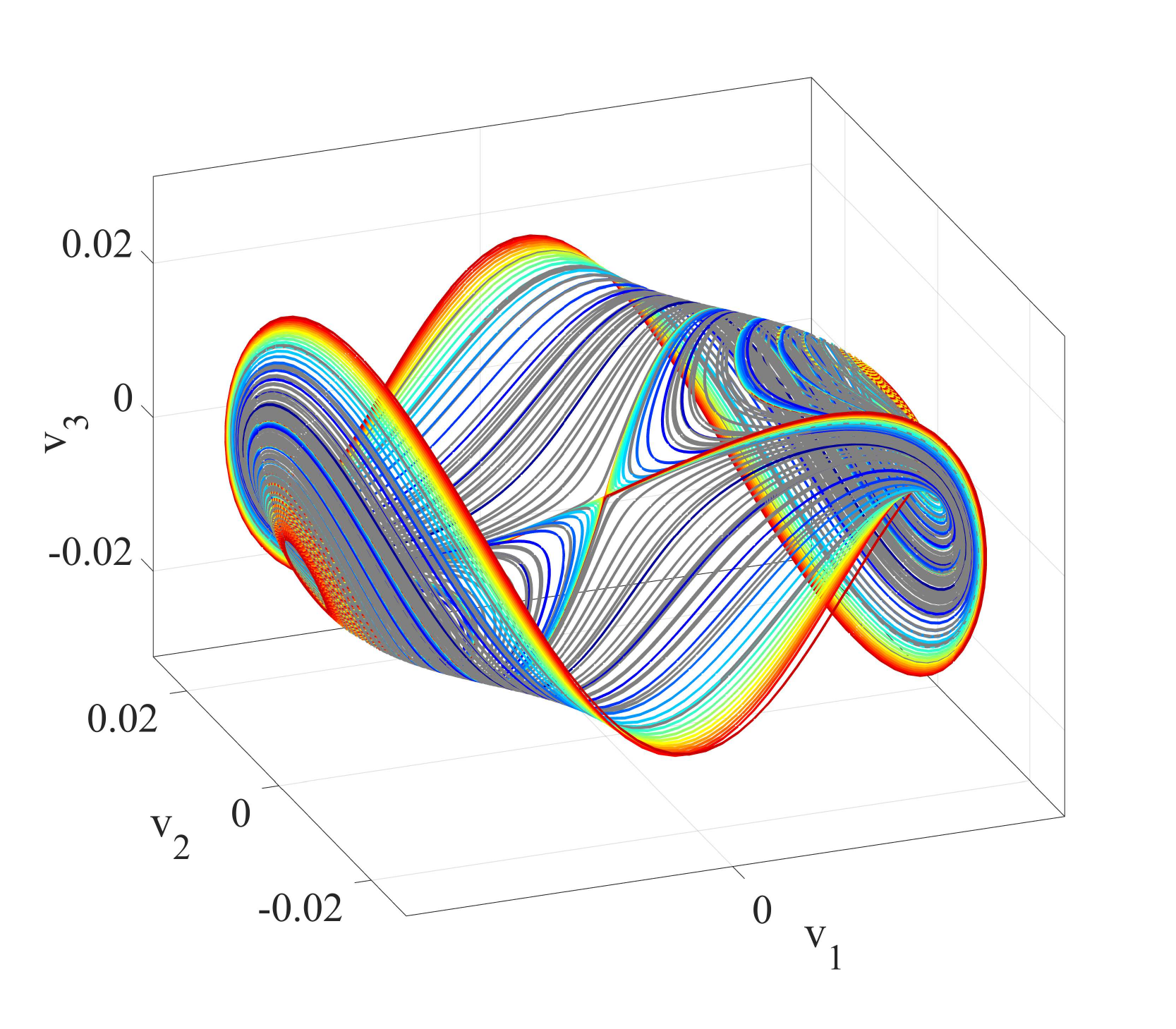}
  \caption{$t_{height}=0.5$}
  \label{fig:Case2Fig1}
\end{subfigure}
\hfill
\hspace*{-0.04\linewidth}\begin{subfigure}[b]{0.36\linewidth}
  \includegraphics{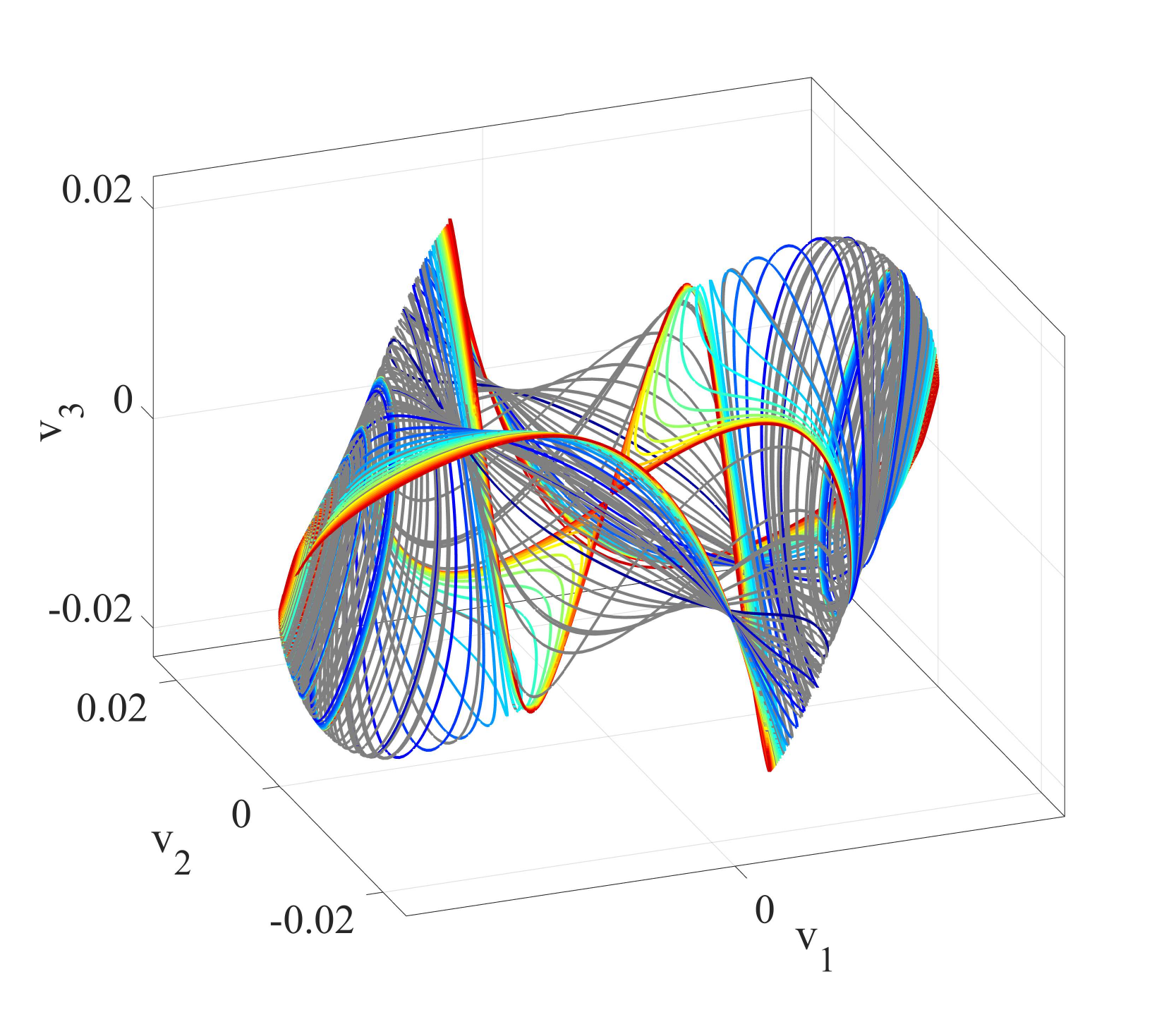}
  \caption{$t_{height}=1$}
  \label{fig:Case2Fig2}
\end{subfigure}
\medskip
\hspace*{-0.04\linewidth}\begin{subfigure}[b]{0.36\linewidth}
  \includegraphics{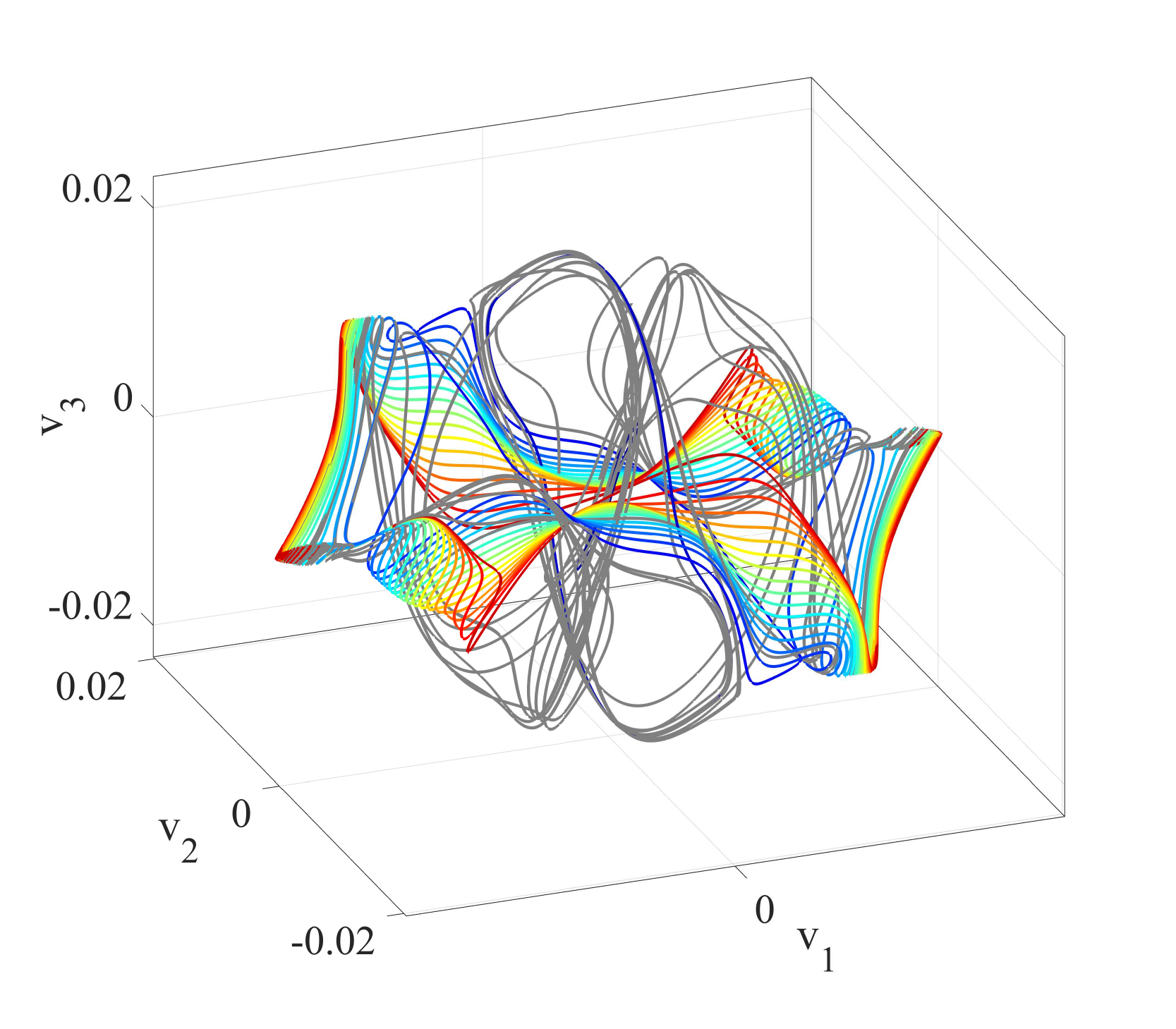}
  \caption{$t_{height}=3$}
  \label{fig:Case2Fig3}
\end{subfigure}
\hfill
\hspace*{-0.02\linewidth}\begin{subfigure}[b]{0.36\linewidth}
  \includegraphics{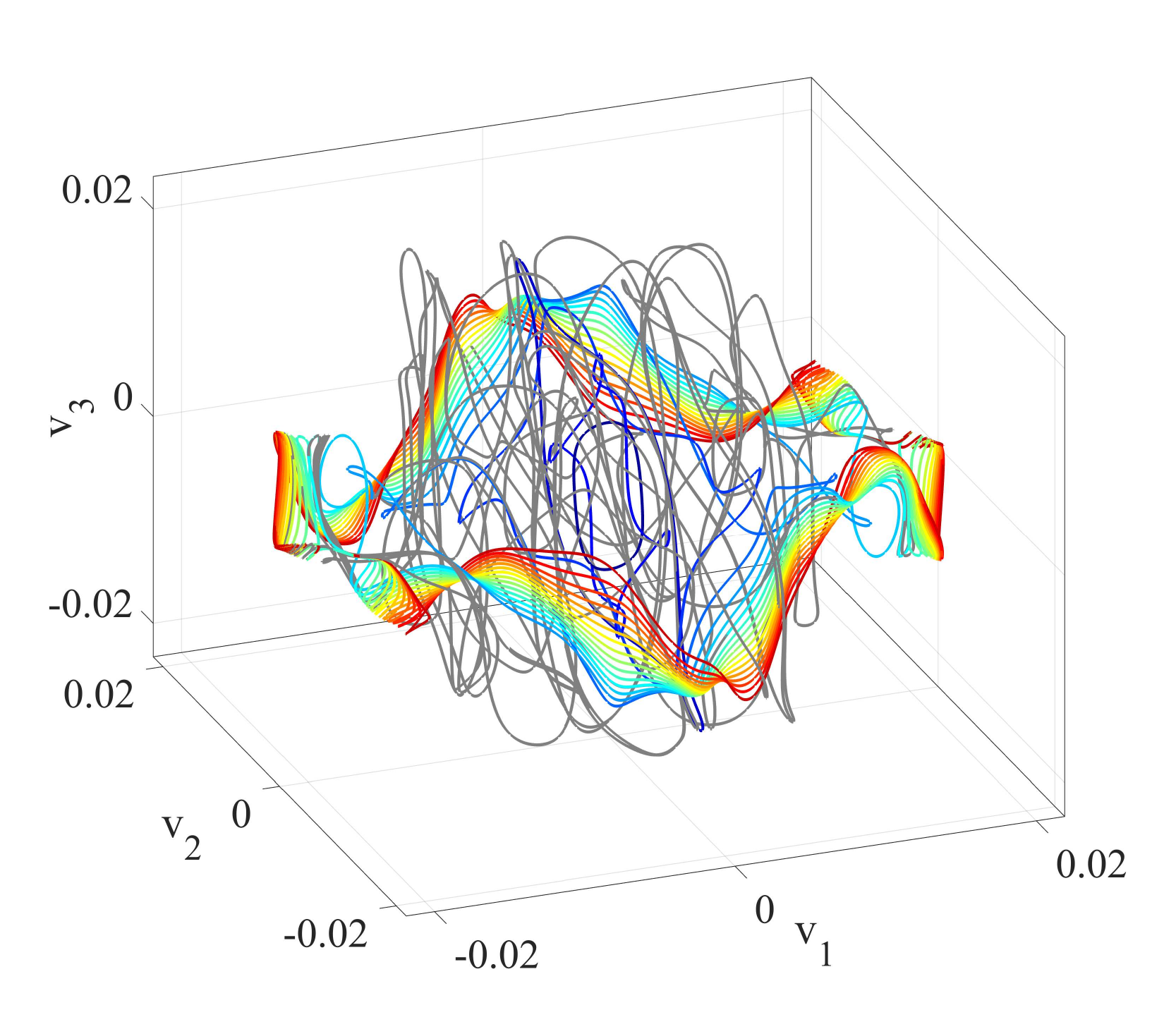}
  \caption{$t_{height}=5$}
  \label{fig:Case2Fig4a}
\end{subfigure}
\hfill
\hspace*{-0.04\linewidth}\begin{subfigure}[b]{0.36\linewidth}
  \includegraphics{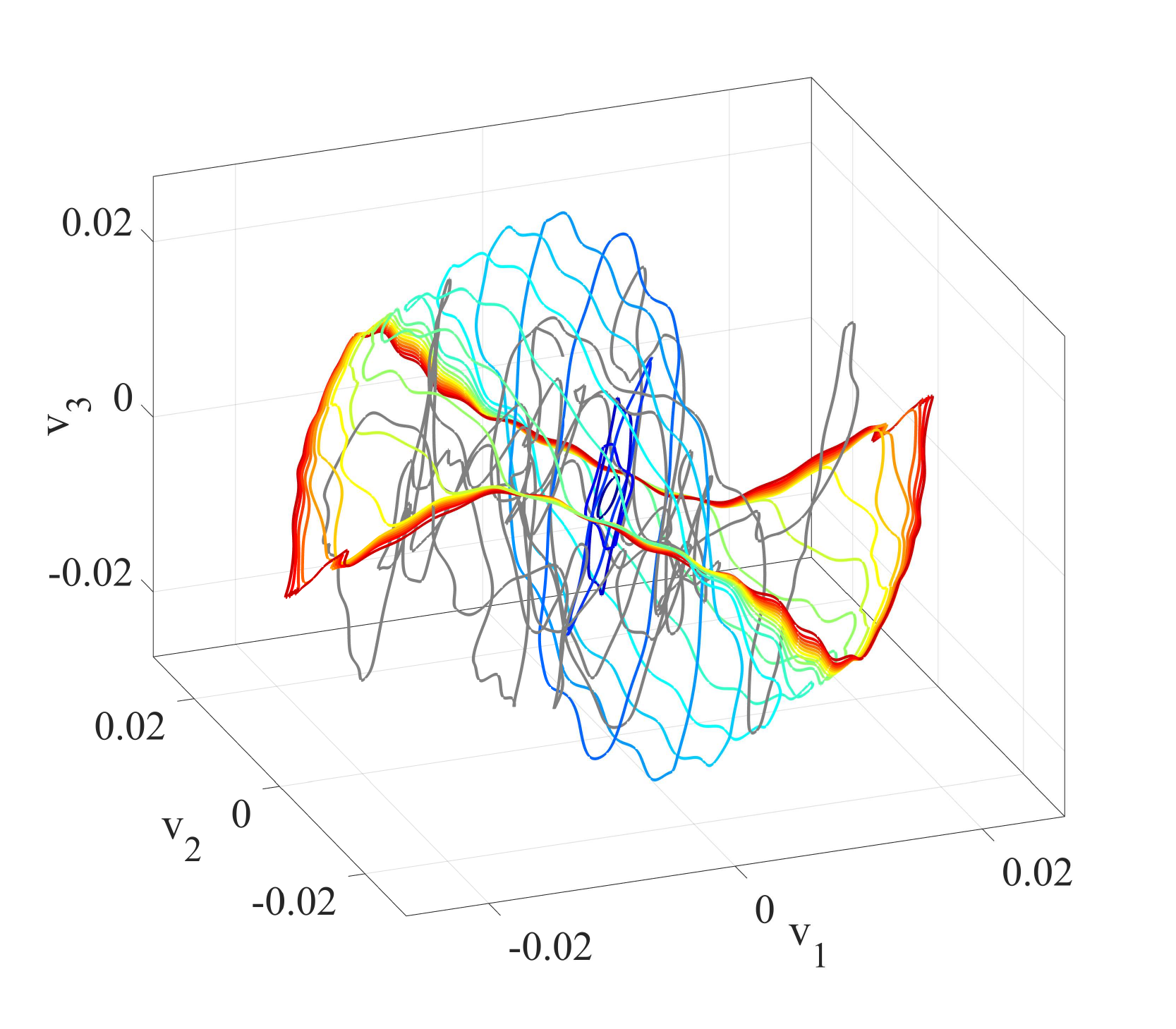}
  \caption{$t_{height}=10$}
  \label{fig:Case2Fig4b}
\end{subfigure}
\hfill
\hspace*{-0.04\linewidth}\begin{subfigure}[b]{0.36\linewidth}
  \includegraphics{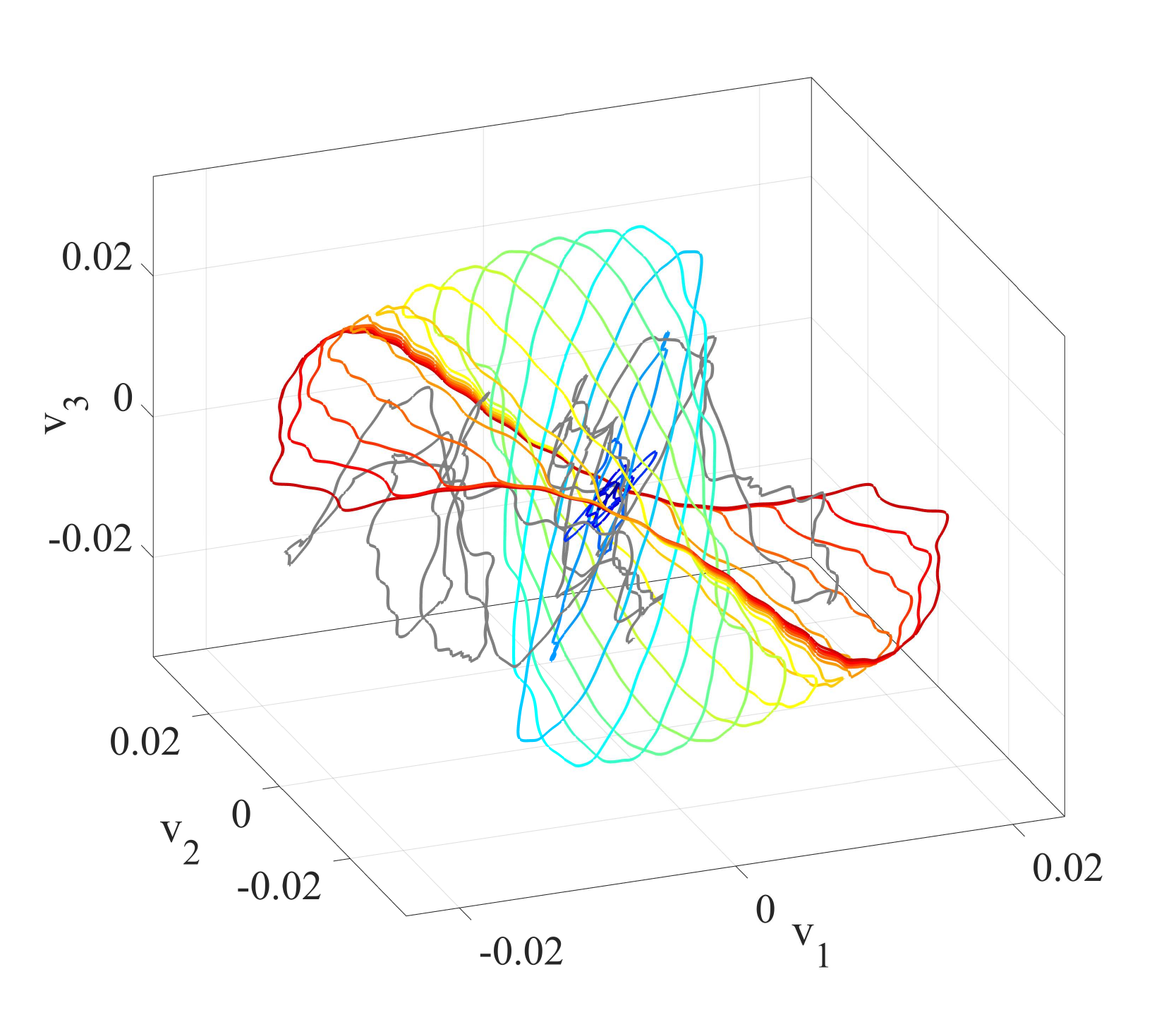}
  \caption{$t_{height}=15$}
  \label{fig:Case2Fig5}
\end{subfigure}
\hfill
    \end{minipage}
\hfill
\hspace{-0.02\linewidth}\begin{minipage}{0.09\textwidth}
\includegraphics{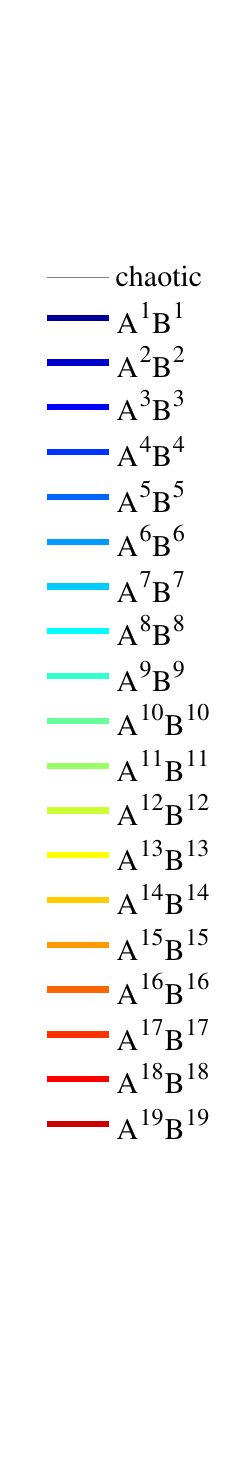}
\end{minipage}
\caption{Unfolding of the attractor for the symmetric unstable periodic orbits of the type $A^nB^n$. We observe that the UPOs cluster around the center of the unfolded chaotic trajectory.}
\label{fig:AnBnLTDEUPOs}
\end{figure}
\subsubsection{Case III: All symbolic sequences of length less than 8}
The data for this case is once again taken from the \textit{Divakar dataset} where all the symbolic sequences for sequence length 13 and less are provided. Summing the unique UPOs obtained for all sequences lengths 13 and less in Table.(\ref{tab:CountUPOs}) gives 1375. However, for the sake of clarity and to avoid overcrowding of the figure, only UPOs of sequence length less than 8 considered. Thus, only 39 UPO are considered. In this case, the separation of UPOs into clusters is observed in Fig~(\ref{fig:SeqLenLTDEUPOs}a-f). The cluster grouping happens based on the ratio of $A$ symbols to $B$ symbols for $A$-heavy UPOs and ratio of $B$ symbols to $A$ symbols for the $B$-heavy UPOs. As observed in the previous case, if the number of $A$ symbols and $B$ symbols are the same in a given UPO it is situated at the center of the unfolded chaotic trajectory. The order of these symbolic sequences does not matter for the arrangement of the UPOs in the delay-embedded space.
\begin{figure}[H]
    \centering
    \subfloat[]{\includegraphics[width=0.49\textwidth]{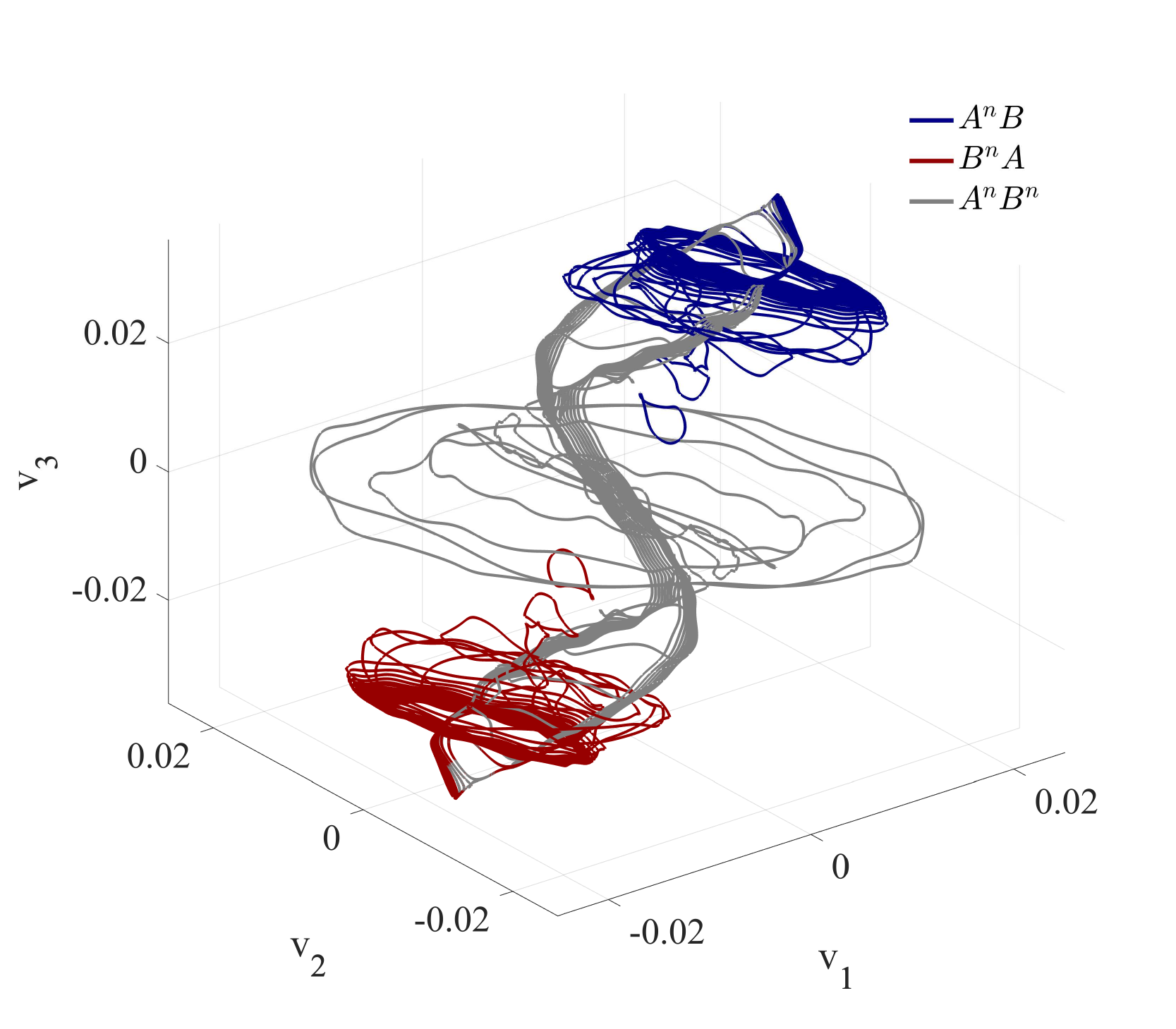}}
    \subfloat[]{
    \includegraphics[width=0.49\textwidth]{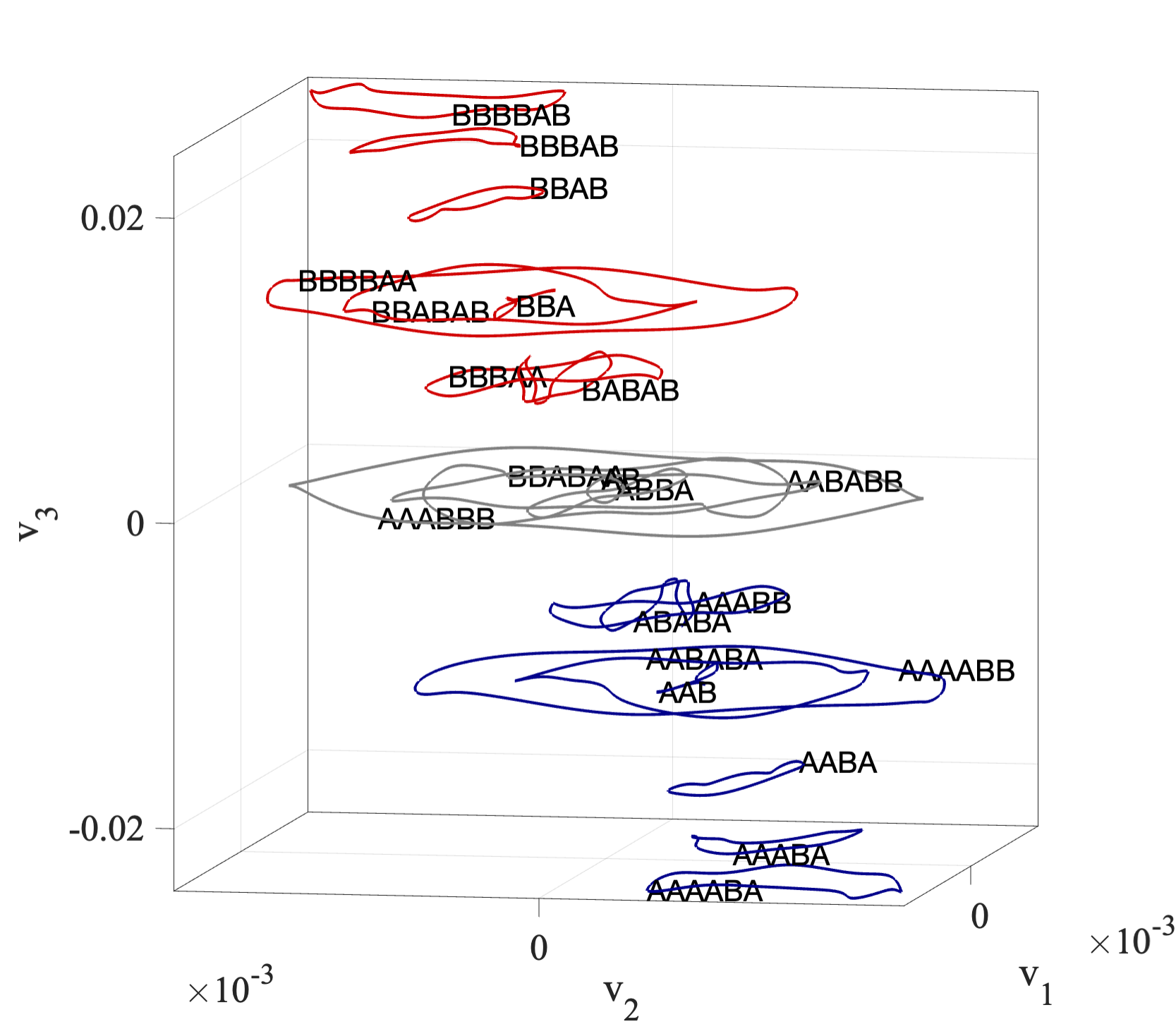}}
    \caption{Lorenz attractor: Clustering of the separated unstable periodic orbits. We observe that the symmetric UPOs cluster around the center whereas the $A$-heavy and $B$-heavy orbits cluster on diametrically opposite ends of the symmetric cluster. }
    \label{fig:AHeavyBHeavyUPOs}
\end{figure}
This is explored further in Fig.~(\ref{fig:AHeavyBHeavyUPOs}b) for all UPOs of sequence length 6 and less. We observe the first cluster at the center of the figure which has the following UPOs: $AB, ABBA, AAABBB$, $AABABB, BBABAA$. We observe that the ratio of $A$:$B$ for all the UPOs in the aforementioned cluster is 1. The second cluster we observe is the cluster on the $A$-heavy side (color coded in blue) adjacent to cluster $1$ which comprises of the UPOs: $ABABA,AAABB$. In this cluster the ratio of $A$:$B$ is $3:2=1.5$. The cluster adjacent to this one is the one comprising of the UPOs: $AABABA, AAAABB, AAB$ where the ratio is $2:1=2$. As we move away from the center the clusters have the ratios $3,4,5$. Similar trend is observed on the $B$-heavy side of the UPOs where the ratio of $B:A$ is taken into consideration. The distance between these clusters reduces as higher and higher ratios are taken into account. The distance of the clusters from the plane of separation $v_3=0$ can be given by the function of the ratios $f(\rho) = \frac{\rho-1}{\rho+1}$, where $\rho$ represents the ratio of the $A$ symbols to $B$ symbols for a given UPO. 
\begin{figure}
\setkeys{Gin}{width=\linewidth}
    \begin{minipage}{0.915\textwidth}
\hspace*{-0.02\linewidth}\begin{subfigure}[b]{0.36\linewidth}
  \includegraphics{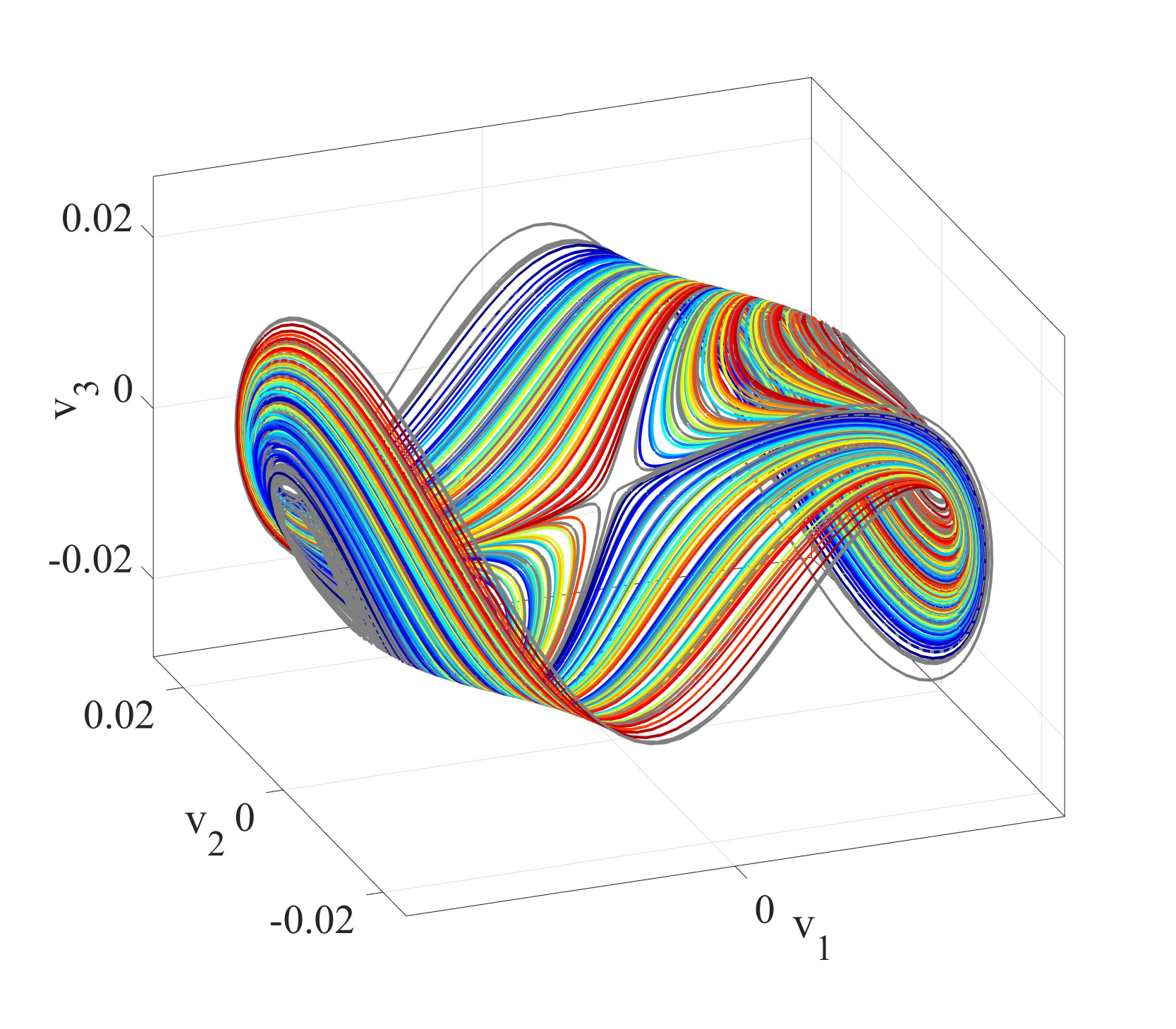}
  \caption{$t_{height}=0.5$}
  \label{fig:Case3Fig1}
\end{subfigure}
\hfill
\hspace*{-0.04\linewidth}\begin{subfigure}[b]{0.36\linewidth}
  \includegraphics{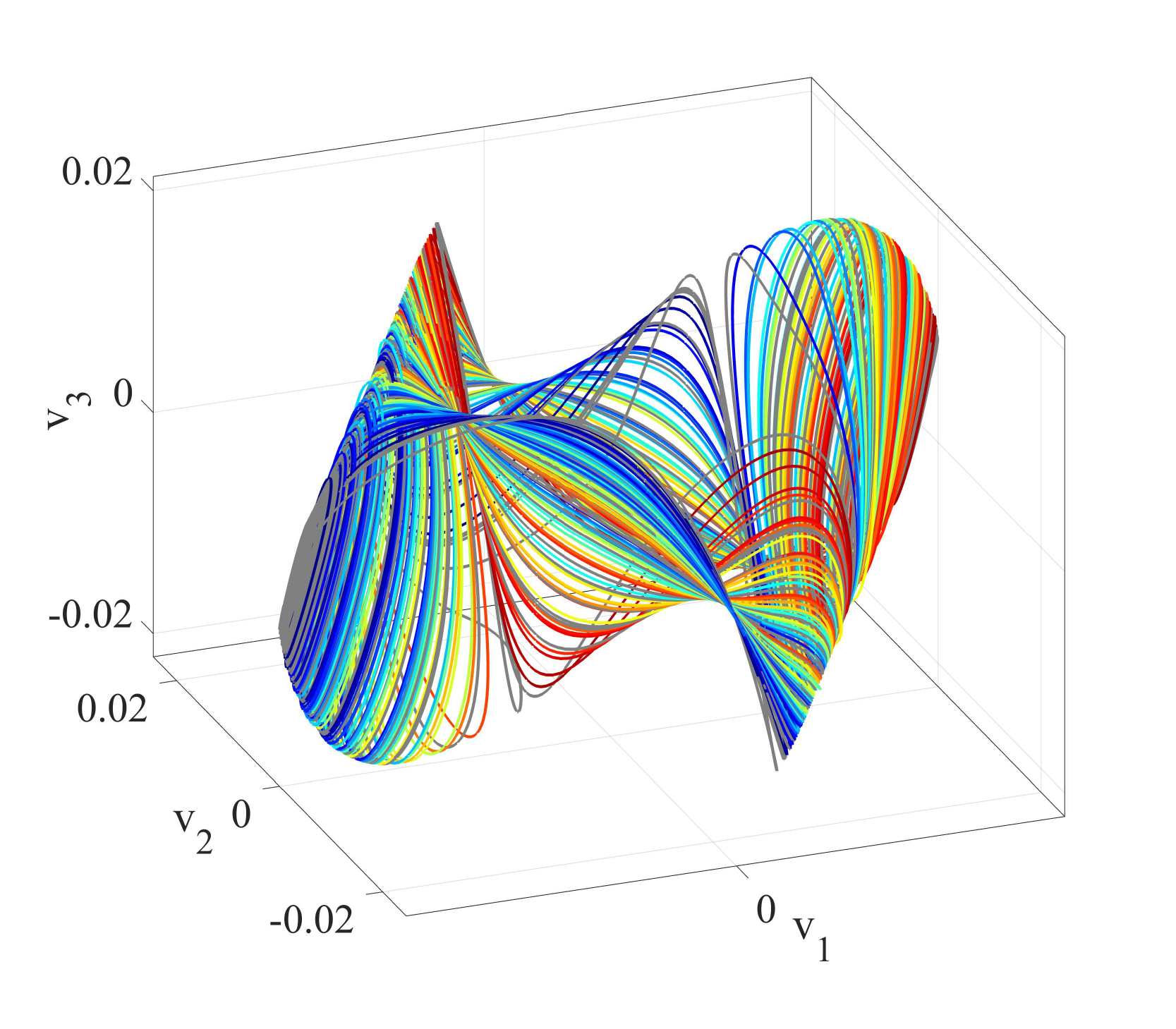}
  \caption{$t_{height}=1$}
  \label{fig:Case3Fig2}
\end{subfigure}
\hfill
\hspace*{-0.04\linewidth}\begin{subfigure}[b]{0.36\linewidth}
  \includegraphics{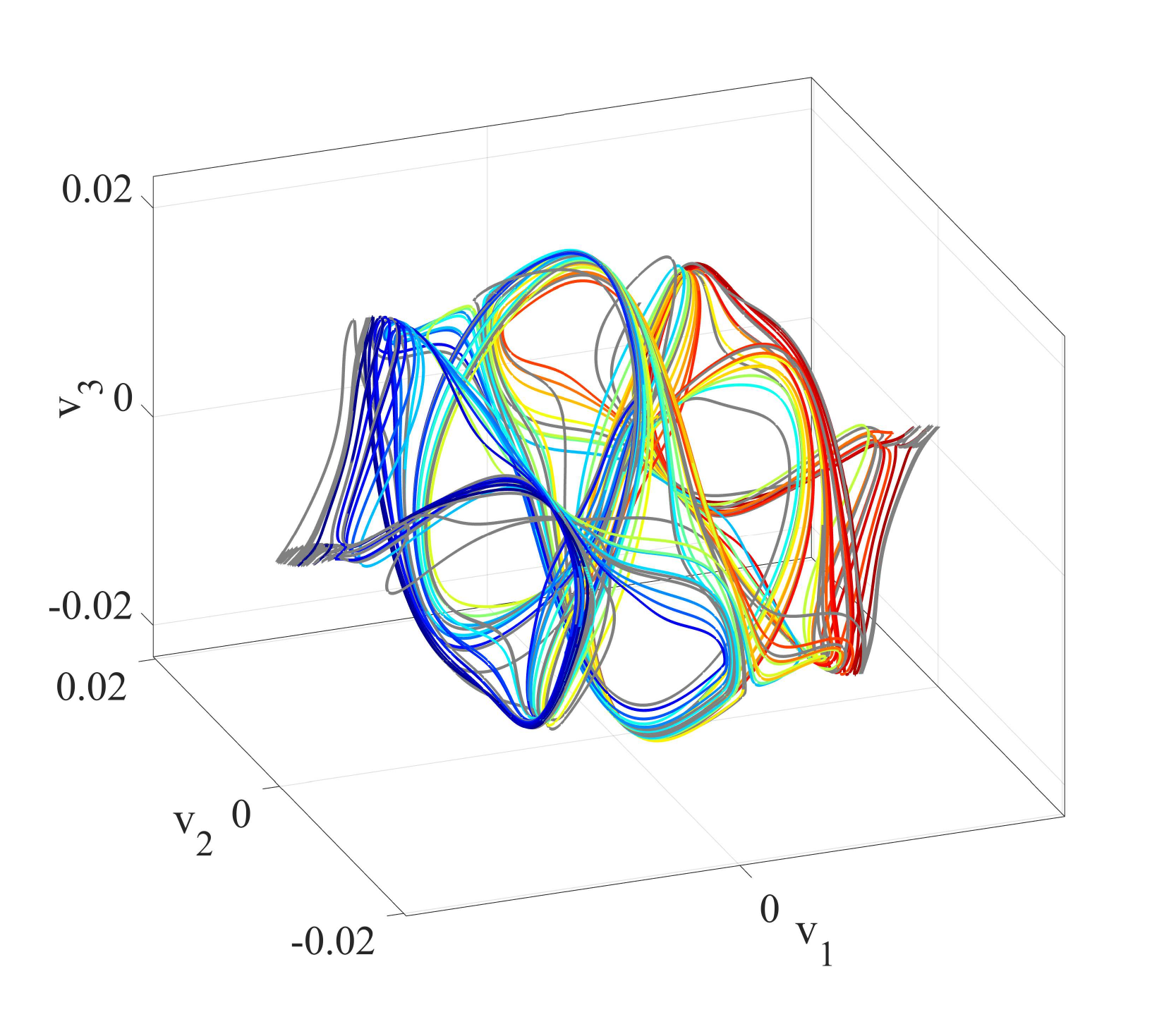}
  \caption{$t_{height}=3$}
  \label{fig:Case3Fig3}
\end{subfigure}
\hfill
\hspace*{-0.02\linewidth}\begin{subfigure}[b]{0.36\linewidth}
  \includegraphics{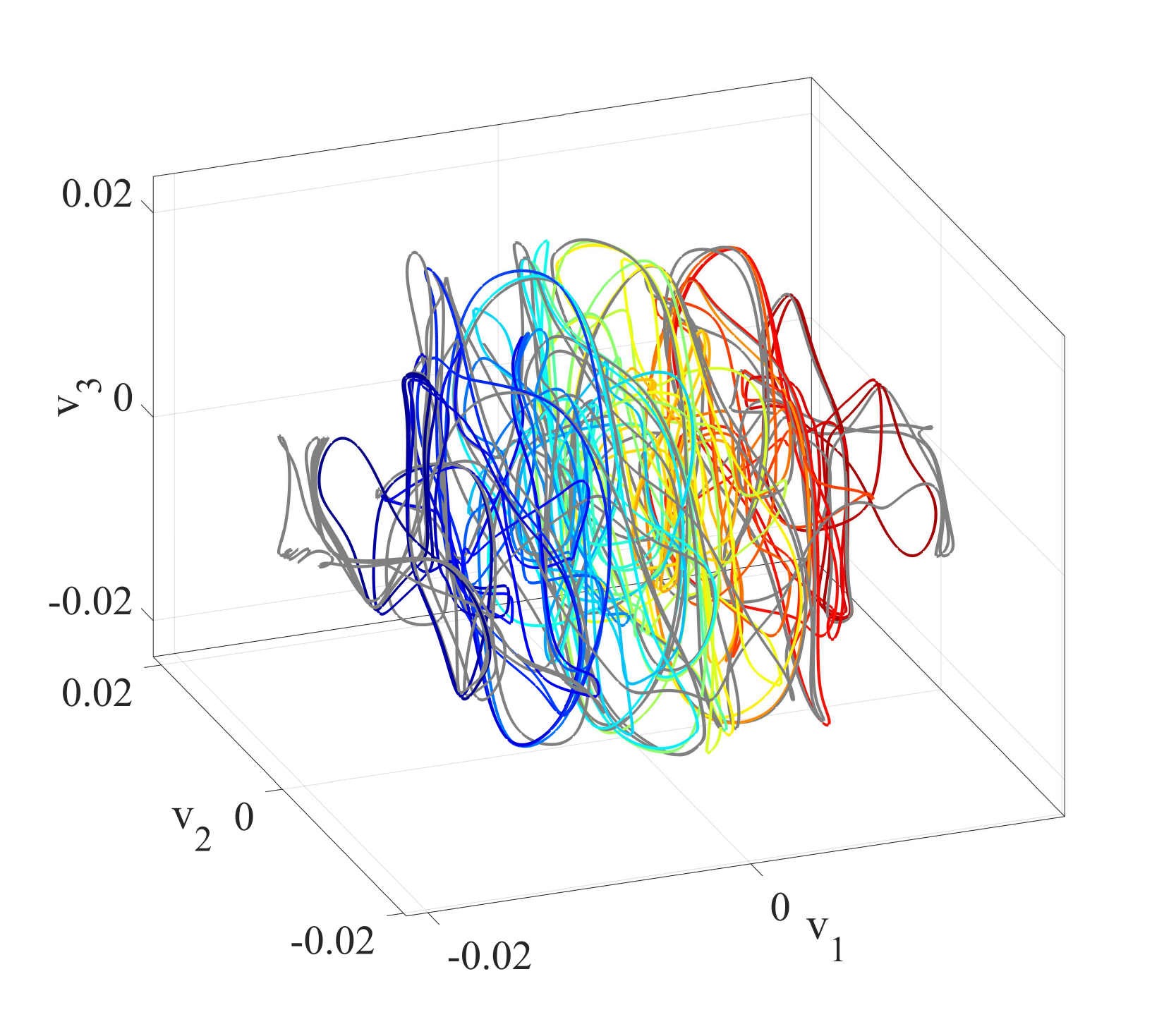}
  \caption{$t_{height}=5$}
  \label{fig:Case3Fig4a}
\end{subfigure}
\hfill
\hspace*{-0.04\linewidth}\begin{subfigure}[b]{0.36\linewidth}
  \includegraphics{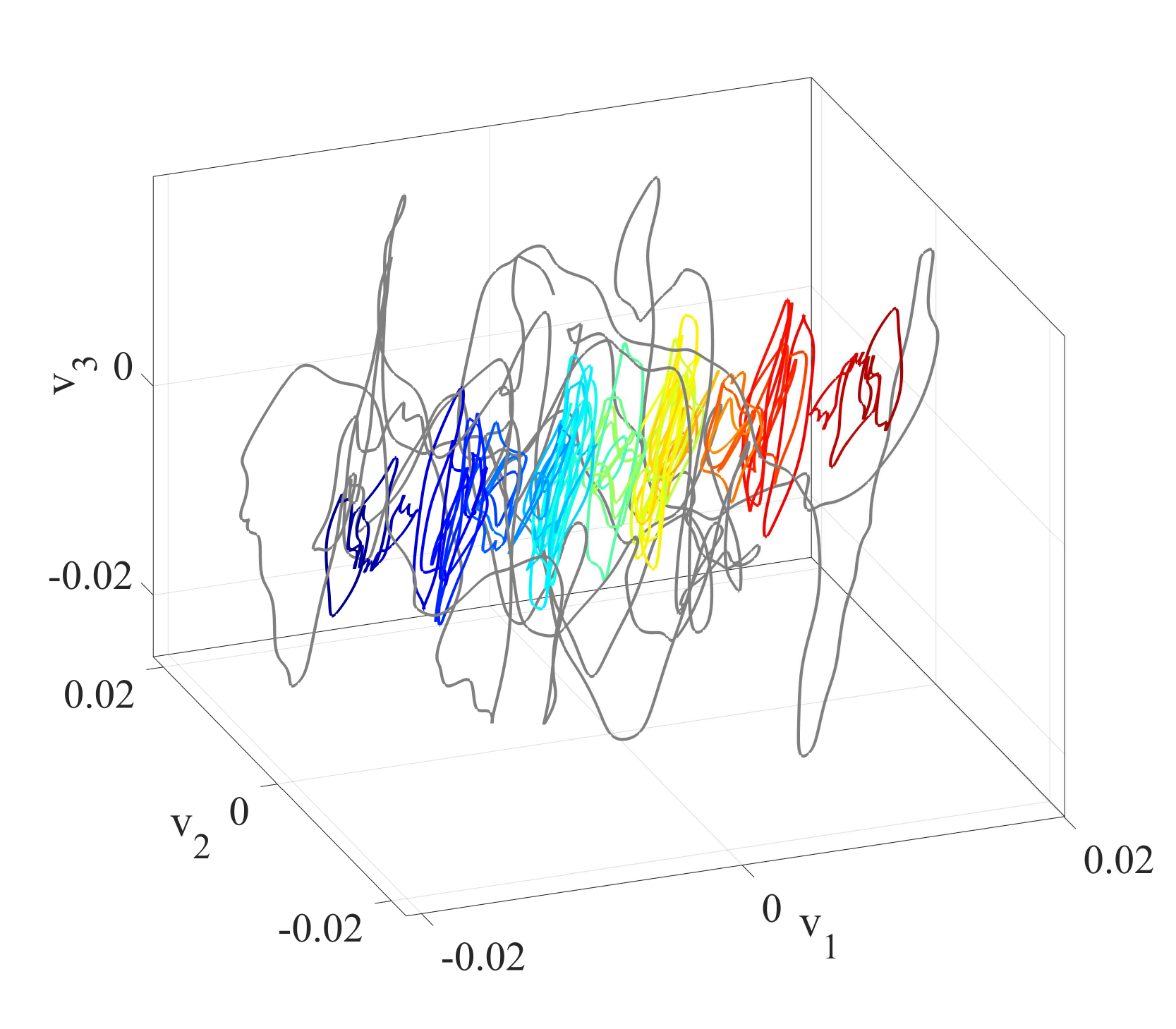}
  \caption{$t_{height}=10$}
  \label{fig:Case3Fig4}
\end{subfigure}
\hfill
\hspace*{-0.04\linewidth}\begin{subfigure}[b]{0.36\linewidth}
  \includegraphics{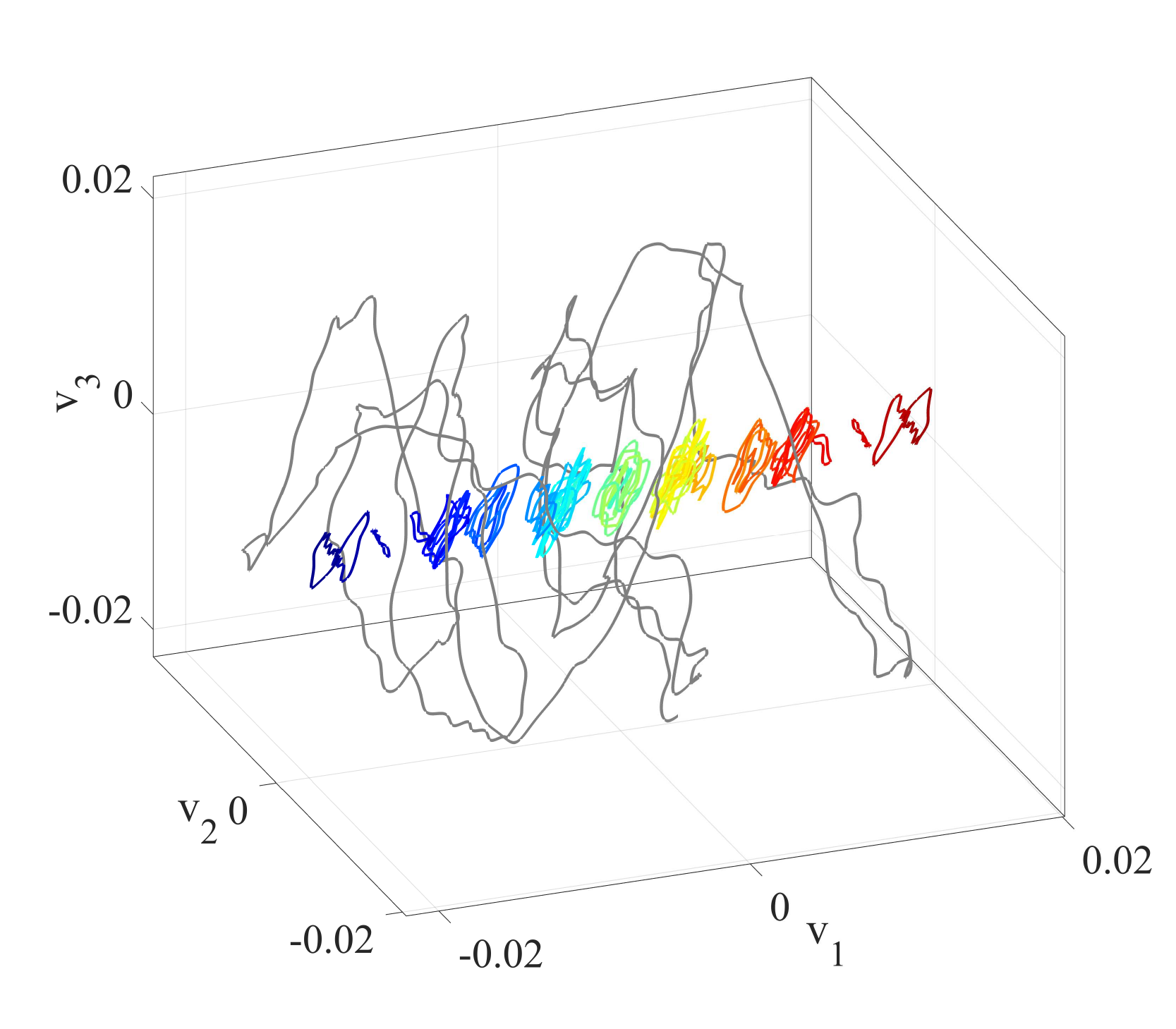}
  \caption{$t_{height}=15$}
  \label{fig:Case3Fig5}
\end{subfigure}
    \end{minipage}
\hfill
\hspace{-0.03\linewidth}\begin{minipage}{0.1\textwidth}
\includegraphics{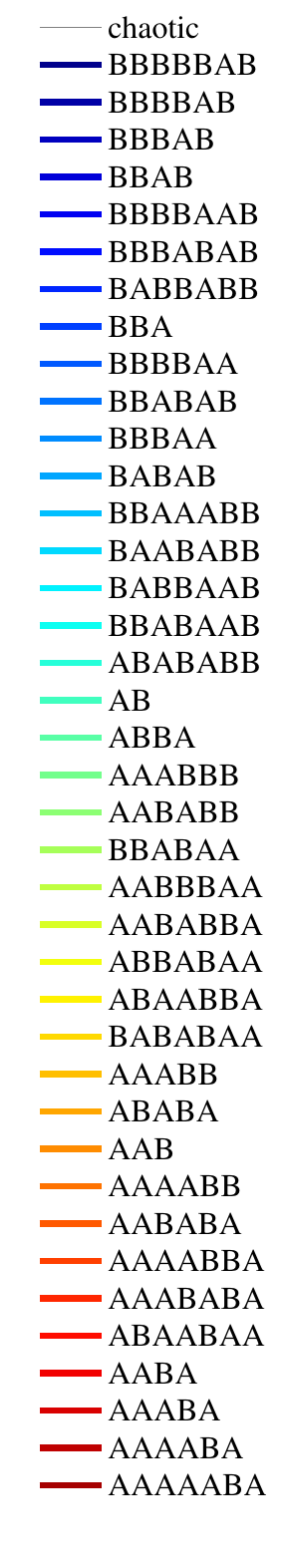}
\end{minipage}
\caption{Lorenz attractor: Unfolding of the attractor for the symbolic sequences length less than 8.}
\label{fig:SeqLenLTDEUPOs}
\end{figure}

%% file: RosslerAttractor.tex
\subsection{Rössler attractor} 
The initial conditions for the chaotic trajectory are taken to be $[x(t_0),y(t_0),z(t_0)]=[-2.8585, 0.2234,$ $ 0.2979]$. The values of the trajectory are computed for $\Delta t = 0.005$. In the Lorenz case, the $\Delta t$ for computation is taken to be $0.01$ whereas, in case of the Rössler attractor the $\Delta t$ is taken to be $0.005$. The Hankel matrix in this case, is constructed from the partial measurements of the $y$ variable. Additionally, the $t_{max}$(maximum time-period of UPO in the dataset) in case of Lorenz and Rössler attractor are $9.15097$ and $44.695$ respectively. The reduced time-step magnitude and increased time-period of the periodic orbits in the Rössler case make the SVD of the Hankel matrix memory-intensive. To mitigate this, variable values of $\tau$ parameter are used when computing the Hankel matrix, preventing memory issues during the SVD computation.

\subsubsection{Case I: Unstable periodic orbits of the type \texorpdfstring{$\mathbf{A^nB}$ and $\mathbf{AB^n}$}{AnB and ABn}}

The data for this case is taken from \cite{letellier1994caracterisation} \textit{Christophe dataset} which contains 41 unstable periodic orbits computed for sequence length less than 8. The dataset has been divided into smaller subsets to further understand the separation of the UPOs. However, as this dataset is not as extensive as the \textit{Divakar dataset}, the number of UPOs for each case is very limited. In case I, the unstable periodic orbits of type $A^nB$ and $AB^n$ is considered for $n=1,2,3,\cdots,6$. The results of the long time delay embeddings are obtained for varying the $t_{height}$ of the Hankel matrices. The $t_{width}=1200$ is kept constant for all the embeddings. However, different values are taken for $\tau$ to reduce the computational time. The Hankel matrices are computed for $t_{height} = [0.5,5,50,250,500,1000]$ and it's corresponding $\tau$ values are $\tau=[0.05,0.05,0.05,0.05,0.1,0.25]$. The results of the embeddings are compiled in Fig.(\ref{fig:RossC1}a-f). 

Similar to the Lorenz case, the $A^nB$ and $AB^n$ sets of orbits cluster on the diametric opposite ends of the unfolded chaotic attractor. However, unlike the Lorenz case the UPO do not show clear separation based on the ratio of the symbols (i.e., A and B) from the symbolic dynamics definition of the UPO. While the division of the phase space into $A$ and $B$ regions is done at $y=-3.04$ ($B$ is taken to be $y>-3.04$), we observe that in Fig.(\ref{fig:RosslerFRM}c) a considerable section of $B$ extends along $y<-3.04$. As as result the corresponding separation of trajectories does not exhibit the same behavior as the Lorenz attractor.

\begin{figure}
\setkeys{Gin}{width=\linewidth}
    \begin{minipage}{0.925\textwidth}
\hspace*{-0.02\linewidth}\begin{subfigure}[b]{0.36\linewidth}
  \includegraphics{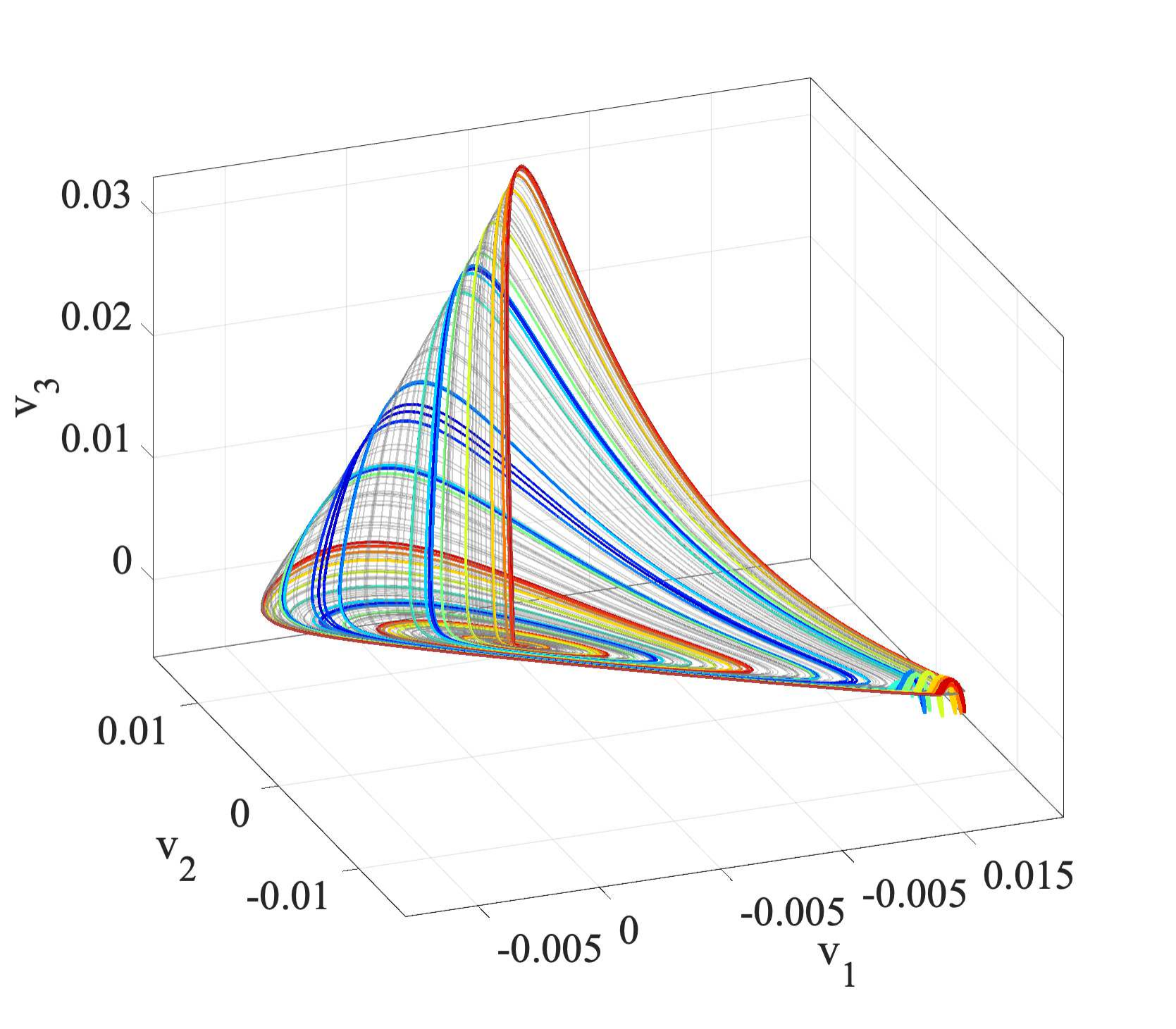}
  \caption{$t_{height}=0.5, \tau=0.05$}
  \label{fig:RCase1Fig1}
\end{subfigure}
\hfill
\hspace*{-0.04\linewidth}\begin{subfigure}[b]{0.36\linewidth}
  \includegraphics{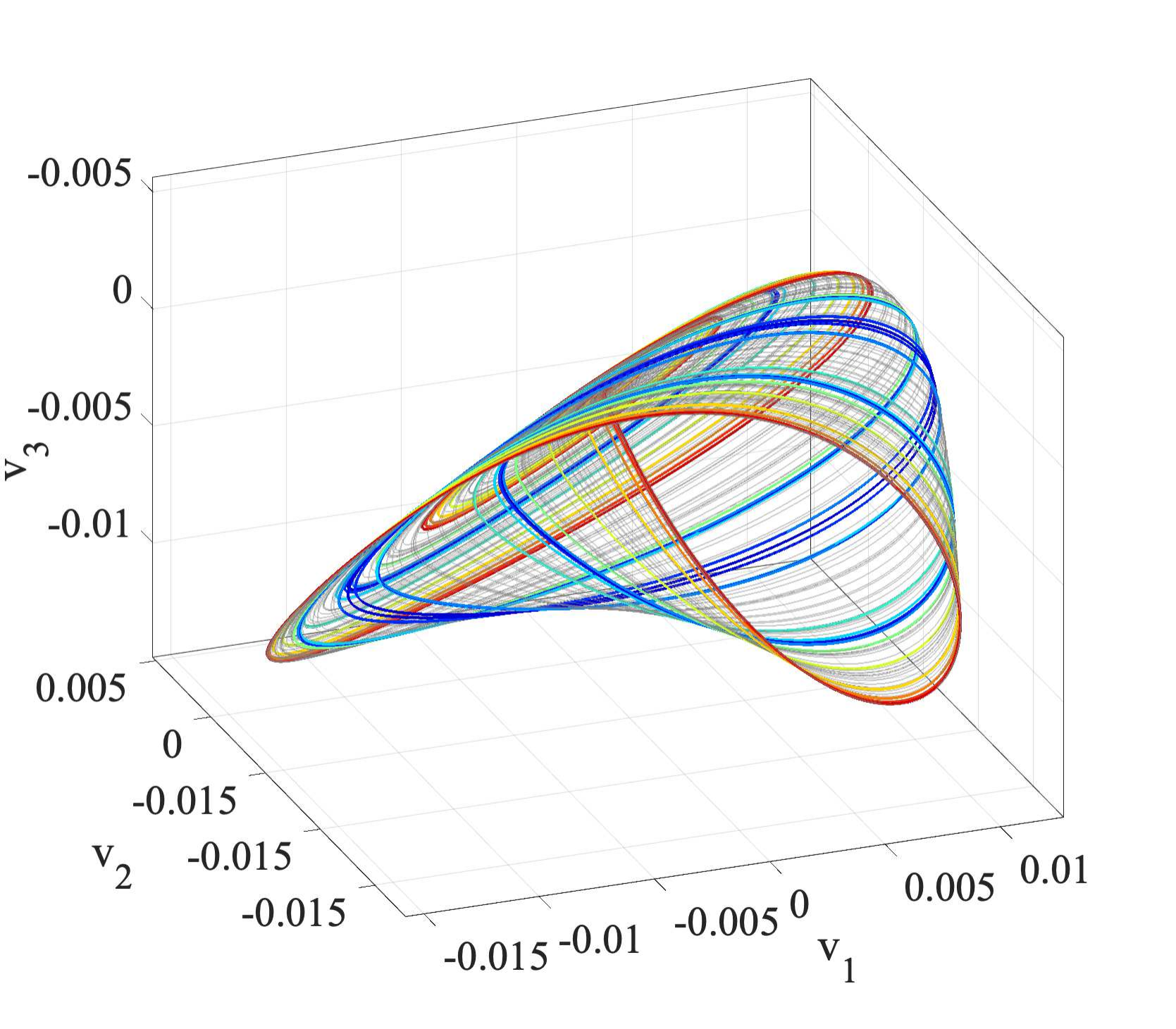}
  \caption{$t_{height}=5,\tau=0.05$}
  \label{fig:RCase1Fig2a}
\end{subfigure}
\hfill
\hspace*{-0.04\linewidth}\begin{subfigure}[b]{0.36\linewidth}
  \includegraphics{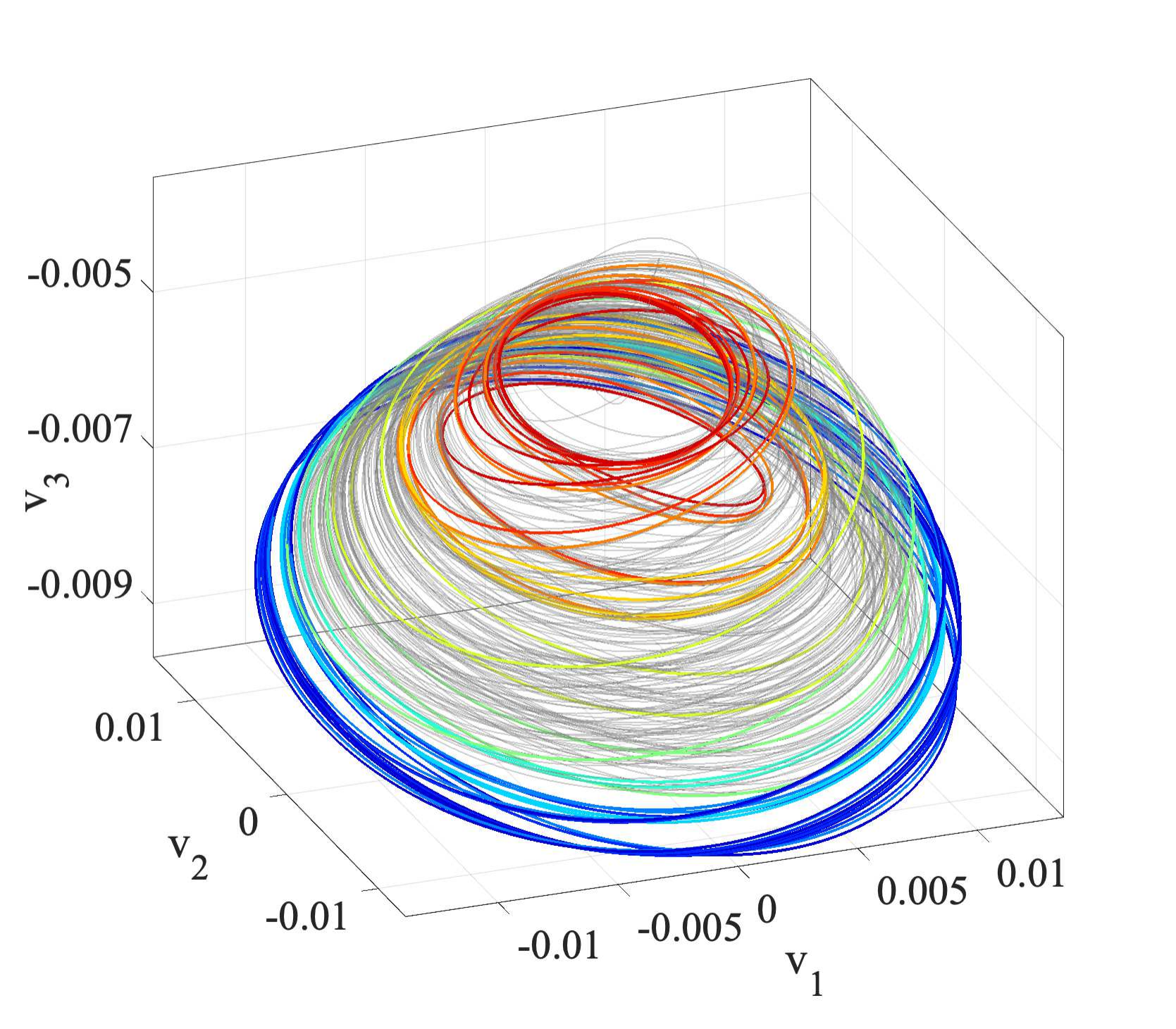}
  \caption{$t_{height}=50,\tau=0.05$}
  \label{fig:RCase1Fig3a}
\end{subfigure}
\hfill
\hspace*{-0.02\linewidth}\begin{subfigure}[b]{0.36\linewidth}
  \includegraphics{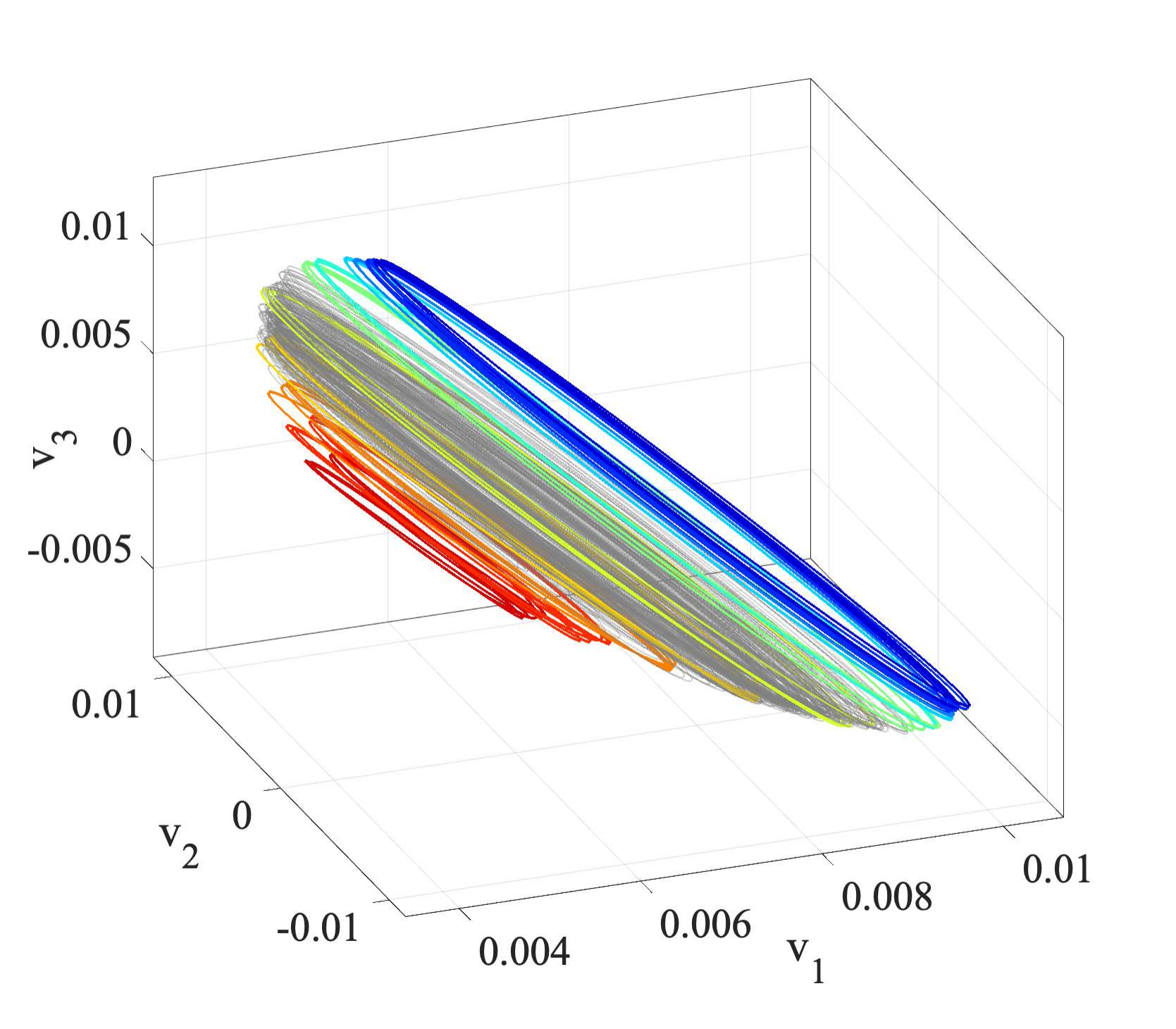}
  \caption{$t_{height}=250,\tau=0.05$}
  \label{fig:RCase1Fig4a}
\end{subfigure}
\hfill
\hspace*{-0.04\linewidth}\begin{subfigure}[b]{0.36\linewidth}
  \includegraphics{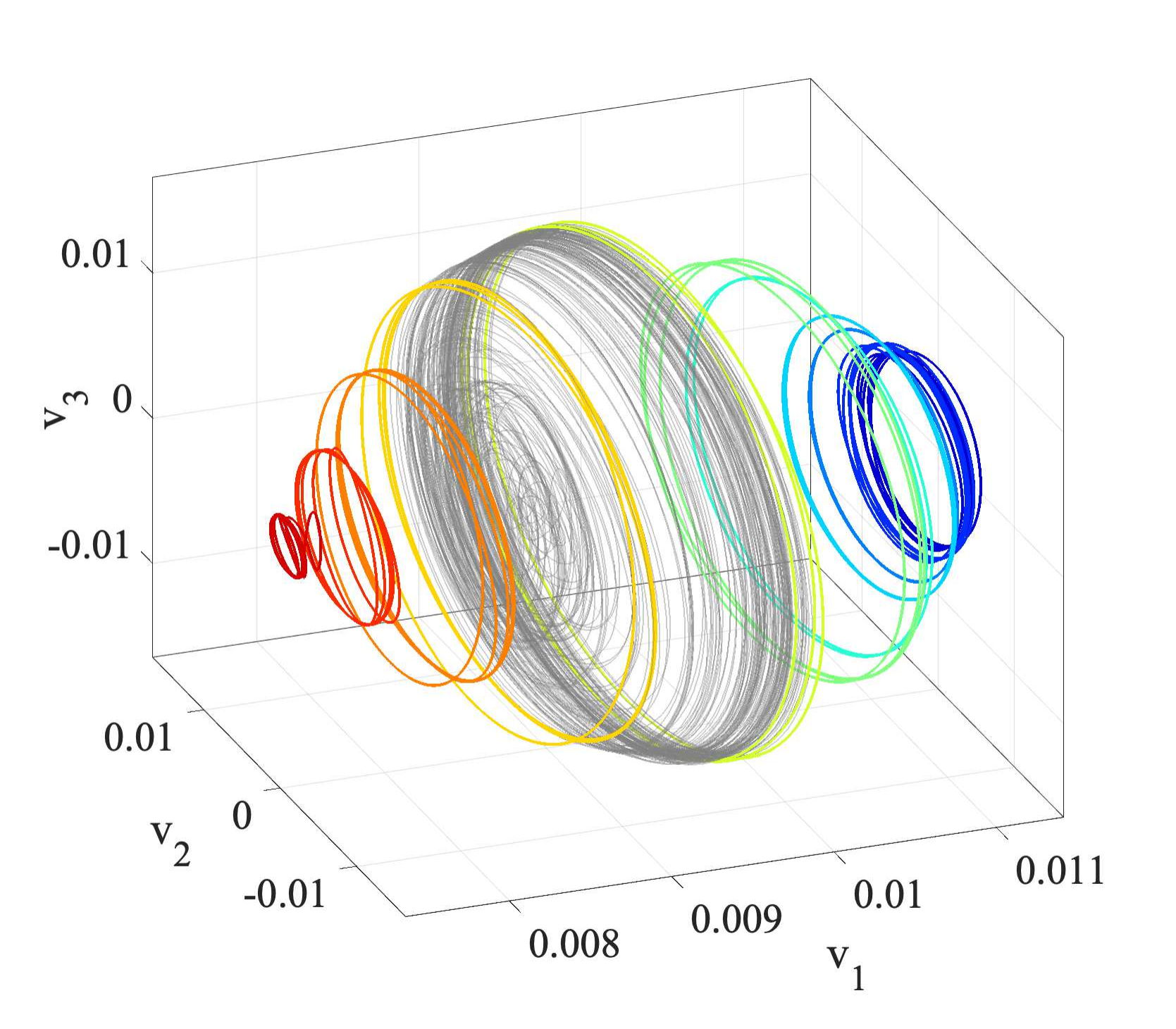}
  \caption{$t_{height}=500,\tau=0.1$}
  \label{fig:RCase1Fig5a}
\end{subfigure}
\hfill
\hspace*{-0.04\linewidth}\begin{subfigure}[b]{0.36\linewidth}
  \includegraphics{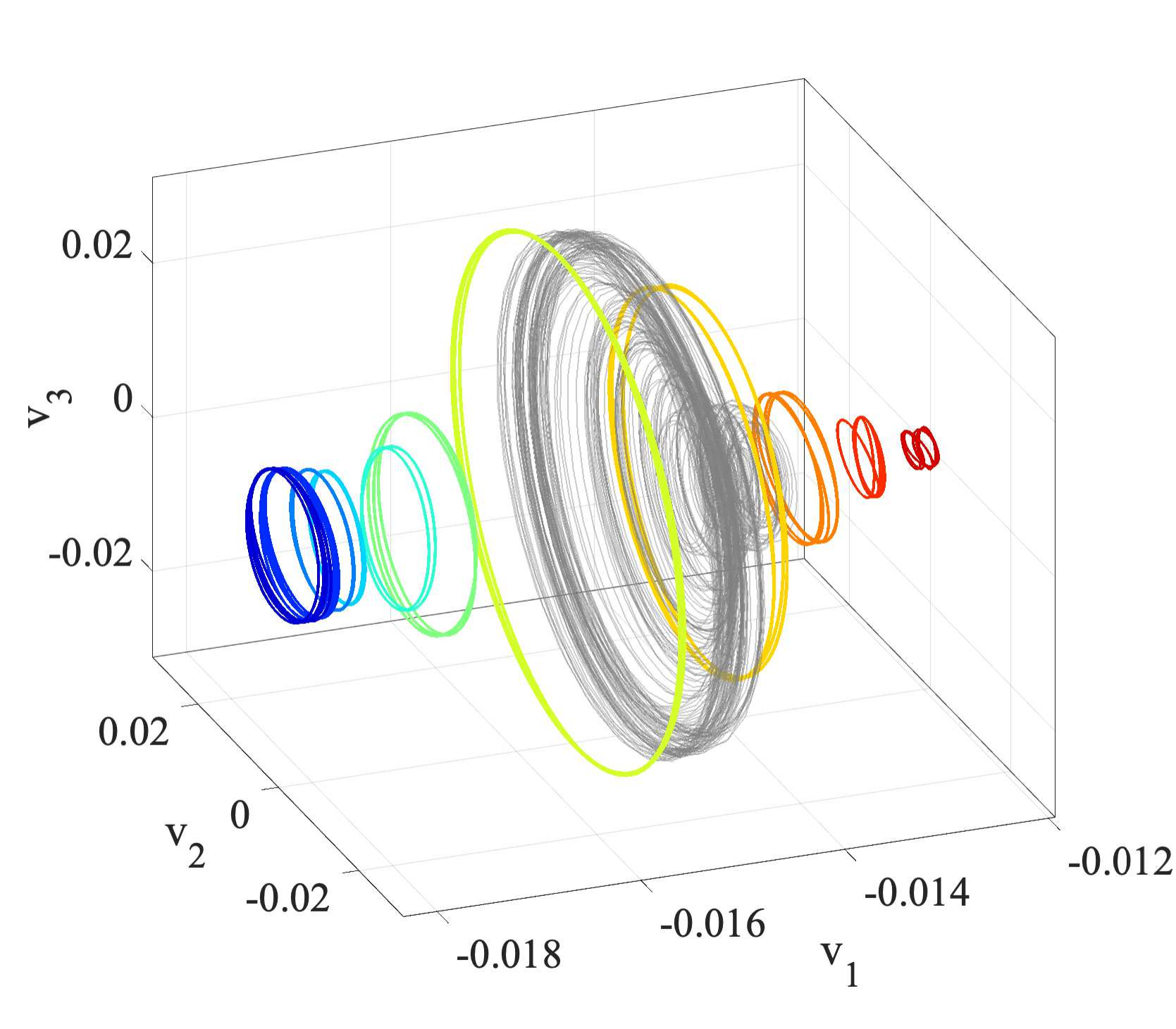}
  \caption{$t_{height}=1000,\tau=0.25$}
  \label{fig:RCase1Fig6a}
\end{subfigure}
\end{minipage}
\hfill
\hspace{-0.03\linewidth}\begin{minipage}{0.1\textwidth}
\includegraphics{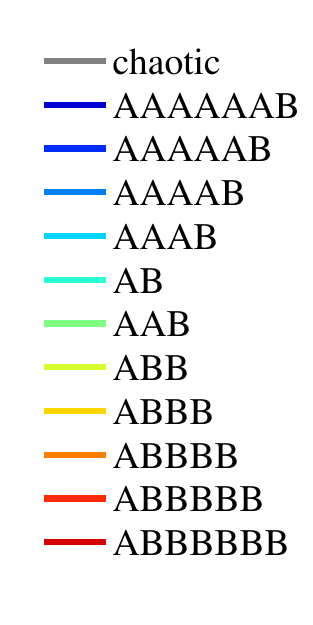}
\end{minipage}
\caption{Rössler attractor: Unfolding of the attractor for UPOs of type $A^nB$ and $AB^n$. While we still observe the clustering of the UPOs on the diametrically opposite ends of the unfolded attractor (similar to the Lorenz case), the arrangement of the UPOs depends on the ratio of time spent by the UPOs in $y>0$ and $y<0$.}
\label{fig:RossC1}
\end{figure}

\subsection{Case II: Unstable periodic orbits of sequence length less than 8}
In this case, we consider all the UPOs with sequence length less than 8 and the results for long time delay embeddings are obtained by varying the $t_{height}=[0.5, 5, 50,250,500,1000]$ with corresponding $\tau$ values:$\tau=[0.05,0.05,0.05,0.1,0.25]$. In Fig.(\ref{fig:RossC2}a-f), we observe the separation of UPOs into clusters wherein, the cluster groupings are based on the ratio $\rho$ which is ratio of time spent by the UPO in the $y>0$ region to $y<0$ region. The UPOs are then sorted and color-coded according to the ratio $\frac{\rho-1}{\rho+1}$ and they arrange themselves monotonically according to this ratio. Note that the color-coding in this case differs from the symbolic dynamics (A-heavy/B-heavy) but is rather recomputed for the ratio based on $y>0$ and $y<0$.

Thus, unlike the Lorenz case, the separation of UPOs in the embedded space for the Rössler attractor is not related to the symbolic dynamics separation plane but is instead determined by their distance from $y=0$. For the Lorenz attractor, the separation occurs along the plane $x=0$, allowing analysis in terms of symbolic dynamics, which does not apply to the Rössler attractor.

\begin{figure}
\setkeys{Gin}{width=\linewidth}
    \begin{minipage}{0.9\textwidth}
\hspace*{-0.02\linewidth}\begin{subfigure}[b]{0.36\linewidth}
  \includegraphics{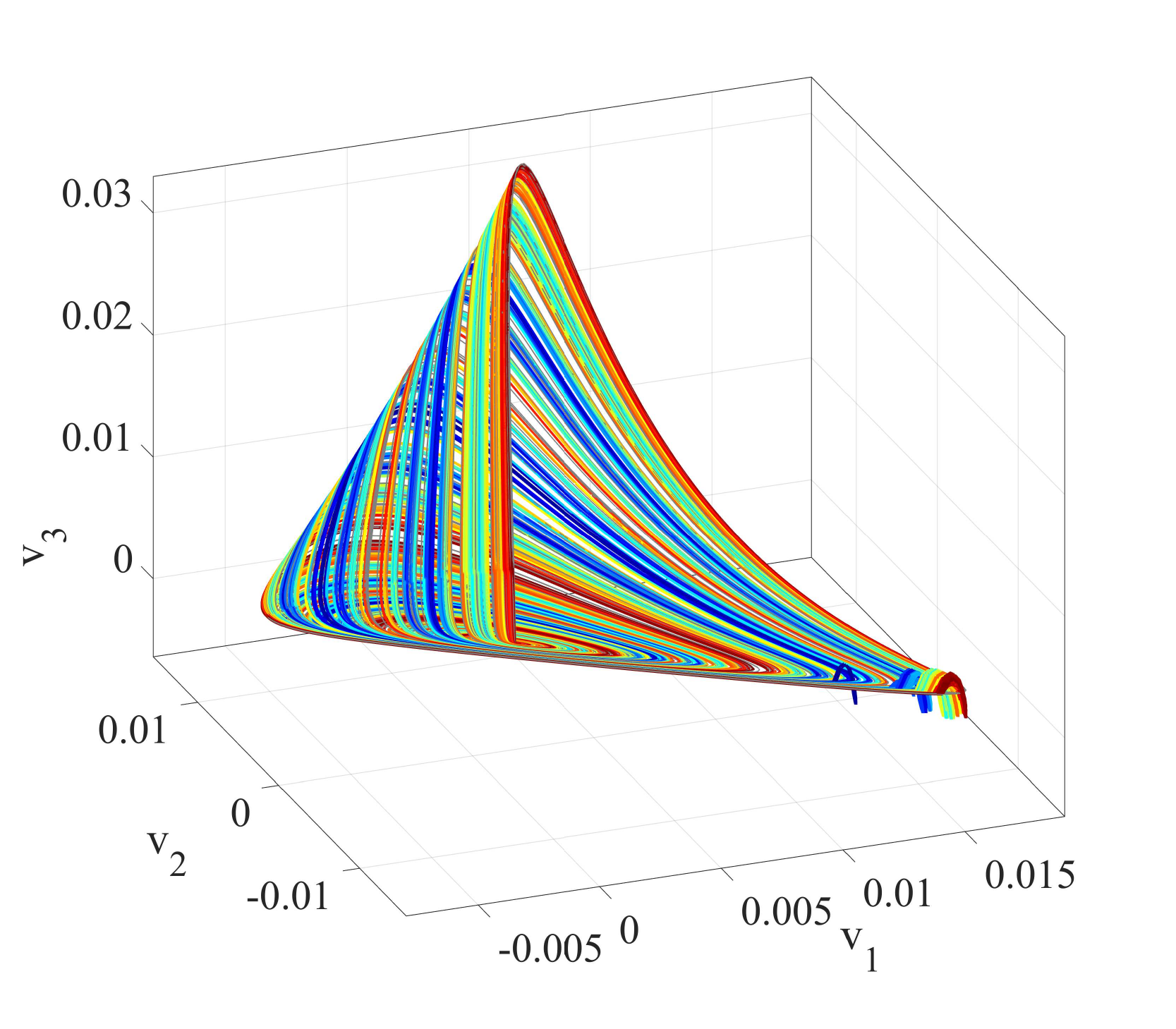}
  \caption{$t_{height}=0.5, \tau=0.05$}
  \label{fig:RCase1Fig1a}
\end{subfigure}
\hfill
\hspace*{-0.04\linewidth}\begin{subfigure}[b]{0.36\linewidth}
  \includegraphics{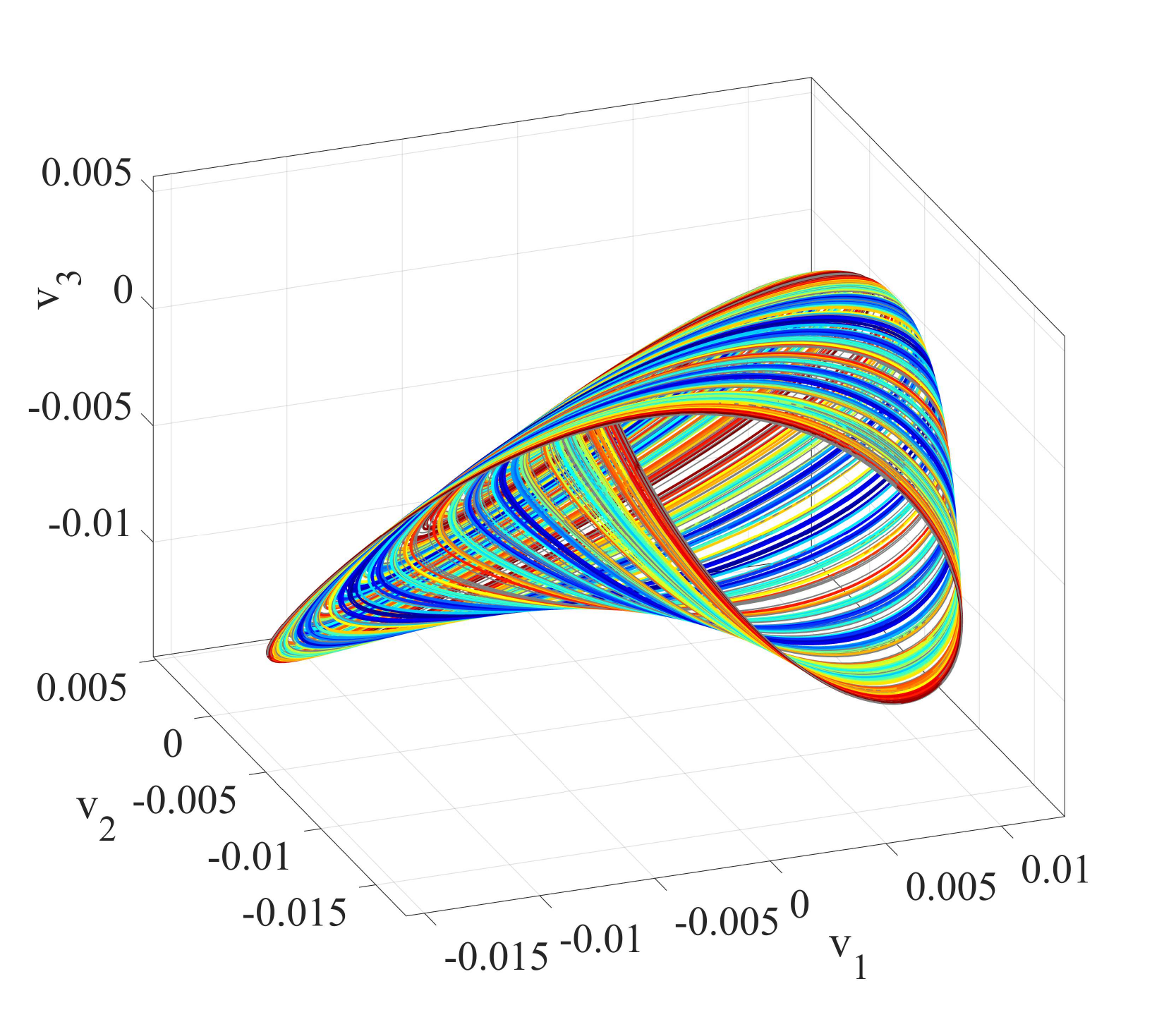}
  \caption{$t_{height}=5,\tau=0.05$}
  \label{fig:RCase1Fig2}
\end{subfigure}
\hfill
\hspace*{-0.04\linewidth}\begin{subfigure}[b]{0.36\linewidth}
  \includegraphics{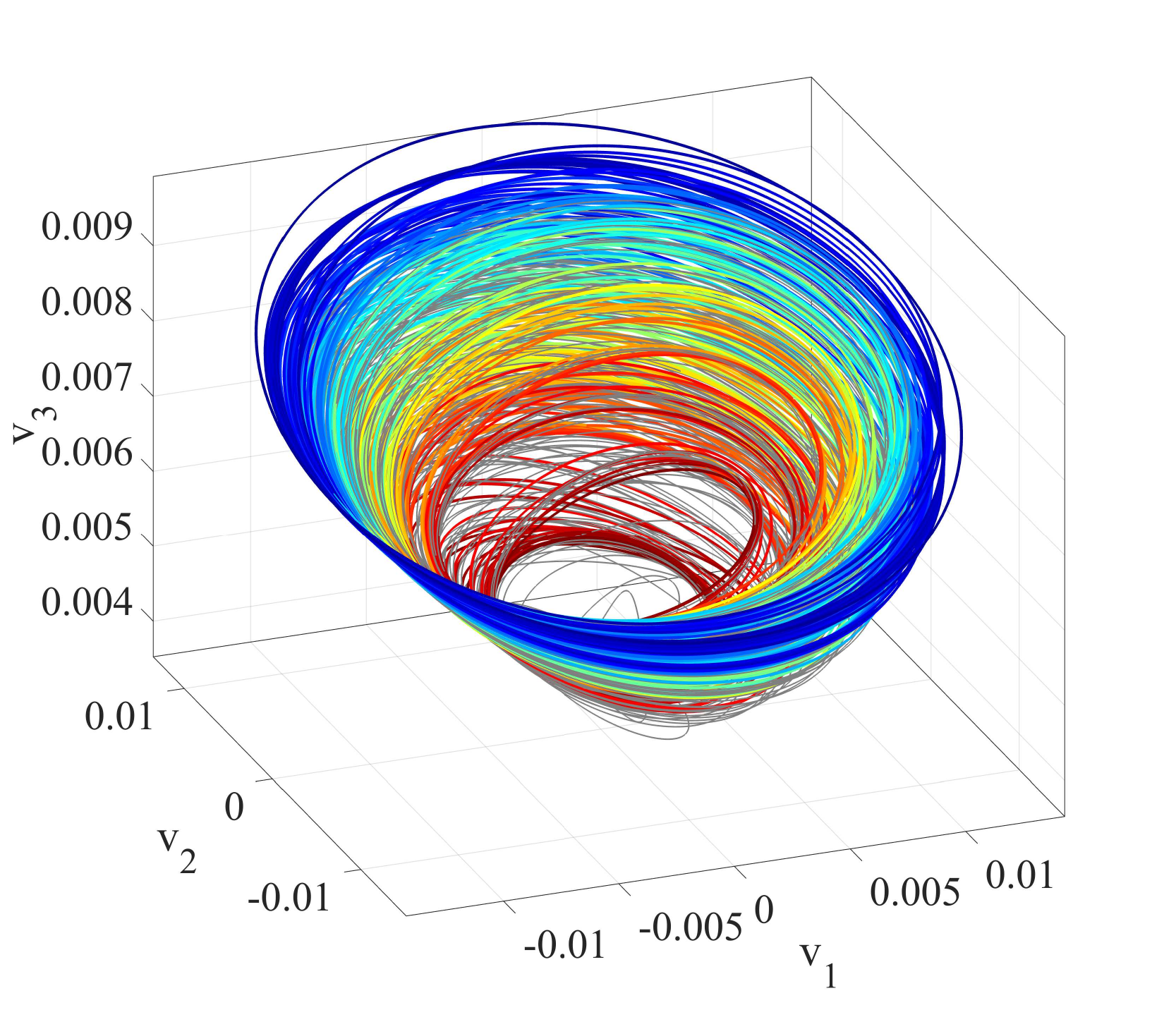}
  \caption{$t_{height}=50,\tau=0.05$}
  \label{fig:RCase1Fig3}
\end{subfigure}
\hfill
\hspace*{-0.02\linewidth}\begin{subfigure}[b]{0.36\linewidth}
  \includegraphics{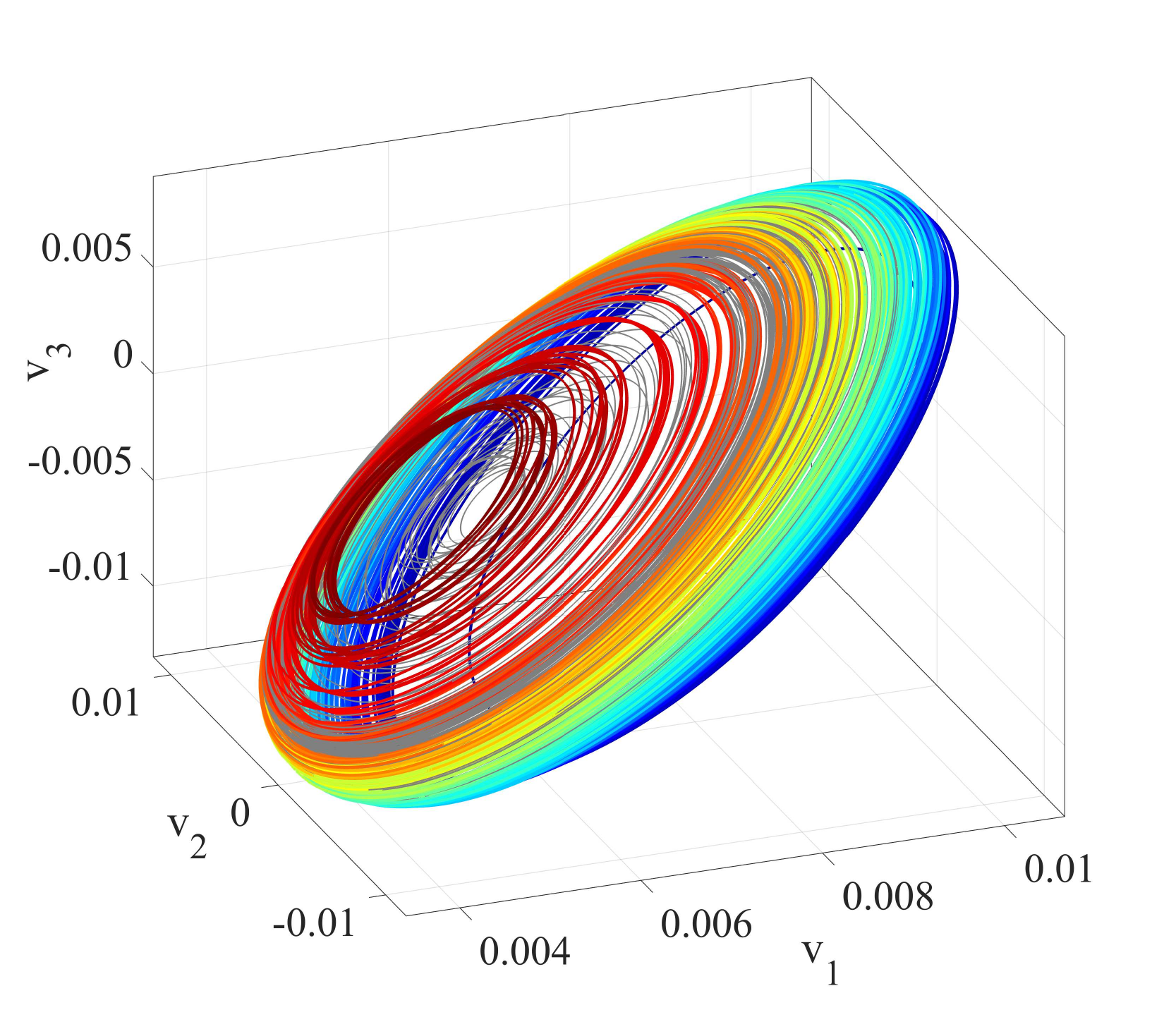}
  \caption{$t_{height}=250,\tau=0.05$}
  \label{fig:RCase1Fig4}
\end{subfigure}
\hfill
\hspace*{-0.04\linewidth}\begin{subfigure}[b]{0.36\linewidth}
  \includegraphics{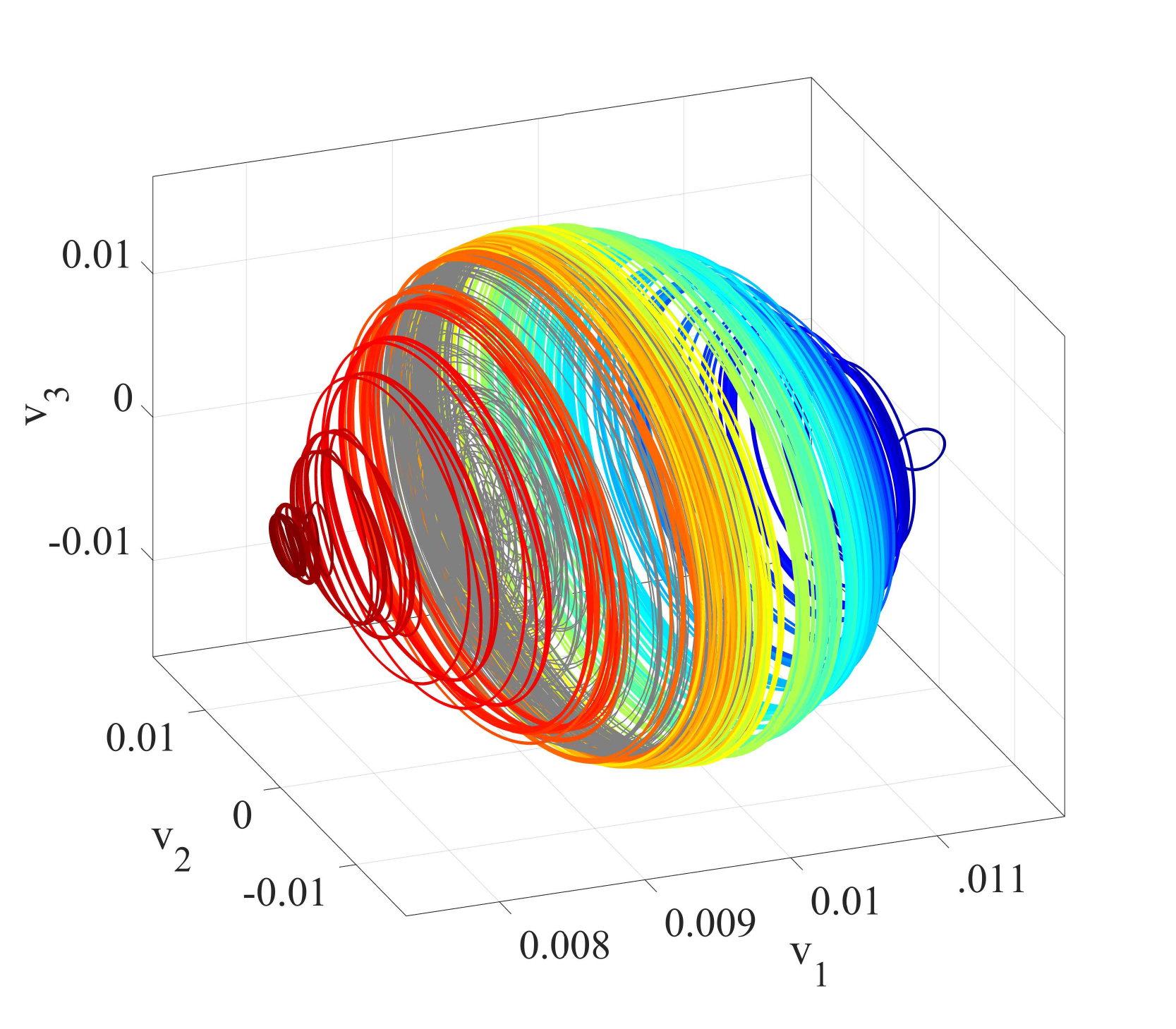}
  \caption{$t_{height}=500,\tau=0.1$}
  \label{fig:RCase1Fig5}
\end{subfigure}
\hfill
\hspace*{-0.04\linewidth}\begin{subfigure}[b]{0.36\linewidth}
  \includegraphics{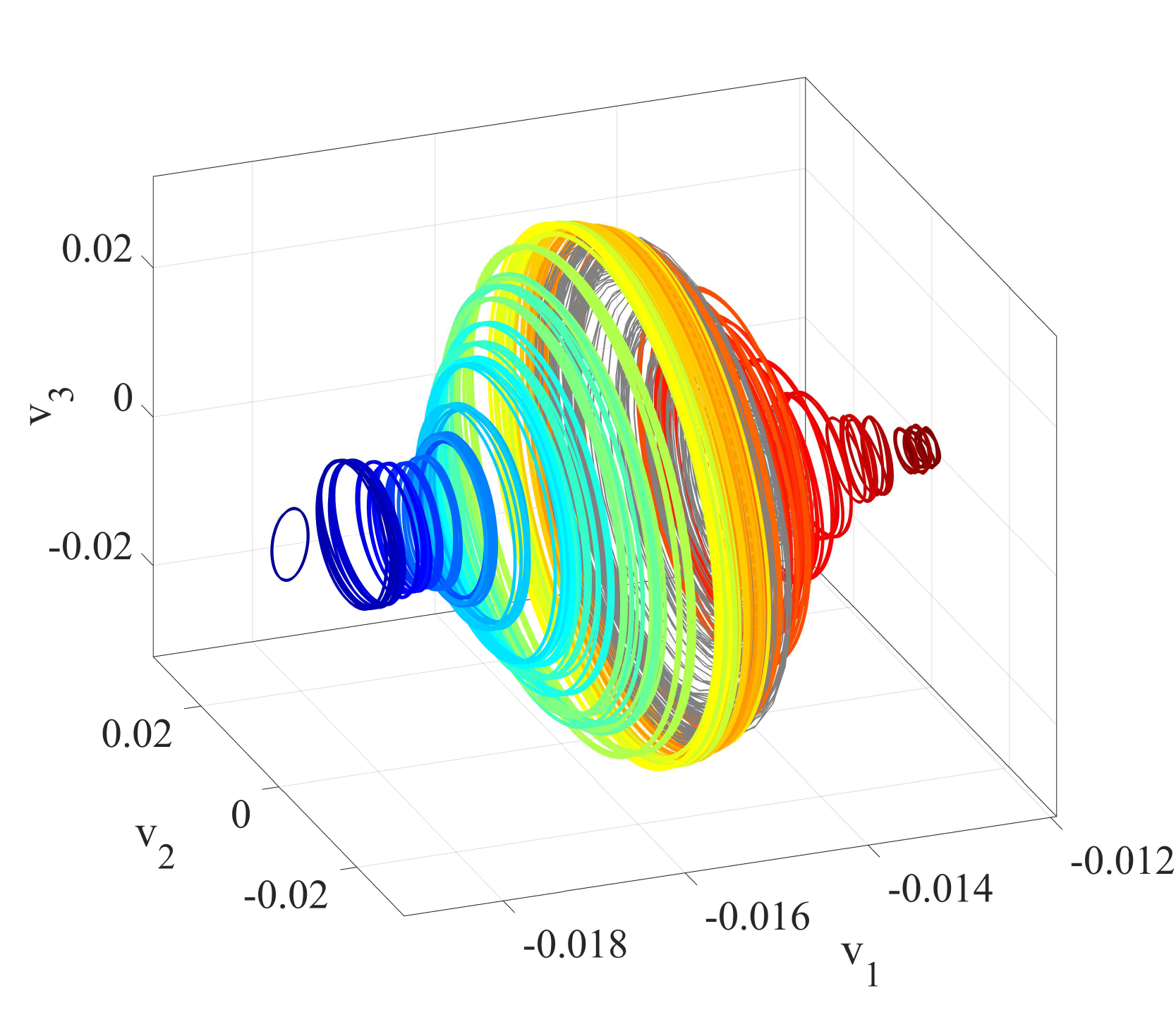}
  \caption{$t_{height}=1000,\tau=0.25$}
  \label{fig:RCase1Fig6}
\end{subfigure}
\end{minipage}
\hfill
\hspace{-0.03\linewidth}\begin{minipage}{0.1\textwidth}
\includegraphics{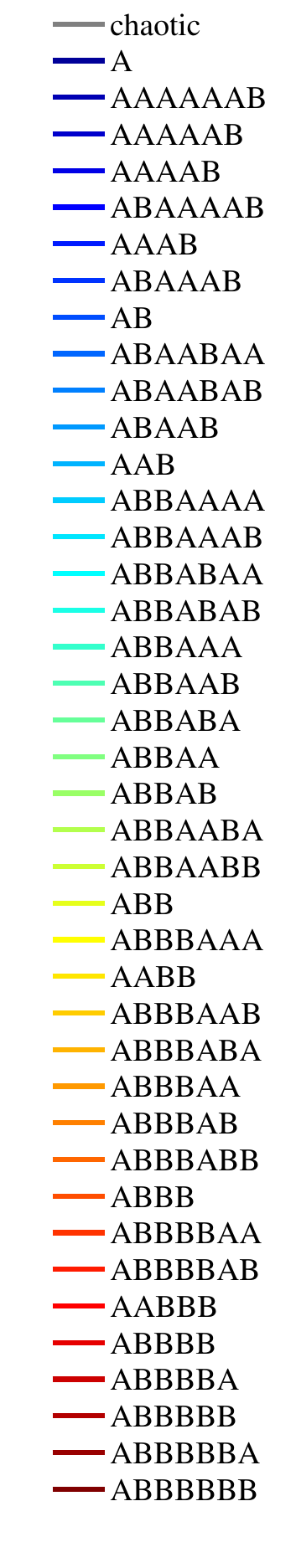}
\end{minipage}
\caption{Rössler attractor: Unfolding of the attractor for UPOs of sequence length less than 8. However, in this case the UPOs are rearranged according the ratio of time spent in $y<0$ and $y>0$. It is observed that the UPOs monotonically arrange themselves based on this ratio.}
\label{fig:RossC2}
\end{figure}